\documentclass[journal,draftcls,onecolumn,12pt,twoside]{IEEEtran}
\usepackage[utf8]{inputenc} 
\usepackage[T1]{fontenc}
\usepackage{url}
\usepackage{ifthen}
\usepackage{cite}
\usepackage{amsmath}%
\usepackage{wasysym}%
\usepackage{amssymb}

\usepackage{mdwmath}
\usepackage{blindtext}
\usepackage{eqparbox}
\usepackage{fixltx2e}
\usepackage{stfloats}

\usepackage[english]{babel}
\usepackage{lipsum}
\usepackage{tikz}
\usepackage{mathtools,amsfonts,amssymb,amsthm}
\usepackage{graphicx}
\usepackage{algorithm}
\usepackage{algpseudocode} 

\usepackage{geometry}
\newtheorem{theorem}{{Theorem}}

\newtheorem{lemma}[theorem]{{Lemma}}
\newtheorem{proposition}[theorem]{{Proposition}}
\newtheorem{corollary}[theorem]{{Corollary}}

\newtheorem{definition}{{Definition}}
\newtheorem{Result}{{Result}}

\theoremstyle{remark}
\newtheorem{remark}{{Remark}}

\newcommand{\cA}{{\cal A}} 
\newcommand{\cB}{{\cal B}}
\newcommand{\cC}{{\cal C}}

\newcommand{\cF}{{\cal F}}

\newcommand{\cR}{{\cal R}}
\newcommand{\cS}{{\cal S}} 
\newcommand{\cT}{{\cal T}}

\DeclareMathAlphabet{\mathbfsl}{OT1}{ppl}{b}{it} %{OT1}{cmr}{bx}{it}

\newcommand{\bQ}{\mathbfsl{Q}} 
\newcommand{\bR}{\mathbfsl{R}}

\newcommand{\bU}{\mathbfsl{U}} 
\newcommand{\bV}{\mathbfsl{V}}
 
\newcommand{\bX}{\mathbfsl{X}}
 
\newcommand{\bZ}{\mathbfsl{Z}}

\newcommand{\bu}{\mathbfsl{u}} 
\newcommand{\bv}{\mathbfsl{v}}
 
\newcommand{\bx}{\mathbfsl{x}}

\newcommand{\mathset}[1]{\left\{#1\right\}}
\newcommand{\abs}[1]{\left|#1\right|}

\newcommand{\floor}[1]{\left\lfloor #1 \right\rfloor}

\def\QEDclosed{\mbox{\rule[0pt]{1.3ex}{1.3ex}}} % for a filled box
\def\QED{\QEDclosed} % default to closed

\def\proof{\noindent\hspace{2em}{\itshape Proof: }}
\def\endproof{\hspace*{\fill}~\QED\par\endtrivlist\unskip}

\newcommand{\be}[1]{\begin{equation}\label{#1}}
\newcommand{\ee}{\end{equation}}

\renewcommand{\leq}{\leqslant}
 
\renewcommand{\geq}{\geqslant}

\newcommand{\script}[1]{{\mathscr #1}}

\renewcommand{\Bbb}{\mathbb}
 
\newcommand{\N}{{\Bbb N}}
\newcommand{\R}{{\Bbb R}}

\newcommand{\Cref}[1]{Co\-ro\-lla\-ry\,\ref{#1}}

\newcommand{\Ftwo}{{{\Bbb F}}_{\!2}}

\newcommand{\zero}{{\mathbf 0}}

\usepackage{stmaryrd} %for \varoast
\usepackage{verbatim} %for large section comment

\usepackage{mathdots}
\usepackage{nicefrac}
\newcommand{\bit}{\ensuremath{\mathset{0,1}}}
\usepackage{float}
\usepackage[justification=centering]{caption} % The captions will be centered
\usepackage{subcaption} %for sub-figures
\newcommand{\RNum}[1]{\uppercase\expandafter{\romannumeral #1\relax}}
\usepackage{enumerate}
\usepackage{enumitem} % for enumerate with different labels
\newcommand{\goto}{{\ \rightarrow\ }}

\usepackage{hyperref}
\hypersetup{
    colorlinks=true,
    linkcolor=black,
    urlcolor=red,
    linktoc=all,
    citecolor=black
}
\usepackage{adjustbox}

\IEEEoverridecommandlockouts
\begin{document}

\title{Capacity-achieving Polar-based LDGM Codes}

\author{
James Chin-Jen Pang, Hessam Mahdavifar, and S. Sandeep Pradhan  
% Michael ̃Shell, ̃\IEEEmembership{Member, ̃IEEE, } John ̃Doe, ̃\IEEEmembership{Fellow, ̃OSA,} and ̃Jane ̃D oe, ̃\IEEEmembership{Life ̃Fellow, ̃IEEE} 
\thanks{This work was supported by the National Science Foundation under grants CCF--1717299, CCF--1763348, and CCF--1909771.
This paper was presented in part at the 2020 IEEE International Symposium on Information Theory.
}%
\thanks{The authors are with the Department of Electrical Engineering and Computer Science, University of Michigan, Ann Arbor, MI 48109, USA (e-mail: cjpang@umich.edu; hessam@umich.edu; pradhanv@umich.edu).}
}

\maketitle

\begin{abstract}
In this paper, we study codes with sparse generator matrices. More specifically, low-density generator matrix (LDGM) codes with a certain constraint on the weight of all the columns in the generator matrix are considered. 
Previous works have shown the existence of capacity-achieving LDGM codes with column weight bounded above by any linear function of the block length $N$ over the BSC, and by the logarithm of $N$ over the BEC.
In this paper, it is first shown that when a binary-input memoryless symmetric (BMS) channel $W$ and a constant $s>0$ are given, there exists a polarization kernel such that the corresponding polar code is capacity-achieving and the column weights of the generator matrix (GM) are bounded from above by $N^s$. 

Then, a general construction based on a concatenation of polar codes and a rate-$1$ code, and a new column-splitting algorithm that guarantees a much sparser GM, is given. More specifically, for any BMS channel and any $\epsilon > 2\epsilon^*$, where $\epsilon^* \approx 0.085$, an existence of a sequence of capacity-achieving codes with all the GM column weights upper bounded by $(\log N)^{1+\epsilon}$ is shown.
Furthermore, two coding schemes for BEC and BMS channels, based on a second column-splitting algorithm, are devised with low-complexity decoding that uses successive-cancellation. 
The  second splitting algorithm allows for the use of a low-complexity decoder  by preserving the reliability of the bit-channels observed by the source bits, and by increasing the code block length.
In particular, for any BEC and any $\lambda >\lambda^* = 0.5+\epsilon^*$, the existence of a sequence of capacity-achieving  codes where all the GM column weights are bounded from above by $(\log N)^{2\lambda}$ and  with decoding complexity $O(N\log \log N)$ is shown. 
The existence of similar capacity-achieving LDGM codes with low-complexity decoding is shown for any BMS channel, and for  any $\lambda >\lambda^{\dagger} \approx 0.631$.
% \blue{
The concatenation-based construction can also be applied to the random linear code ensemble to yield capacity-achieving codes with all the GM column weights being $O(\log N)$ and with (large-degree) polynomial decoding complexity.

\end{abstract}
\tableofcontents
\newpage

\section{Introduction}

Capacity-approaching error-correcting codes such as low-density parity-check (LDPC) codes \cite{gallager1962low} and polar codes \cite{arikan2009channel} have been extensively studied for applications in wireless and storage systems. Besides conventional applications of codes for error correction, a surge of new applications has also emerged in the past decade including crowdsourcing \cite{karger2011iterative,vempaty2014reliable}, distributed storage \cite{dimakis2010network}, and speeding up distributed machine learning \cite{lee2018speeding, SGDviaLDGM2019}. To this end, new motivations have arisen to study codes with sparsity constraints on their generator and/or parity-check matrices.
For instance, the stored data in a failed server needs to be recovered by downloading data from a few servers only, due to bandwidth constraints, imposing sparsity constraints in the decoding process in a distributed storage system. In crowdsourcing applications, e.g., when workers are asked to label items in a dataset, each worker can be assigned only a few items due to capability limitations, imposing sparsity constraints in the encoding process. More specifically, low-density generator matrix (LDGM) codes become relevant for such applications \cite{mazumdar2017semisupervised,pang2019coding}. 

\subsection{LDGM and Related Works}

LDGM codes, often regarded as the dual of LDPC codes, are associated with sparse factor graphs.
The sparsity of the generator matrices of LDGM codes leads to a low encoding complexity. However, unlike LDPC and polar codes, LDGM codes have not received significant attention. In \cite{mackay1999good, mackay1995good} it was pointed out that certain constructions of LDGM codes are not asymptotically \textit{good}, a behavior which is also studied using an error floor analysis in \cite{zhong2005approaching, garcia2003approaching}. 
Several prior works, e.g., \cite{zhong2005approaching, garcia2003approaching, zhong2005ldgm}, adopt concatenation of two LDGM codes to construct systematic capacity-approaching LDGM codes with significantly lower error floors in simulations. As a sub-class of LDPC codes, the systematic LDGM codes are advantageous for their low encoding and decoding complexity. 

In terms of the sparsity of the GM, the authors of \cite{LDGM_capAchieving2011} showed the existence of capacity-achieving codes over binary symmetric channels (BSC) using random linear coding arguments when the column weights of the GM are bounded by $\epsilon N$, for any $\epsilon > 0$, where $N$ is the code block length. Also, it is conjectured in \cite{LDGM_capAchieving2011} that column weight upper bounds that scale polynomially sublinear in $N$ suffice to achieve the capacity. For binary erasure channels (BEC), bounds that scale as $O(\log N)$ suffice for achieving the capacity, again using random linear coding arguments \cite{LDGM_capAchieving2011}. Furthermore, the scaling exponent of such random linear codes are studied in \cite{mahdavifar2017scaling}.
Later, in \cite{lin2018coding}, the existence of capacity-achieving systematic LDGM ensembles over any BMS channel with the expected value of the weight of the entire GM bounded by $\epsilon N^2$, for any $\epsilon > 0$, is shown.

In \cite{mazumdar2017semisupervised, pang2019coding},  the problem of label learning through queries from  a crowd of workers was formulated as a coding theory problem. Due to practical constraints in such crowdsourcing scenarios, each query can only contain a small number of items. 
When some workers do not respond, resembling a binary erasure channel, the authors showed that a combination of LDPC codes and LDGM codes gives a query scheme where the number of queries approaches the information-theoretic lower bound \cite{pang2019coding}. 

% \blue{
In the realm of quantum error correction, quantum low-density-generator-matrix (QLDGM) codes, quantum low-density-parity-check (QLDPC) codes, and other sparse-graph schemes have been extensively studied due to  the small numbers of quantum interactions per qubit during the encoding and/or error correction procedure, avoiding additional quantum gate errors and facilitating fault-tolerant decoding.
Amongst these schemes, the error correction performance of the LDGM-based codes proposed in \cite{PhysRevA.102.012423} was shown to outperform all other Calderbank-Steane-Shor (CSS) and non-CSS codes of similar complexity. 
In \cite{PhysRevA.103.022617}, Fuentes et. al. showed how codes tailored to the symmetric Pauli channel are not well suited to asymmetric quantum
channels, and derived quantum CSS LDGM codes that perform well for the latter case.
% }

% \blue{
In the realm of machine learning, gradient-based methods, such as the gradient descent (GD) algorithm, are one of the most commonly used algorithms to fit the machine learning models over the training data. 
Using LDGM codes, authors of \cite{horii2019distributed} proposed a distributed implementation of a stochastic GD scheme, which allowed the master node to recover a
``high-quality unbiased'' estimate of the gradient at low computational cost and provided overall performance improvement over the GD scheme with gradient coding. 
% }

% \blue{
In all the three applications highlighted above, the benefit of the LDGM codes follows from a certain upper bound on the column weights of the GM ensuring the columns are relatively low weight. 
Motivated by these applications, 
the main goal of this work is to construct sequences of LDGM codes where all or the vast majority of the columns of the GM are \textit{low weight}, where certain upper bounds on the weight will be specified later.
% }

\subsection{Our Contributions}
In this paper, we focus on studying capacity-achieving LDGM codes over BMS channels with sparsity constraints on column weights. 
A type of code is said to be capacity achieving over a BMS channel $W$ with capacity $C= I(W) > 0$ if, for any given constant $R <C$, there exists a sequence of codes  with rate $R$ and the block-error probability  for the codes under maximum likelihood (ML) decoding vanishes as the block length $N$ grows large. 
Leveraging polar codes, invented by Ar{\i}kan \cite{arikan2009channel}, and their extensions to large kernels, with errors exponents studied in \cite{ korada2010polar}, we show that capacity-achieving polar codes with column weights bounded by any polynomial of $N$ exist. However, we show that  for rate-$1$ polar codes, the fraction of columns of the GM with weight $O( (\log N)^r)$ vanishes as $n$ grows large, for any $r>0$.
% .  a similar result cannot be obtained with any polynomial of $\log N$ as the constraint on column weights for rate-$1$ polar codes. 
A new construction of LDGM codes is devised in such a way that most of the column weights can be bounded by a degree-$(1+\delta)$ polynomial of $\log N$, where $\delta>0$ can be chosen to be arbitrarily small.
A drawback of  the new construction is the existence of, though only small in fraction, \textit{heavy} columns in the GM. In order to resolve this, we propose a splitting algorithm which, roughly speaking, splits \textit{heavy} columns into several \textit{light} columns, a process which will be clarified later. The rate loss due to this modification is characterized, and is shown to approach zero as $N$ grows large. Hence, the proposed modification leads to capacity-achieving codes with column weights of the GM upper bounded by $(\log N)^{1+\epsilon}$, for any $\epsilon > 2\epsilon^*$, where $\epsilon^* = \log 3 - \frac32 \approx 0.085$. 
For the constructed LDGM codes, 
one cannot use the successive cancellation decoding of the underlying polar codes, and hence we show achievability of the capacity using an ML decoder.

% \blue{
Leveraging moderate-deviation results for the random linear code ensemble (RLE) and polar codes, where the gap between the channel capacity and code rate vanishes as blocklength grows, the new construction is applied to yield two families of LDGM codes. In both codes, the trade-off between the gap-to-capacity and the GM sparsity and the decoding time complexity can be observed. 
In general, the RLE-based codes achieve better GM sparsity at the cost of higher (quasi-polynomial) decoding complexity compared to the polar-based codes. When the optimal GM column sparsity is desired, we have a fixed rate $R <C$, and the GM column weights can be bounded by $O(\log N)$ and $(\log N)^2$  for RLE-based and polar-based codes, respectively.
Note that the latter is larger than the aforementioned $(\log N)^{1+\epsilon}$ bound since the column-splitting algorithm is not employed. 
One caveat of the RLE-based codes is the hidden constant in the $O(\log N)$ bound, which may be arbitrarily large for $R$ close to $C$.  
Another caveat for RLE-based codes is the decoding time complexity - polynomial in $N$ with arbitrarily large degree for $R$ close to $C$ - when compared to almost-linear complexity of the polar-based codes. 
% }

To  address  the  decoding  complexity in the polar-based codes with  $(\log N)^{1+\epsilon}$ GM sparsity, we first consider BEC,  and propose a different splitting algorithm, termed \textit{decoder-respecting splitting (DRS)} algorithm.
Leveraging the fact that the polarization transform of a BEC leads to BECs, the DRS algorithm converts the encoder of a standard polar code into an encoder defined by a sparse GM which does not decrease the reliability of the bit-channels observed by the source bits.  As a consequence, the information bits can be recovered by a low-complexity successive cancellation decoder in a recursive manner. 
In particular, we show a sequence of codes defined by GMs with column weights upper bounded by $(\log N)^{2\lambda}$, for any $\lambda > \lambda^* = \log3 -1 \approx 0.585$ and the existence of a decoder with computation complexity $O(N \log\log N)$
under which the sequence is capacity achieving.

Next, for general BMS channels, a scheme that modifies the encoder of a standard polar code in a slightly different way from the DRS algorithm, referred to as \textit{augmented-DRS (A-DRS)} scheme, is devised which requires additional channel uses and decoding complexity.  
In spite of these limitations, we show that there exists a sequence of capacity-achieving A-DRS codes.
The sequence  of codes is defined by GMs with column weights upper bounded by $(\log N)^{2\lambda}$, for any $\lambda > \lambda^{\dagger} = (\log3)^{ -1} \approx 0.631$, and can be decoded with complexity $O(N \log\log N)$. 
 
% \blue{
The main contribution of this work can be summarized as follows:
\begin{itemize}
    \item Show a sequence of capacity-achieving polar codes with GM column weights upper bounded by a  polynomial sublinear in $N$, with any given degree $s \in (0, 1)$. This verifies the conjecture given in \cite{LDGM_capAchieving2011}.
    \item Propose a construction of polar-based capacity-achieving LDGM codes with efficient decoding where almost all GM column weights are upper bounded by $ (\log N)^{1+ \delta} $ for any given $\delta >0$. The ratio of columns not bounded in this way vanishes as $N$ grows large. 
    \item Show the GMs of above LDGM code sequence can be modified by a column-splitting algorithm so that the resulting code sequence is capacity-achieving under ML decoding while all the GM column weights are $ O( (\log N)^{1.17} )$. 
    \item Show a sequence of capacity-achieving RLE-based codes with all the GM column weights bounded by $ O( \log N )$ under polynomial-time decoding algorithm.
    \item For BEC transmission, propose a second splitting algorithm, such that  the resulting code sequence is capacity-achieving under an efficient decoder while all the GM column weights are $ O( (\log N)^{1.17} )$
    \item For general BMS channels, propose a third splitting algorithm to generate code sequence, which is capacity-achieving under an efficient decoder while all the GM column weights are $ O( (\log N)^{1.262} )$.
\end{itemize}
% }

\subsection{Organization}
The rest of the paper is organized as follows. In Section \ref{sec:Prelim}, we introduce basic notations and definitions for channel polarization and polar codes. 
% \blue{
In Section \ref{Sec:AllMain}, the statements of all the main results are given formally.
% }
In Section \ref{sec:SparseLDGM}, we show a sparsity result for polar codes with general kernels. Also, the new construction and the first splitting algorithm are described, and the corresponding code is analyzed in terms of the sparsity of the GM and the probability of error. 
% \blue{
Furthermore, the RLE-based and polar-based code families that allow moderate deviation trade-off are constructed and compared with each other. 
% }
In Section \ref{sec:LDGMwithDEC}, we introduce the DRS algorithm and the A-DRS scheme, and the corresponding code constructions over the BEC and BMS channels, respectively. The successive cancellation decoders are also described and shown to be of low computation complexity.
In Section \ref{sec:L>2}, the relationship between the polarization kernel and the sparsity of the GM is evaluated. 
Finally, Section \ref{sec:Conclusion} concludes the paper. 
The proofs for the results in Sections \ref{sec:SparseLDGM}, \ref{sec:LDGMwithDEC}, \ref{sec:L>2} are included in the Appendix.

\section{Preliminaries}\label{sec:Prelim}

Let $h_b(\cdot)$ denote the binary entropy function, $\exp_2{(x)}$ denotes the function $2^x$,  $\ln(\cdot)$ be the logarithmic function with base $e$,  
and $\log(\cdot)$ be the logarithmic function with base $2$. $Z(W)$ denote the Bhattacharyya parameter of a channel $W$. 
% \blue{
A binary memoryless symmetric channel (BMS)  $W: \mathcal{X} \rightarrow \mathcal{Y} $ is a noisy channel with binary input alphabet $ \mathcal{X}$, and channel output alphabet $ \mathcal{Y}$, (we use $ \mathcal{X}  = \mathset{0, 1}$, and assume the output alphabet $ \mathcal{Y}$ is finite, in this paper.) satisfying two conditions:
\begin{enumerate}
    \item The channel output is conditionally independent on past channel inputs, given the input at the same time.
    \item There exist a bijective mapping $\phi :  \mathcal{Y} \rightarrow  \mathcal{Y}$ satisfying  $ \phi^2(y) = y$ for all $y \in  \mathcal{Y}$ such that, the probability of receiving $y \in S$ on input $0$ is equal to the probability of receiving $y \in \phi(S)$ on input $1$ for any set $S \subset \mathcal{Y}$.
    % The probability of receiving output $y \in [a, b]$ on input $0$ is the same as the one of receiving output $y \in [-b, -a]$ on input $1$.
\end{enumerate}
% }

\subsection{Channel Polarization and Polar Codes}\label{Prelim:polar}

The \emph{channel polarization} phenomenon was discovered by Ar{\i}kan \cite{arikan2009channel} and is based on a $2 \times 2$ polarization transform as the building block.
Let $A \otimes B$ denote the Kronecker product of matrices $A$ and $B$, and $A^{\otimes n}$ denote the $n$-fold Kronecker product of $A$, i.e., $A^{\otimes n} = A \otimes A^{\otimes (n-1)}$ for $n\geq 2$ and $A^{\otimes 1} = A.$
Let $\mathcal{W}$ denote the class of all BMS channels.
Consider a channel transform $W \mapsto (W^{-}, W^{+})$  that maps $\mathcal{W}$ to $\mathcal{W}^2$. 
Suppose the transform operates on an input channel $W: \mathcal{X} \rightarrow \mathcal{Y}$ to generate the channels 
$W^{-} :\mathcal{X} \rightarrow  \mathcal{Y}^2$ and $W^{+} :\mathcal{X} \rightarrow  \mathcal{Y}^2 \times \mathcal{X}$ with transition probabilities
\begin{align*}
 W^{-}(y_1,y_2|x_1) &= \sum_{x_2 \in \mathcal{X}}\frac{1}{2}W(y_1|x_1 \oplus x_2)W(y_2|x_2), \\
 W^{+}(y_1,y_2,x_1|x_2) &= \frac{1}{2}W(y_1|x_1 \oplus x_2)W(y_2|x_2), 
\end{align*} 
where $\oplus $ denotes mod-$2$ addition. 
This transformation $W \mapsto (W^{-}, W^{+})$ is referred to as a polarization recursion. 
Then a channel $W^{s_1, s_2, \ldots , s_n}$ is defined for $s_i \in \mathset{-,+}, i= 1,2,\ldots, n$, resulting from applying the channel transform $n$ times recursively as 
\begin{equation*}
    W^{s_1, s_2, \ldots ,s_n} = \begin{cases}
    (W^{s_1, s_2, \ldots, s_{n-1}})^{-} &\mbox{ if } s_n = -, \\
    (W^{s_1, s_2, \ldots, s_{n-1}})^{+} &\mbox{ if } s_n = +.
    \end{cases}
\end{equation*} 
 
For $N=2^n$, the polarization transform is obtained from the $N \times N$ matrix $G_2^{\otimes n}$, where $G_2=\begin{bmatrix}
		1 & 0 \\
		1 & 1 \\
		\end{bmatrix}$~\cite{arikan2009channel}. 
% 		\red{defin $\otimes$}
A polar code of length $N$ is constructed by selecting certain rows of $G_2^{\otimes n}$ as its generator matrix.
More specifically, let $K$ denote the code dimension. Then all the $N$ bit-channels in the set $\{W^{s_1, s_2,\ldots, s_n}:$ $s_i\in \mathset{-,+} \mbox{ for } i =1,2,\ldots, n\}$, resulting from the polarization transform, are sorted with respect to an associated parameter, e.g., their probability of error (or Bhattacharyya parameter), the best $K$ of them with the lowest probability of error are selected, and then the corresponding rows from $G_2^{\otimes n}$ are selected to form the GM. 
Hence, the GM of an $(N,K)$ polar code is a $K \times N$ sub-matrix of $G_2^{\otimes n}$. Then the probability of error of this code, under successive cancellation decoding, is upper bounded by the sum of probabilities of error of the selected $K$ best bit-channels \cite{arikan2009channel}. Polar codes and polarization phenomenon have been successfully applied to a wide range of problems including data compression~\cite{Arikan2,abbe2011polarization}, broadcast channels~\cite{mondelli2015achieving,goela2015polar}, multiple access channels~\cite{STY,MELK}, physical layer security~\cite{MV,andersson2010nested}, and coded modulations \cite{mahdavifar2015polar}.

\subsection{General Kernels and Error Exponent}\label{Prelim:polarExp}
It is shown in \cite{korada2010polar} that if $G_2$ is replaced by an $l \times l$ matrix $G$, then polarization still occurs if and only if $G$ is an invertible matrix in $\mathbb{F}_2$ and it is not upper triangular, in which case the matrix $G$ is called a \textit{polarization kernel}. Furthermore, the authors of \cite{korada2010polar} provided a general formula for the error exponent of polar codes constructed based on an arbitrary $l\times l$ polarization kernel $G$. More specifically, let $N= l^n$ denote the block length and $C$ denote the capacity of the channel. For any fixed $\beta < E(G)$ and fixed code rate $R< C$, there is a sequence of polar codes based on $G$ with probability of error $P_e$ bounded by 
% (\cite[Thm 11]{korada2010polar})
$$
P_e(n) \leq 2^{-N^{\beta}},
$$
for all sufficiently large $n$, where  
the rate of polarization (see \cite[Definition 7]{korada2010polar}) $E(G)$ is given by 
\begin{equation}
	E(G) = \frac{1}{l}\sum_{i=1}^l \log_l D_i, \label{eq:EGformula}
\end{equation}
and $\mathset{D_i}_{i=1}^l$ in \eqref{eq:EGformula} are the \textit{partial distances} of $G$. 
More specifically, for $G= [g_1^T, g_2^T, \ldots, g_l^T]^T$, the partial distances $D_i$ are defined as follows:
% \begin{definition}
 \begin{align}
    &D_i \triangleq d_H(g_i, \mbox{span}(g_{i+1}, \ldots, g_l) ), \qquad i= 1,2,\ldots, l-1 \label{eq:DiDef}\\
    &D_l \triangleq d_H(g_l, 0)= w_H(g_l) \label{eq:DlDef}, 
\end{align} where $d_H(a,b)$ is the Hamming distance between two vectors $a$ and $b$, and $d_H(a, U)$ is the minimum distance between a vector $a $ and a subspace $U$, i.e., $d_H(a, U)= \min_{u\in U}d_H(a,u)$.

\subsection{Moderate Deviation Results}\label{subsec:prelim:moderate}
% \color{blue} 
We briefly review the moderate deviation results of two sequences of codes, the random code
ensemble (RCE) and the polar code in this section. 
\begin{theorem}[Thm 2.1 of \cite{altuug2010moderate}]\label{Thm:RCEmoderate}
For any discrete memoryless channel (DMC) $W$ with dispersion $\sigma^2(W) >0$, for any sequence of real numbers $\mathset{\epsilon_N}_{N\geq 1}$ satisfying 
\begin{enumerate}[label=(\roman*)]
    \item $\epsilon_N \goto 0$, as $N \goto \infty$, 
    \item $\epsilon_N \sqrt{N} \goto \infty$, as $N \goto \infty$,
\end{enumerate}
there exists a sequence of codes $\mathset{C_N}$ with rate $R_N \geq I(W)  -\epsilon_N$ for all $N$ and 
\begin{equation*}
    \limsup_{N\goto \infty} \frac{1}{N \epsilon_n^2} \log P_{e, N} \leq - \frac{1}{2 \sigma^2(W) },
\end{equation*} where $P_{e,N}$ denotes the maximal probability of error among all messages in code $C_N$.
\end{theorem}
The proof of Theorem \ref{Thm:RCEmoderate} hinges on Corollary 2 on page 140 of \cite{gallager1968information}, which establishes the existence of an $(N,R)$ block code for the transmission over a DMC such that the probability of error for each message $m$, denotedy by $P_{e,m}$, is bounded as follows
\begin{equation*}
    P_{e,m} < 4 \exp{[-N E_r(R)]}, \mbox{ for each } m, 1\leq m \leq M = 2^{NR}. 
\end{equation*}
Let $K$ and $J$ denote the sizes of the input and output alphabets of $W$, respectively. 
The random coding exponent  $E_r(R)$ \cite{gallager1968information} is given  by 
\begin{align}
    E_r(R) &= \max_{ 0 \leq \rho \leq 1}\max_{\mathbf{Q}} [E_0(\rho, \mathbf{Q} ) - \rho R], \mbox{ where} \label{eq:Er_1} \\
    E_0(\rho, \mathbf{Q} ) &= - \ln \sum_{j= 0}^{J-1} [\sum_{k=0}^{K-1} Q(k) W(j | k)^{1/ (1+\rho) } ]^{1+\rho}, \mbox{ and }
\end{align}
$\mathbf{Q}$ denotes a probability distribution on the input alphabet of $W$. 
\begin{remark}\label{rem:RLEmoderate}
The random linear ensemble (RLE), as described in \cite[Chapter 6.2]{gallager1968information}, is given by 
\[
\mathbf{x} = \mathbf{u} G \oplus \mathbf{v}
\]
where $G$ is an $L = NR$ by $N$ binary matrix with each element chosen independently according to the Bernoulli($0.5$) distribution, and each entry of $ \mathbf{v} \in \bit^{N}$ is also given by independent Bernoulli($0.5$) distribution. 

For transmission over a BMS channel, the maximizing channel input distribution for \eqref{eq:Er_1} is the Bernoulli($0.5$) distribution. It is shown in \cite[p. 208]{gallager1968information} that there exists a linear code in RLE for which, over a BMS $W$, the block error probability can be bounded as 
\begin{equation}
    P_e \leq \exp{[-N E_r(R)]}.
\end{equation}
With argument similar to that for the Corollary in \cite[p. 207]{gallager1968information}, one may show the existence of a binary linear code for which, over a BMS $W$, the error probability for message $m$ is bounded as 
\begin{equation*}
    P_{e,m} < \exp{[-N E_r(R)]}, \mbox{ for each }  1\leq m \leq M = 2^{NR},
\end{equation*}
Adopting the same proof technique as in \cite{altuug2010moderate},  Theorem \ref{Thm:RCEmoderate} can be shown to hold where each code in the sequence $\mathset{C_N}$ is linear.
\end{remark}

For polar code, \cite{mondelli2016unified} shows the following moderate deviation result for the transmission over a BMS $W$ with capacity $I (W)$ by using a polar code of rate $R < I (W)$ with polarization kernel $G_2$. For a family of codes used for the transmission of a BMS channel, the \textit{scaling exponent} is defined \cite{mondelli2016unified} as the value of $\mu$ such that
\begin{equation*}
    \lim_{\N \rightarrow \infty, N^{1/\mu}(C-R) =z } P_e(N, R, C) = f(z) 
\end{equation*}
for some function $f(z)$.

\begin{theorem}[Theorem 3 of \cite{mondelli2016unified}]\label{Thm:PolarModerate}
Let $\mu$ denote the scaling exponent of the channel $W$ using a polar code, bounded by $3.579 \leq \mu \leq 4.714$. 
Consider the transmission over a BMS $W$ with capacity $I(W)$ by
using a polar code of rate $R < I(W)$.
Then for any $\gamma \in (1/ (1+\mu), 1)$, the block length $N$ and the block error probability under successive cancellation (SC) decoding $P_e$ are such that
\begin{align*}
    P_e &\leq N \cdot \exp_2{ (-N^{\gamma \cdot h_2^{-1} (\frac{\gamma(\mu +1) -1}{\gamma \mu}) } )},\\
    N &\leq \frac{\beta_3}{ (I(W) -R )^{\mu / (1-\gamma)}  },
\end{align*} where $\beta_3$ is a universal constant that does not depend on $W$ or $\gamma$. 
\end{theorem}
% \color{black}

\section{Main Result: Statement and Discussion} \label{Sec:AllMain}
% \color{blue}
We state the main results of this work informally in this section. Let $N$ denote the blocklength of the code, $R$ the code rate, and $C$ the capacity of the given channel. 

Our first result is about the construction of capacity achieving polar codes with sparse generator matrices for BMS channels. 
\begin{Result}[Polynomial Bound over BMS][Formally stated in Theorem \ref{prop:polarWithPolyColumnWeight}] \label{thm:0-1}
	For  any fixed $s>0$, there exists capacity-achieving polar codes for BMS channels, where the Hamming weight of all the columns in the generator matrices are upper bounded by $N^s$.
\end{Result}

Our second result shows that a concatenation of polar codes with rate-1 code yields a sequence of capacity achieving codes with sparse generator matrices and almost-linear decoding complexity under SC-based decoding. 
\begin{Result}[Polar-based Moderate Deviation][Formally stated in Proposition \ref{prop:PolarExtended}]  \label{thm:0-4}
Let a BMS channel with scaling exponent $\mu$ with respect to the polar code family with kernel $G_2$, and two constants $\lambda \in (0, \frac{1}{1+ \mu})$ and $\alpha \in (0, \frac12)$  be given. 
Denote by $\beta_1$ and $\beta_2$ the constants $\beta_1 = \frac{-\lambda}{(1-\lambda \mu)  \cdot h_2^{-1} ( 1- \frac{\lambda}{ 1- \lambda \mu}) }$ and $\beta_2 = \frac{0.585}{(1-\lambda \mu)  \cdot h_2^{-1} ( 1- \frac{\lambda}{ 1- \lambda \mu})}$.
There exists a sequence  of (polar-based) linear codes  with rate $R = R(N)$ for which the following scaling behaviour holds
\begin{enumerate}
    \item The gap to capacity ${C - R(N)} = O( (\log N)^{\beta_1 } )$,
    \item The block error probability $P_e \leq 2^{- O( \log N ) } $,
    \item The decoding time complexity $Comp =(N)^{1+ o(1)}$,  
    \item The average GM column weight 
    $ W_{col,avg} =  O( (\log N)^{ \beta_2} ).$
\end{enumerate}
\end{Result}
Our third result shows that a concatenation of capacity achieving RLE-based codes and a rate-1 code gives a sequence of capacity achieving codes with sparse generator matrices and with sub-exponential-time decoding complexity under ML decoding.

\begin{Result}[RLE-based Moderate Deviation][Formally stated in Proposition \ref{prop:RLEextended}] 
\label{thm:0-3}
Let a BMS channel and a constant $\alpha \in (0, \frac12)$  be given. 
There exists a sequence  of (RLE-based) linear codes with rate $R = R(N)$ for which the following scaling behaviour holds
\begin{enumerate}
    \item The gap to capacity ${C - R(N)} = O( (\log N)^{\frac{-\alpha}{1-2\alpha}} )$,
    \item The maximal probability of error $P_e = 2^{- O( \log N ) } $,
    \item The decoding time complexity $    Comp \leq \exp_2{( O( (\log N)^{\frac{1}{1-2\alpha}} ) )}$,  
    \item The average GM column weight 
    $    W_{col,avg} = O( (\log N)^{\frac{1}{1-2\alpha}} ).$
\end{enumerate}
\end{Result}

A novel DRS column-splitting algorithm is proposed, using which the following result shows the existence of a sequence of capacity achieving polar-based codes over BEC with sparse generator matrices and with almost-linear-time decoding. In comparison to the previous result, in this setting, we have a better bound on each column weight. 
\begin{Result}[LDGM with Efficient Decoding over BEC][Formally stated in Theorem \ref{coro:DRS_lowComp_sparse_BEC}] 
\label{thm:0-5}
    Let $\beta  \in (0, 0.5) $, $d > 1.17$ and a BEC be given.
    There exists a sequence of capacity-achieving codes satisfying the following properties:
    \begin{enumerate}
        \item 
        The error probability  $P_e \leq  \exp_2{( -(\log N)^{2\beta} )}$,
        \item
        Each GM column weight is upper bounded by $(\log N)^{d}$,
        % \item
        % The rate approaches $C$ as $n$ grows large.
        \item 
        There is a SC-based decoder with time complexity $O(N \log\log N)$.
    \end{enumerate}    
\end{Result}

The above result on BEC is extended to general BMS channels in the following result. This is done by proposing a second column-splitting algorithm called A-DRS.
\begin{Result}[LDGM with Efficient Decoding over BMS][Formally stated in Theorem \ref{prop:Sparse_ADRScompl}]     \label{thm:0-6}
   Let $\beta   \in (0, 0.5)  $,  $d > 1.262$, and a BMS channel be given. There exists a sequence of capacity-achieving codes with the following properties:
    \begin{enumerate}
        \item 
        The error probability is upper bounded by $ \exp_2{(-(\log N)^{2\beta})}$,
        \item
        Each GM column weight is upper bounded by $(\log N)^{d}$,
        % \item
        % The rate approaches the capacity $C$ as $n$ grows large.
        \item
        The decoding time complexity is upper bounded by $O(N \log\log N).$
    \end{enumerate}
\end{Result}

% \color{black}

\section{Sparse LDGM Code Constructions}\label{sec:SparseLDGM}
In Section \ref{sec:sparsityConstraint}, we prove the existence of capacity-achieving polar codes with large kernels with polynomially sublinear column weights in the GM. In Section \ref{sec:ConstructionNewApproach}, a coding scheme which expands a polar code into one with larger blocklength, and a \textit{splitting algorithm} which reduces the maximal column weight by slightly increasing the blocklength, are given. In Section \ref{sec:PeForNewApproach}, we provide an upper bound for the block error probability of the proposed polar-based coding scheme, under a SC-based decoding algorithm. In Section \ref{subsection:mainL2}, we analyze the sparsity in the GM when the underlying polarization kernel is the $G_2$ matrix. 
% \blue{
In Section \ref{subsec:constructionA} the coding scheme is generalized to general binary linear codes. In Sections \ref{subsec:RLE}  and \ref{subsec:polar} the relationship between the scaling of gap-to-capacity, block error probability, decoding complexity, and GM column weight bound are investigated for two families of capacity-achieving codes obtained by applying the scheme to RLE and polar codes. In Section \ref{subsec:Compare} we compare the performance of RLE-based and polar-based codes and showed that the former can achieve $O( \log(N) )$ GM column weight at the cost of a high decoding complexity. 
% }

\subsection{Generator Matrix Sparsity} \label{sec:sparsityConstraint}

% Leveraging results in polar coding theory, 
In this section we first show the existence of capacity-achieving polar codes with generator matrices for which all column weights scale at most polynomially with arbitrarily small degree in the block length $N$, hence validating
% \gray{(right word?)} 
the conjecture in \cite{LDGM_capAchieving2011}. 
Second, we show that, for any polar code of rate $1$, \textit{almost} all of the column weights of the GM are larger than polylogarithmic in $N$.
% \red{add proof sketch}
\begin{theorem}\label{prop:polarWithPolyColumnWeight}
	For  any fixed $s>0$ and any BMS channel, there are capacity-achieving polar codes with generator matrices having column weights bounded by $N^s$, where $N$ denotes the block length of the code.
\end{theorem} 
    \begin{proposition}\label{prop:noLogColumns}
	Given any $l \geq 2$, $l\times l $  polarizing kernel $G$, and $r>0$,
	%     $G^{\otimes n}$ has dimension $N$-by-$N$, where $N = l^n.$ 
	the fraction of columns in $G^{\otimes n}$ with $O({(\log N)}^r)$ Hamming weight vanishes as $n$ grows large.
%	, where $N= l^n$
\end{proposition}

\subsection{Construction and Splitting Algorithm} \label{sec:ConstructionNewApproach}
In this section, we propose a new construction of polar-based codes with even sparser GMs than those given in Section  \ref{sec:sparsityConstraint}. In particular, \textit{almost all} the column weights of the GMs of such codes scale at most logarithmically in the code block length, and there is an upper bound $w_{u.b.}$, polynomial in the logarithm of the block length, on \textit{all} the column weights.

Formally, let $G= G_l^{\otimes n}\otimes I_{n'} $, where $G_l $ is an $l\times l$ polarization kernel
    % binary non-upper triangular matrix 
and $I_{n'} $ is the $n' \times n' $ identity matrix. 
The matrix $G$ has the following form: 
\begin{equation}
G=
  \begin{bmatrix}
    G_l^{\otimes n}       & \zero_{l^n} & \zero_{l^n} & \dots & \zero_{l^n} \\
    \zero_{l^n}       & G_l^{\otimes n} & \zero_{l^n} & \dots & \zero_{l^n} \\
    \zero_{l^n} & \zero_{l^n}      & G_l^{\otimes n}  & \dots & \zero_{l^n} \\
    \vdots & \vdots & \vdots & \ddots & \vdots \\
    \zero_{l^n}       & \zero_{l^n} & \zero_{l^n} & \dots & G_l^{\otimes n}
  \end{bmatrix}. \label{eq:Gdef}
% \vspace{2mm}
\end{equation} 
Let $N= l^n$,  $N' = N\times n'$ denote the overall block length, and $K'= n'K$ denote the code dimension.
Then $\frac{K'}{N'}= \frac{K}{N}$ is the code rate. To construct the polar-based code, we use the $K'$ bit-channels with the lowest probability of error (or Bhattacharyya parameter), and the GM of the resulting $(N',K')$ code is the corresponding $K' \times N'$ sub-matrix of $G$ same as how it is done for regular polar codes. The code constructed in this fashion is referred to as a \textit{polar-based code corresponding to} $G$, also referred to as a PB($G$) code, in this paper.

When all columns are required to be sparse, that is, have low Hamming weights, a \textit{splitting algorithm} is applied. 
Given a column weight threshold  $w_{u.b.}$, the splitting algorithm splits any column in $G$ with weight exceeding $w_{u.b.}$ into columns that  sum to the original column both in $\Ftwo$ and in $\R$, and that have weights no larger than $w_{u.b.}$. 
Hence, if the $i$-th entry of the original column has a one, then exactly one of the new columns has a one in the $i$-th entry. If the $i$-th entry of the original column is zero, then the $i$-th entries of all the new columns are zeros.
That is, for a column in $G$ with weight $w$, if $w \leq w_{u.b.} $, the algorithm keeps the column as is. 
If $w= m\cdot w_{u.b.} + r$ for some $m\in \N$ and some integer $0<  r \leq w_{u.b.}$, the algorithm replaces the column with $m+1$ columns, such that each of the new columns has at most $w_{u.b.}$ ones. 
Let $G'$ denote the resulting $N'\times N'(1+\gamma) $ matrix. A new code\textit{ based on}
$G'$ selects the same $K'$ rows as the polar-based code corresponding to  $G$ to form the generator matrix, where all the column weights are bounded by $w_{u.b.}$. 
Such a code is referred to as a \textit{polar-based code corresponding to} $G'$, or a PB$(G')$ code, in this paper.

\begin{figure}
     \centering
     \begin{subfigure}[b]{0.35\textwidth}
         \centering
        \includegraphics[trim=1cm 1cm 17.5cm 10cm, clip, width= \textwidth]{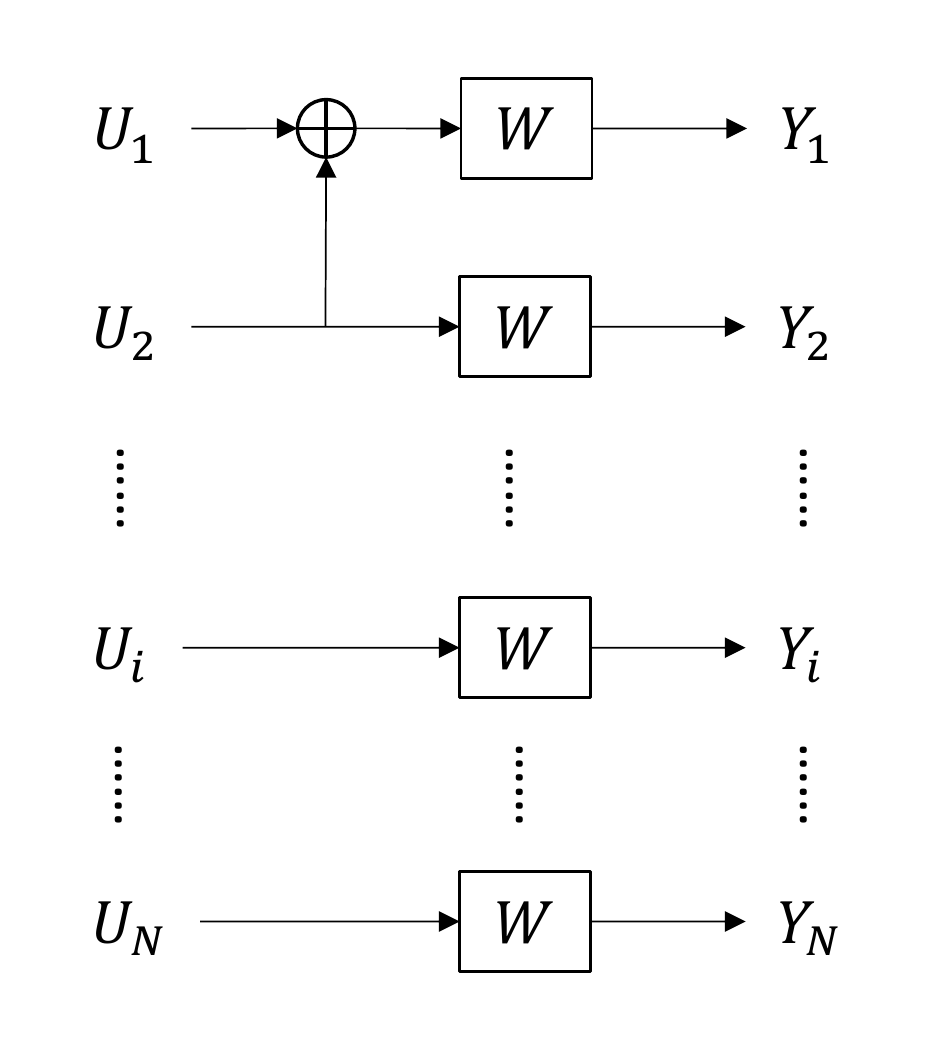}
         \caption{Encoding block diagram for $G$}
         \label{fig:Example_beforeSplit}
     \end{subfigure}
    %  \hfill
     \begin{subfigure}[b]{0.35\textwidth}
         \centering
        \includegraphics[trim=1cm 1cm 18cm 10.5cm, clip, width= \textwidth]{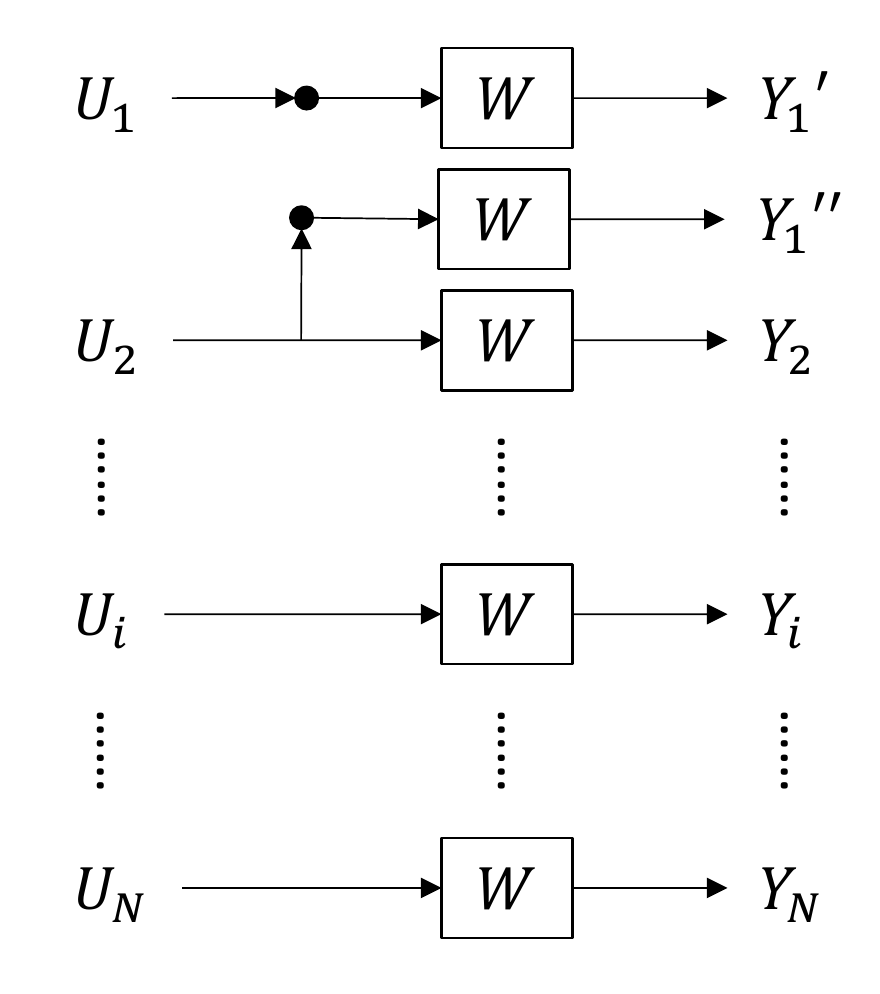}
         \caption{Encoding block diagram for $G'$}
         \label{fig:Example_afterSplit}
     \end{subfigure}
    \caption{An example for splitting algorithm with  $w_{u.b.}=1$ and $v_j =e_j\in \Ftwo^N$ for $j=2,3,\ldots, N$, where $e_j$ has a one at and only at the $j$-th entry.}
    \label{fig:Example_split}
\vspace{-7mm}
\end{figure}

The operation of the splitting algorithm is demonstrated through an example. 
Suppose that the threshold $w_{u.b.}$ is chosen to be $1$, and the first column of an $N$-column matrix $G$ is  $(1,1, 0, \ldots, 0)^T$. Then this column will be split into two new columns, $(1,0, 0, \ldots, 0)^T$ and $(0,1, 0, \ldots, 0)^T$, denoted by $v_1'$ and $v_1''$ here. 
Assuming all the other columns of $G$ have weights $0$ or $1$,  the resulting $G'$ will be  
\begin{equation*}
G'=[v_1', v_1'', v_2, \ldots, v_N],
\end{equation*} where $v_i$ denotes the $i$-th column of $G$. In this example, we see that $\gamma = 1/N$. The encoding block diagrams for the generator matrices $G$ and $G'$ are shown in Figure\,\ref{fig:Example_split}.

\subsection{Analysis of the Error Probability} \label{sec:PeForNewApproach}

In this section, we first show that for a fixed $G_l$, PB($G$) codes have vanishing probability of error as $n$ grows large, where $n'$ is chosen as a function of $n$ and $G_l$, as stated in the sequel, and $G$ is constructed as in \eqref{eq:Gdef}.
Second, if $G'$ is the output of the splitting algorithm described in Section \ref{sec:ConstructionNewApproach}, we show that the probability of error for the PB($G'$) codes can be bounded from above in the same fashion as the   PB($G$) codes.
A \emph{sparsity benchmark} for the generator matrix of a code, in terms of the block length, as well as the 
geometric mean column weight and  maximum column weight, are also defined in this section.

Given a BMS channel $W$ with capacity $C$, let $R < C$ and a parameter $\beta < E(G_l) $ be fixed. 
Then there are polar codes with rate $R$ constructed using the kernel $G_l$ with the probability of error upper bounded by $2^{-N^\beta}$  under SC decoding for sufficiently large $n$\cite{korada2010polar}.
Using the union bound,  it can be observed that the probability of error of a corresponding PB($G$) code is upper bounded by $n' \cdot 2^{-N^\beta}$.
Hence, we choose
    \begin{equation}
    n' = 2^{N^{(1-\delta)E(G_l)}}\label{eq:n'def},
    \end{equation} for any constant $1 > \delta > 1- \frac{\beta}{E(G_l)}$. This choice of $n'$ is used throughout the paper. We then have the following lemma.
    \begin{lemma}\label{Lemma:newG_PeBound}
    	Let $G$ be as in \eqref{eq:Gdef} and $n'$ be as in \eqref{eq:n'def}. Consider the transmission over a BMS channel $W$ with capacity $C$.  Then for any  $R < C$ and any $\beta < E(G_l) $, there exists a sequence of PB($G$) codes with code rate $R$ with the probability of error upper bounded by  $\exp_2{(-(\log N') ^{ \frac{\beta}{E(G_l) }} )} $, for all sufficiently large $n$. 
    \end{lemma}

Suppose that the splitting algorithm with a certain given threshold, as discussed in Section\,\ref{sec:ConstructionNewApproach}, is applied to $G$ and that $G'$ is obtained. We show in the following proposition that the probability of error of the codes corresponding to $G'$ and $G$ can be bounded in the same way.

\begin{proposition}\label{prop:codeWithG'}
	 Consider the transmission over a BMS channel $W$ with capacity $C$. For any $\beta < E(G_l)$ and any rate $R < C$, there is a decoding scheme based on successive cancellation (SC) decoding such that the probability of error of the PB$(G')$ code with dimension $K' = N' R$ is bounded from above by $\exp_2{(-N^\beta)} \leq \exp_2{(-(\log N') ^{ \frac{\beta}{E(G_l) }} )} $, for sufficiently large $n$.
\end{proposition}

While there are no lower bounds on the sparsity limit of capacity-achieving LDGM codes, 
the best known results for the column weights are $O(\log N)$ and $O(N)$ for BEC and general BMS channels, respectively. Theorem \ref{prop:polarWithPolyColumnWeight} improves the result for general BMS channels by showing the existence of capacity-achieving LDGM codes with $O(N^s)$ column weights for any fixed $s>0.$ 
Nevertheless, the achievability with $O(\log N)$ sparsity remains unknown and hence, 
% \red{End here}
we use $ \log(N')$ as \textit{sparsity benchmark} in this paper, where
\begin{equation}\label{eq:N'def}
N' =n' N= 2^{N^{(1-\delta)E(G_l)}}  N 
\end{equation}
denotes the blocklength of the PB($G$) code, and  $ log(N')$ scales polynomially in $N$ as 
\begin{align}
	log(N')&= N^{(1-\delta)E(G_l)}+\log{N}  =   N^{(1-\delta)E(G_l) + o(1)} 
	\label{eq:logN'formula}
\end{align}{for all sufficiently large $n$.}

We analyze the column weights of $G$ compared to $log(N')$ in two main scenarios: (1) geometric mean column weight and (2) maximum column weight. 
% \blue{
The geometric mean is of interest because, as we show in Section \ref{subsection:mainL2}, the logarithm of the column weights of a polar code generator matrix concentrate around its (arithmetic) mean, which equals the logarithm of the `geometric mean column weight'. In other words, for a polar or a polar-based code, the scaling behaviour of the geometric mean column weight represents that the weights of typical columns. Hence in order to improve in the maximum column weight scenario, it suffices to study the `outliers' - the columns with weights much larger than the geometric mean column weight.  
% }

\begin{definition}
For a binary matrix $G$ with $m$ columns, whose weights are denoted by $w_1, w_2,$ $\ldots,$ $w_m $, the geometric mean column weight $ w_{GM}(G) $ and the maximum column weight $ w_{max}(G) $ are defined as follows:
	\begin{align}
	 w_{GM}(G) &\triangleq  (w_1\times w_2 \times \ldots \times w_m)^\frac{1}{m} ,\\
	 w_{max}(G) &\triangleq  \max_{i}{w_i}.
	\end{align}  
\end{definition}

When $G$ is constructed as in \eqref{eq:Gdef}, the geometric mean column weight and the maximum column weight of $G$ are equal to those of $G_l^{\otimes n}$, respectively. 
Since the values do not depend on $n'$, we set the following definitions: 
\begin{align}
    & w_{GM}(n, G_l)  \, \triangleq \, w_{GM}(G_l^{\otimes n}) = w_{GM}(G_l^{\otimes n}\otimes I_{n'} ). \label{eq:wMCdef} \\
    & w_{max}(n, G_l ) \, \triangleq \, w_{max}(G_l^{\otimes n}) = w_{max}(G_l^{\otimes n}\otimes I_{n'} ). \label{eq:wmaxDef}
\end{align}
Note that $w_{GM}(n, G_l) =w_{GM}(G_l)^n$  
and $ w_{max}(n, G_l )= (max_i(w_i) )^n \leq l^n $.

\subsection{Sparsity with Kernel $G_2$}\label{subsection:mainL2}
In this section, we study the GM sparsity when the polarization kernel is chosen as $G_2$. 
Let $G= G_2^{\otimes n}\otimes I_{n'} $\,  with 
$n'$ chosen as in \eqref{eq:n'def}.
We show two things in this section: 
$w_{GM}(n, G_2)$ $ \approx \log N'$ and, after careful splitting we get a matrix $G'$ such that $w_{max}(G') \leq (\log N')^{1+2\epsilon^*}$ for a constant $\epsilon^* \approx 0.085$ with vanishing loss of rate compared to $G$. 

Given any BMS channel $W$, the following proposition gives a sequence of capacity-achieving PB$(G)$ codes over  $W$ with the geometric mean column weight \textit{almost} logarithmic in the block length.
\begin{proposition}\label{prop:wMCloginN}
For any fixed $\delta' >0$, $n'$ in \eqref{eq:n'def} can be chosen such that 
	\begin{align}
	w_{GM}(n, G_2)=[\log(N')]^{1+\delta'} \label{eq:wMCforG2}
	\end{align} for all sufficiently large $n$. \label{prop:wMC_l=2}
\end{proposition}

Based on Lemma \ref{Lemma:newG_PeBound} and Proposition \ref{prop:wMCloginN}, given a BMS channel $W$, there is a sequence of capacity-achieving PB$(G)$ codes over $W$ with their geometric mean column weight upper bounded by $[\log(N')]^{1+\delta'}$ for large $n$.

\begin{lemma}\label{lem:wMClogN_most}
For any fixed $\delta' >0$ and any $\delta'' >\delta'$,  $n'$ in \eqref{eq:n'def} can be chosen such that the ratio of columns with weights exceeding $[\log(N')]^{1+\delta''}$ is vanishing as $n$ grows large.
\end{lemma}

Although the geometric mean column weight of $G$ and the weights of most columns are almost logarithmic in $N'$, the maximum column weight is $w_{max}(G)= 2^n = [{w_{GM}(G) }]^2$ and is at least ${(\log{N'})}^2$.
However, we show next that a matrix $G'$ can be obtained from the splitting algorithm such that all column weights are below some threshold $w_{u.b.}$ whose power over $\log{N'}$ is smaller than $2$. 

Since polar codes and the code corresponding to  $G$ are capacity-achieving, as shown in Lemma \ref{Lemma:newG_PeBound}, and that the rates of the PB$(G)$ and PB$(G')$ codes differ by a ratio $1+\gamma$, the latter is capacity achieving if $\gamma$ vanishes as $n$ grows large.
In the following, we explore appropriate choices of the column weight threshold for the splitting algorithm that allow the value $\gamma$ goes to $0$ exponentially fast. 

Let $\epsilon>0 $ be given and  \begin{equation}
    w_{u.b.}= (\log N')^{1+\epsilon}= {N}^{  \frac{1}{2}+\epsilon' },   \label{eq:wubDef}
\end{equation} be the upper bound for the column weights, where 
\begin{equation}\label{eq:epsilon'def}
    \epsilon'= (1+\epsilon ) \left( \frac{1-\delta}{2}  +o(1)\right)-\frac{1}{2},
\end{equation} for sufficiently large $n$, according to equation \eqref{eq:logN'formula}. In order to estimate the multiplicative rate loss of $1+\gamma$ in terms of the threshold $w_{u.b.}$, we may study the effect on $G_2^{\otimes n}$.

First note that $\gamma $ is the ratio of the number of extra columns resulting from the splitting algorithm to the number of columns $N$ of $G_2^{\otimes n}$. Let $w_1, w_2, \ldots, w_N$ denote the column weights of $G_2^{\otimes n}$. 
 The term $\gamma$ can be characterized as follows:
\begin{equation}\label{eq:RweightCount_1}
	\gamma =\frac{1}{N} \sum_{k=1}^{k_{max}}\left(
	\big\vert{ 
%		\Bigl\{ 
\{		i:  k\, w_{u.b.}+1 \leq w_i < (k+1)\, w_{u.b.}    \}
%		\Bigr\} 
	}\big\vert \times k \right),
\end{equation}where $k_{max}=\floor{\nicefrac{2^n}{w_{u.b.}} }$.
Let $X_1, X_2, \ldots$ be a sequence of i.i.d. Bernoulli($\frac{1}{2}$) random variables 
% defined by $
% X_i \triangleq \log { \frac{W_{i}}{W_{i-1}}}$,
and  $X(n) \triangleq \sum_{i=1}^n X_i \sim \text{Binomial}(n, \frac12)$.  
Then $\gamma$ can be written as a sum of probability terms involving $X(n)$. 
\begin{lemma}\label{lemma:RinXn}
	The ratio $\gamma$, characterized in \eqref{eq:RweightCount_1}, is equal to
	\begin{align}\label{eq:RinXn}
	\gamma= 
    \sum_{k=1}^{k_{max}} \Pr\left(X(n)  > \log{(k\cdot w_{u.b.}  )} \right).
	\end{align} 
\end{lemma}

Suppose that $ \log{(w_{u.b.})}$ is an integer denoted by $ n_{lub}$. (Otherwise, one may use $ \floor{ \log{(w_{u.b.})} }$ and the analysis still holds.) By grouping the $k_{max}$ terms in \eqref{eq:RinXn}, the ratio $\gamma$ can be expressed as a sum of  $\log{k_{max}}$ terms, 
% . Another alternative expression for $R$ is , 
as stated in the following lemma. Note that $  \log{k_{max}} = n-n_{lub}$.
\begin{lemma}\label{lemma:RasSum_ai}
We have $
	\gamma = a_0+a_1+\ldots +a_{n-n_{lub}-1}, 
$ where $a_i \triangleq    2^i \, \Pr{\left( X(n) \geq 1+i+n_{lub}\right)}$.	
\end{lemma}

Let $\lambda(x,y) \triangleq -D_{KL}\left(\frac{1}{2}+x+y || \frac{1}{2} \right)+ y$ for $x,y \geq 0$ and $x+y \leq \frac{1}{2},$ where  $D_{KL}(p_1||p_2)$ denotes the Kullback–Leibler divergence between two distributions Ber$(p_1)$ and Ber$(p_2).$
We characterize the asymptotic behaviour of the terms in Lemma \ref{lemma:RasSum_ai} in the following lemma. 
\begin{lemma}  \label{lemma:aiAsymp}
Let 
$\epsilon'$ be as specified in \eqref{eq:epsilon'def}. Then we have
	\begin{equation}
	a_i \doteq   \, 2^{n \lambda(\epsilon', \alpha_i)}, \label{eq:aiApprox}
	\end{equation} where  $\alpha_i = \frac{i+1}{n}$, 
	and $a_n \doteq b_n$ denote that $\frac{1}{n}\log \frac{a_n}{b_n} \to 0$ as $n\to \infty$.
\end{lemma}
\begin{proposition}
    \label{Thm:ratelossAndEpsi'} 
Let $G= G_2^{\otimes n}\otimes I_{n'} $, where $n', N', w_{u.b.} $, and $\epsilon'$ are given by \eqref{eq:n'def}, \eqref{eq:N'def}, \eqref{eq:wubDef}, and \eqref{eq:epsilon'def}, respectively. 
Suppose that the splitting algorithm is applied to form a matrix $G'\in \{0,1\}^{N'\times N'(1+\gamma)}$ such that $w_{max}(G') \leq w_{u.b.}$. 
Then $\gamma$ has the following asymptotic expression:
\begin{align}\label{eq:Rasymp}
	\gamma \doteq \begin{cases}
	2^{n(\epsilon^*- \epsilon')} \quad \to 0, &\mbox{if } \epsilon'> \epsilon^* \\
	2^{n\lambda(\epsilon', \alpha_{max} ) } \, \to \infty 
	, &\mbox{if } \epsilon' < \epsilon^*\end{cases},
\end{align}
where $ \epsilon^* \triangleq  \hspace{1mm}
	\log3 -\frac32
	\approx 0.085$ and $\alpha_{max}= \max_i \alpha_i$. 
\end{proposition}

The conditions in \eqref{eq:Rasymp} can be expressed in terms of the relation between $\epsilon$ and $\epsilon^*$ leading to the following corollary. 
\begin{corollary} \label{Coro:ratelosAndEpsi}
	Let $n', N', \epsilon'$, $\epsilon^*$, $w_{u.b.}$, and $\alpha_{max}$ be as specified in Proposition \ref{Thm:ratelossAndEpsi'}. Then $\gamma \rightarrow 0$ exponentially fast in $n$ if $\epsilon> 2 \epsilon^*$, and $\gamma \rightarrow \infty$ exponentially fast in $n$ if $\epsilon < 2 \epsilon^*$.
\end{corollary}

The rate loss $1+\gamma$ of the code corresponding to  $G'$ compared to the code corresponding to $G$ can thus be made arbitrarily close to $1$ when the column weight upper bound $w_{u.b.}$ is appropriately chosen. By combining the results in Section \ref{sec:PeForNewApproach} and the Corollary \ref{Coro:ratelosAndEpsi} we have the following theorem:
\begin{theorem}\label{Coro:LDGMwithUniformBound}
    Let $\epsilon >2\epsilon^* $, and a BMS channel with capacity $C$ be given, and let  $\beta$ be chosen such that $1 > \frac{\beta}{0.5} > \frac{1 + 2\epsilon^*}{1 + \epsilon}$.
    For any $R < C$, there exists a sequence of codes corresponding to  $G'$ with rate $R$, generated by applying the splitting algorithm to $G= G_2^{\otimes n}\otimes I_{n'}$,
    with the following properties:
    \begin{enumerate}
        \item 
        The error probability is upper bounded by $\exp_2{(-(\log N') ^{ 2\beta} )} $.
        \item
        The  weight of each column of the GM is upper bounded by $(\log N')^{1+\epsilon}$.
    \end{enumerate}
\end{theorem}

% \color{blue}
\subsection{Generalized LDGM Construction} \label{subsec:constructionA}

The coding scheme given in Section \ref{sec:ConstructionNewApproach} can be generalized to other sequences of codes as follows. 

\textbf{Construction A}:~
Let $\cC$ be an $(N, R)$ binary block code and $n'$ be an integer. 
Let $K = NR$ and $\bx(\bu) \in \cC$ denote the codeword for the $K$-tuple $\bu \in \bit^K$ in $\cC$.
We define an $(n'N, R)$ binary block code $\cC'$ through the following mapping: given 
an $(n'K)$-tuple $[\bu_1 \ \bu_2 \ \ldots \ \bu_{n'} ]$, $\bu_i \in \bit^K$ for $1 \leq i \leq n'$, $\cC'$ maps it to the codeword $[\bx(\bu_1) \ \bx(\bu_2) \ \ldots\  \bx(\bu_{n'})]$ of length $n'N$.

\begin{figure}[ht]
     \centering
        \includegraphics[width=0.95\textwidth]{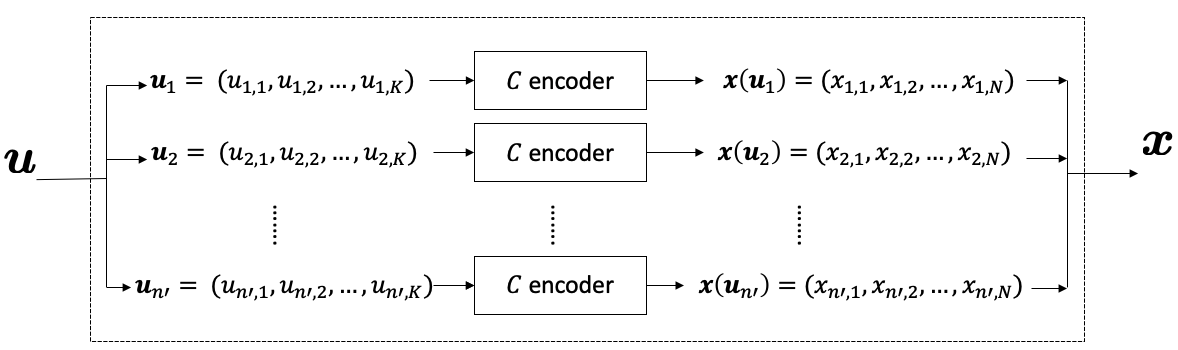}
         \caption{Input-output relationship for encoder for $C'$}
        \label{fig:ConstructionA}
\end{figure}

We denote the block length of $\cC'$ by $N' = n'N$ and the dimension by $K' = n'K$. The construction is illustrated in Figure \ref{fig:ConstructionA}.
We may bound the probability of error of $\cC'$ through that of $\cC$ in the following proposition. 
\begin{proposition}
The probability of error of the code $\cC'$, constructed based on code $\cC$ as described by Construction A, denoted by $P_e(\cC')$, over a BMS channel, can be bounded by 
\begin{equation*}
   P_e(\cC') \leq n' P_e(\cC),
\end{equation*} where $P_e(\cC)$ denotes the block error probability of the code $\cC$. 
\end{proposition}
\begin{proof}
The claim follows from union bound and memorylessness of the channel. 
\end{proof}

For a BMS channel with capacity $C$, we construct two sequences of codes with rates approaching $C$ from below as the blocklength grows, based on Construction A with the RLE and polar code, in Subsections \ref{subsec:RLE} and \ref{subsec:polar}.  Their asymptotic performance in terms of the gaps to capacity, block error probabilities, decoding complexities, and the generator matrix sparsity, are compared therein. 

\subsection{Random Linear Code-based Construction}\label{subsec:RLE}
Let a BMS channel $W$ with capacity $C = I(W)$ be given.
For RLE, as described in \ref{subsec:prelim:moderate},  choosing $\epsilon_N = N^{-\alpha}$ for some fixed $\alpha \in (0, \frac12)$, Theorem \ref{Thm:RCEmoderate} and Remark \ref{rem:RLEmoderate} show the existence of a sequence of linear codes with gap to capacity $\abs{I(W) - R_N} = O(N^{-\alpha} )$ and probability of error $P_{e,N} \leq \exp_2{(- \frac{1}{2\sigma(W)^2} N^{1-2\alpha} )}$ for large $N$. 
Assume we choose $n' = \exp_2{( \frac{1}{2\sigma(W)^2} N^{1-2\alpha} (1 - \delta) )}$ for some  $ \delta \in (0, 1)$ and expand the sequence of codes $\mathset{C_N}$ to a sequence of codes $\mathset{C'_{N', lin}}$ with block lengths $N' = n' N = \exp_2{( \frac{1}{2\sigma(W)^2} N^{1-2\alpha} (1 - \delta) +  \log N )} =   2^{O (N^{1-2\alpha })} $ and rate $R'_{N'} = R_N$. The gap to capacity of ${C'_{N', lin}}$ is 
\begin{equation*}
    {I(W) - R'_{N'}}=  O(N^{-\alpha} ) = O( (\log N')^{\frac{-\alpha}{1-2\alpha}} ), 
\end{equation*} since $\log N' = O (N^{1-2\alpha })$.
The maximal probability of error can be bounded, via union bound, as 
\begin{align*}
    P_e(C'_{N', lin}, W) &\leq n' \times P_{e,N} 
    \leq \exp_2{( - O(N^{1-2\alpha}) )} = 2^{- O( \log N' ) }  
\end{align*}

The decoding complexity of $C'_{N', lin}$, denoted by $Comp(C'_{N', lin})$, is upper bounded by 
\begin{align*}
    Comp(C'_{N', lin}) &\leq n' \times O(2^{NR} ) = 2^{O(N)} = \exp_2{(O( (\log N')^{\frac{1}{1-2\alpha}} ) )}
\end{align*}
The (arithmetic) mean column weight of the generator matrix for $C'_{N', lin}$, written as $W_{col,avg}$, is upper bounded by 
\begin{equation*}
    W_{col,avg} = N/2 = O( (\log N')^{\frac{1}{1-2\alpha}} )
\end{equation*}

\begin{proposition}
\label{prop:RLEextended}
By applying Construction A on the RLE with $n' = \exp_2{( \frac{1}{2\sigma(W)^2} N^{1-2\alpha} (1 - \delta) )}$, for some fixed $\alpha \in (0, \frac12)$ and $\delta \in (0, 1)$, there exists a sequence  of binary linear codes $\mathset{C'_{N',lin}}$ with block length $N'$ for which the following scaling behaviour holds
\begin{itemize}
    \item The gap to capacity ${I(W) - R'_{N'}} = O( (\log N')^{\frac{-\alpha}{1-2\alpha}} )$,
    \item The maximal probability of error $    P_e(C'_{N', lin}, W) = 2^{- O( \log N' ) } $,
    \item The decoding time complexity $    Comp(C'_{N', lin}) \leq  \exp_2{(O( (\log N')^{\frac{1}{1-2\alpha}} ) )}$, and 
    \item The average column weight of the generator matrix
    $    W_{col,avg} = O( (\log N')^{\frac{1}{1-2\alpha}} ).$
\end{itemize}
\end{proposition}
% \color{black}    

% \color{blue}
\subsection{Polar Code-based Construction}\label{subsec:polar}

Let $\lambda = \frac{1 - \gamma}{\mu} \in (0, \frac{1}{1+ \mu})$. For the polar code based on kernel $G_2$, Theorem \ref{Thm:PolarModerate} can be stated in terms of $\lambda$ as follows:
\begin{align*}
    P_e &\leq N \cdot \exp_2{ (-N^{ (1- \lambda \mu) \cdot h_2^{-1} ( 1- \frac{\lambda}{ 1- \lambda \mu}) } )},\\
    I(W) - R_N &= O(N^{-\lambda} )
\end{align*}
Assume we choose $n' = \exp_2{( (1-\delta)  \cdot N^{(1-\lambda \mu)  \cdot h_2^{-1} ( 1- \frac{\lambda}{ 1- \lambda \mu}) } )}$ for some  $ \delta \in (0, 1)$ and construct a code sequence $\mathset{C'_{N',pol}}$ using Construction A. 
The block length of $C'_{N',pol}$ is 
\[N' = n'N = \exp_2{( (1-\delta)  N^{(1-\lambda \mu)  \cdot h_2^{-1} ( 1- \frac{\lambda}{ 1- \lambda \mu})} + \log N )} 
=   \exp_2{( O( N^{(1-\lambda \mu)  \cdot h_2^{-1} ( 1- \frac{\lambda}{ 1- \lambda \mu}) } ))}.\]
Noting that the rate $R'_{N'}$  of the code $C'_{N',pol}$ is equal to $R_N$, the gap to capacity of $C'_{N',pol}$ is 
\begin{equation*}
    {I(W) - R'_{N'} }=  O(N^{-\lambda} ) = O( (\log N')^{\frac{-\lambda}{(1-\lambda \mu)  \cdot h_2^{-1} ( 1- \frac{\lambda}{ 1- \lambda \mu}) } } ).
\end{equation*}
The block error probability can be bounded, via union bound, as 
\begin{align*}
    P_e(C'_{N',pol}, W) \leq n' \times P_{e} 
    = 2^{-O(\log N') } 
    =     1/ poly(N')
\end{align*}
The decoding complexity of $C'_{N',pol}$ is upper bounded by 
\begin{equation*}
    Comp(C'_{N', pol}) \leq n' \cdot N\log N  = \exp_2{( (1-\delta) N^{(1-\lambda \mu)  \cdot h_2^{-1} ( 1- \frac{\lambda}{ 1- \lambda \mu}) }  + \log N + \log log N )} = (N')^{1+ o(1)}.
\end{equation*}
The mean column weight of the generator matrix of $C'_{N',pol}$ is upper bounded by 
\begin{equation*}
   W_{col,avg} \leq  \frac{3^n}{N} = N^{\log 3 -1} = O( (\log N')^{ \frac{0.585}{(1-\lambda \mu)  \cdot h_2^{-1} ( 1- \frac{\lambda}{ 1- \lambda \mu})}} ),
\end{equation*}
where $n = \log N$ is the number of layers in the factor graph representation for the encoder for polar code with block length $N$.
\begin{proposition}
\label{prop:PolarExtended}
Let $n' = \exp_2{( (1-\delta)  \cdot N^{(1-\lambda \mu)  \cdot h_2^{-1} ( 1- \frac{\lambda}{ 1- \lambda \mu}) } )}$ for some  $ \delta \in (0, 1)$ and $\lambda \in (0, \frac{1}{1+ \mu})$. 
There exists a sequence  of polar-based binary linear codes $\mathset{C'_{N',pol}}$ with block length $N'$ for which the following scaling behaviour holds
\begin{itemize}
    \item The gap to capacity ${I(W) - R'_{N'}} = O( (\log N')^{\frac{-\lambda}{(1-\lambda \mu)  \cdot h_2^{-1} ( 1- \frac{\lambda}{ 1- \lambda \mu}) } } )$,
    \item The block probability of error $    P_e(C'_{N', pol}, W) = 2^{- O( \log N' ) } $,
    \item The decoding time complexity $    Comp(C'_{N', pol}) =(N')^{1+ o(1)}$, and 
    \item The average column weight of the generator matrix
    $    W_{col,avg} =  O( (\log N')^{ \frac{0.585}{(1-\lambda \mu)  \cdot h_2^{-1} ( 1- \frac{\lambda}{ 1- \lambda \mu})}} ).$
\end{itemize}
\end{proposition}

\subsection{Comparison}\label{subsec:Compare}
 We collect the results in Sections \ref{subsec:RLE} and \ref{subsec:polar} in Table \ref{table:RLEvsPolar}. Note that the error probabilities do not depend on the choices of $\alpha$ and $\lambda$ in both cases. Note also that  $\mathset{C'_{N', pol}}$ is a sequence of capacity-achieving code with almost linear decoding complexity. 

\begin{table}[]
    \centering
    \begin{adjustbox}{width=0.8\textwidth}
    \begin{tabular}{ p{ 2cm} | c c }
      & $ C'_{N', lin}$ & $C'_{N', pol}$ \\ 
     \hline
     $I(W) - R$ & $O( (\log N')^{\frac{-\alpha}{1-2\alpha}} )$ & $ O( (\log N')^{\frac{-\lambda}{(1-\lambda \mu)  \cdot h_2^{-1} ( 1- \frac{\lambda}{ 1- \lambda \mu}) } } )$ \\  
     $P_e$ & $2^{- O( \log N' ) }$ & $2^{-O(\log N') }$    \\
     $Comp$ & $2^{O( (\log N')^{\frac{1}{1-2\alpha}} ) }$ & $(N')^{1+ o(1)} = (2^{ \log N' })^{1+o(1)}$    \\
     $W_{col,avg}$ & $ O( (\log N')^{\frac{1}{1-2\alpha}} )$ & $O( (\log N')^{ \frac{0.585}{(1-\lambda \mu)  \cdot h_2^{-1} ( 1- \frac{\lambda}{ 1- \lambda \mu})}} )$
    \end{tabular}
    \end{adjustbox}
    \caption{Results from Propositions \ref{prop:RLEextended} and \ref{prop:PolarExtended}}
    \label{table:RLEvsPolar}
\vspace{-7mm}
\end{table}

We may compare the two constructions by equating the scaling performance of their gaps to capacity and numerically evaluate the exponents of the decoding complexity and $W_{col,avg}$. 
Specifically, we define variables as follows. 
\begin{definition}
We define the \textit{exponent terms} for code sequences $\mathset{ C'_{N', lin}}$ and  $ \mathset{C'_{N', pol} }$ as:
 \begin{align*}
    &exp(gap, C'_{N', lin}) =\frac{\alpha }{1-2\alpha} \hspace{5mm}
    &&exp(gap, C'_{N', pol}) =\frac{\lambda}{(1-\lambda \mu)  \cdot h_2^{-1} ( 1- \frac{\lambda}{ 1- \lambda \mu}) }\\
     &exp(Comp, C'_{N', lin}) =\frac{1}{1-2\alpha} 
     &&exp(Comp, C'_{N', pol}) = 1 \\
     &exp(W_{col,avg}, C'_{N', lin}) = \frac{1}{1-2\alpha} 
     &&exp(W_{col,avg}, C'_{N', pol}) = \frac{0.585}{(1-\lambda \mu)  \cdot h_2^{-1} ( 1- \frac{\lambda}{ 1- \lambda \mu})}
 \end{align*}
\end{definition}
Using this definition, the gap to capacity, decoding time complexity, and average GM column weight for a code sequence $\mathset{C_{N'}}$ scales exponentially in $\log N'$, $2^{\log N'}$, and $\log N'$ with exponents given by 
$exp(gap, C_{N'}), exp(Comp, C_{N'})$, and $ exp(W_{col,avg}, C_{N'})$, respectively.
 
Note that the gaps to capacity scale exponentially in $\log N'$ with exponents $- exp(gap, C'_{N', lin})$ and $-exp(gap, C'_{N', pol})$ for $\mathset{ C'_{N', lin}}$ and  $ \mathset{C'_{N', pol} }$, respectively.
Assume $\mu = 3.579$, for each $\lambda \in (0, \frac{1}{1+ \mu})$, we may find $\alpha \in (0, \frac12)$ such that  $exp(gap, C'_{N', lin})  = exp(gap, C'_{N', pol})$. The relationship between the gap to capacity and decoding complexity is shown in Figure \ref{fig:comp_lower}. 
One may observe that the decoding complexity for the polar-based $C'_{N', pol}$ is independent of the choice of $\lambda$ and hence the gap-to-capacity exponent $exp(gap, C'_{N', pol})$. In fact, it remains almost linear in the blocklength $N'$. 
On the other hand, the exponent of the decoding complexity for the RLE-based $C'_{N', lin}$ code sequence,  $exp(Comp, C'_{N', lin})$, grows linear as the gap-to-capacity exponent $exp(gap, C'_{N, pol})$ increases. 
The relationship between the average column weight of the generator matrix, $W_{col,max}$, and the gap-to-capacity exponent $exp(gap)$ is plotted in Figure \ref{fig:sparse_lower}. It may be observed that the code $C'_{N', lin}$ scales exponentially in $\log N'$ with a smaller exponent that $C'_{N', pol}$ over the entire range. 
\begin{figure}[ht]
     \centering
     \begin{subfigure}[b]{0.48\textwidth}
         \centering
         \includegraphics[width=\textwidth]{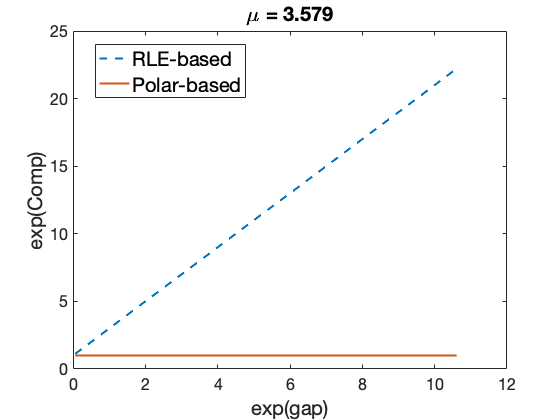}
         \caption{}
         \label{fig:comp_lower}
     \end{subfigure}
     \hfill
     \begin{subfigure}[b]{0.48\textwidth}
         \centering
         \includegraphics[width=\textwidth]{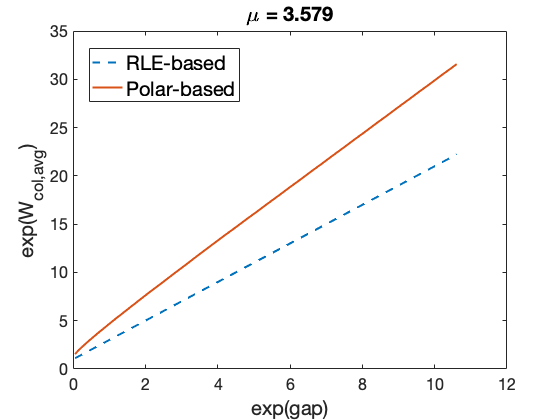}
         \caption{}
         \label{fig:sparse_lower}
     \end{subfigure}
        \caption{Results with $\mu = 3.579$.}
        \label{fig:compare_lower_mu}
\vspace{-7mm}
\end{figure}

We may also plot the decoding complexity exponent and the mean column weight exponent versus $exp(gap)$ when assuming the scaling exponent of the channel $W$ is $\mu = 4.714$. The results are plotted in Figures \ref{fig:comp_upper} and \ref{fig:sparse_upper}.

\begin{figure}[ht]
     \centering
     \begin{subfigure}[b]{0.48\textwidth}
         \centering
         \includegraphics[width=\textwidth]{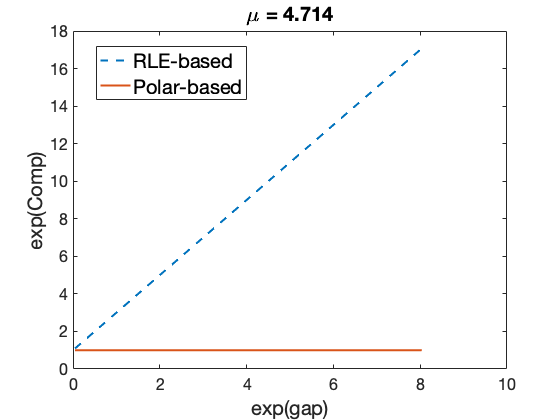}
         \caption{}
         \label{fig:comp_upper}
     \end{subfigure}
     \hfill
     \begin{subfigure}[b]{0.48\textwidth}
         \centering
         \includegraphics[width=\textwidth]{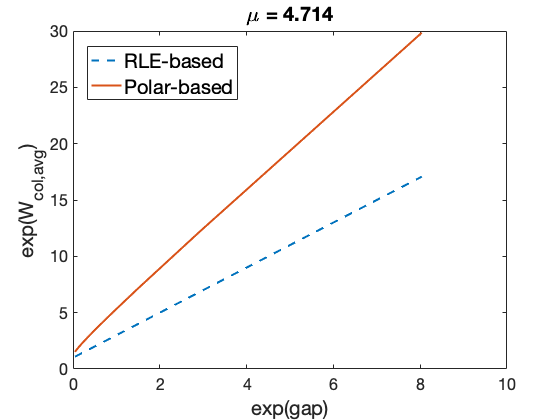}
         \caption{}
         \label{fig:sparse_upper}
     \end{subfigure}
        \caption{Results with $\mu = 4.714$.}
        \label{fig:compare_upper_mu}
\vspace{-7mm}
\end{figure}

\begin{remark}
For the RLE-based construction of the code sequence, if one considers $\alpha \rightarrow 0$, or equivalently the conventional error exponent regime where the rate $R < C$ is fixed, the error probability of the $m$-th message can be upper bounded by 
\begin{equation*}
    P_{e,m} < 4 \exp{[-N E_r(R)]}, \mbox{ for each } m, 1\leq m \leq M = 2^{NR}. 
\end{equation*}
Following the steps as in Section \ref{subsec:RLE}, we may choose
 $n' = 2^{  E_r(R) N (1 - \delta) }$ for some  $ \delta \in (0, 1)$ and hence the sequence of codes $\mathset{C'_{N', lin}}$ with block lengths $N' = n' N = 2^{ E_r(R) N (1 - \delta) +  \log N} =   2^{O (N)} $. 
The maximal probability of error is upper bounded by
\begin{align*}
    P_e(C'_{N', lin}, W) < 4 \exp{[-N E_r(R)]} = 2^{- O(N) } = 2^{- O( \log N' ) } = 1/ poly(N').
\end{align*}

The decoding complexity of $C'_{N', lin}$ is upper bounded by 
\begin{align*}
    Comp(C'_{N', lin}) \leq n' O(2^{NR} ) = 2^{ E_r(R) N (1 - \delta) + NR}  \approx (N')^{\frac{ E_r(R) (1 - \delta) + R}{E_r(R) (1 - \delta) }}.
\end{align*}
The column weight of the generator matrix for $C'_{N', lin}$, written as $W_{max}(C'_{N', lin})$, is upper bounded by 
\begin{equation*}
    W_{max}(C'_{N', lin}) \leq  N \approx \frac{\log N'}{E_r(R) (1 - \delta)}  = O( \log N' ).
\end{equation*}

We note first that the scaling performance of $W_{max}(C'_{N', lin})$ as a logarithmic function in $N'$ means that for any BMS channel, one may construct a sequence of capacity-achieving LDGM codes that meets the sparsity benchmark. One caveat is that the hidden factor in the $O( \log N' )$ bound for GM column weight is inversely proportional to the error exponent, and may be arbitrarily large for $R$ sufficiently close to $C$.
Second, note that $E_r(R) \rightarrow 0 $ as $R \rightarrow C$. Hence if the code rate $R$ is chosen close to $C$, the decoding complexity of  $C'_{N', lin}$ scales polynomially in $N'$ with a large degree, whereas that of  $C'_{N', pol}$ scales as $ (N')^{1 + o(1) }$. 
The observations above motivate us to study the optimal GM sparsity for LDGM codes with efficient decoders. 

\end{remark}
% \color{black}

\section{Sparse LGDM Codes with Low-complexity Decoding}\label{sec:LDGMwithDEC}

\subsection{Decoder-Respecting Splitting Algorithm} \label{sub:DRSalg}

In this section, we consider another splitting algorithm, referred to as decoder-respecting splitting (DRS)  algorithm. This algorithm enables a low-complexity SC decoder based on likelihood ratios that can be calculated with a recursive algorithm. 
The main idea of the algorithm is to construct a generator matrix that can be realized with an encoding pattern similar to conventional polar codes such that the column weights of the matrix associated with the diagram are bounded above by a given threshold $w_{u.b.}$. 

\begin{algorithm} % enter the algorithm environment
\caption{DRS algorithm} % give the algorithm a caption
\label{alg1} % and a label for \ref{} commands later in the document
\textbf{Input:} $w_{u.b.}\in \N$,  $n\in \N$, $v\in \bit^{2^n \times 1}$    \\
\textbf{Output:} \textproc{DRS-Split}($w_{u.b.}, v$)

\begin{algorithmic}[1]
\Function{DRS-Split}{$w_{u.b.}, \mathbf{x}$}
    \State $k\gets \mbox{length}(\mathbf{x}) /2$
    \State $\mathbf{x}_h \gets (x_1,\ldots, x_k)^t$, $\mathbf{x}_t \gets (x_{k+1},\ldots, x_{2k})^t$

    \If{$w_H(\mathbf{x})> w_{u.b.}$}
        \State $Y_h \gets$ \textproc{DRS-Split}($w_{u.b.}, \mathbf{x}_h$)
        \State $Y_t \gets$ \textproc{DRS-Split}($w_{u.b.}, \mathbf{x}_t$)
        \If{$\mathbf{x}_h= \zero_{k\times 1} $} %\vspace{1mm}
            \State \Return $\bigcup\limits_{y\in Y_t} \{(\zero_{1\times k}, y^t)^t \}$ 
        \ElsIf{$\mathbf{x}_t= \zero_{k\times 1} $} %\vspace{1mm}
            \State \Return $\bigcup\limits_{y\in Y_h} \{ (y^t, \zero_{1\times k})^t \}$ 
        \Else 
            \State \Return $\bigcup\limits_{y\in Y_t} \{ (\zero_{1\times k}, y^t)^t\} \cup  \bigcup\limits_{y\in Y_h} \{(y^t, \zero_{1\times k})^t \}$ 
        \EndIf
    \ElsIf{$w_H(\mathbf{x})=0 $}
        \State \Return \{\}
    \Else
        \State \Return $\mathset{ \mathbf{x}}$
    \EndIf
\EndFunction
\end{algorithmic}
\end{algorithm}
The core of the algorithm is the \textproc{DRS-Split} function. 
When the weight of the input vector $\mathbf{x}$ is larger than the threshold, it splits the vector in half into vectors $\mathbf{x}_h$ and $\mathbf{x}_t$, and recursively finds two sets, $Y_h$ and $Y_t$, composed of vectors with the length halved compared to the length of $\mathbf{x}$. 
The vectors are then appended to the length of $\mathbf{x}$, which collectively form the output of the function.
We note that the weights of vectors in  $Y_h$ and $Y_t$ are respectively upper bounded by the weights of $\mathbf{x}_h$ and $\mathbf{x}_t$, both of which are bounded by $k= \abs{ \mathbf{x}_h}=\abs{ \mathbf{x}_t}$, and that the value of $k$ is halved each iteration.
Hence, the function is guaranteed to terminate as long as the threshold is a positive integer.

We use a simple example to illustrate the algorithm. Let $n=3 $, $v= [0,0,0,0, 1,1,1, 1]^t$ and $w_{u.b.}= 2$. Since the weight of $v$ exceeds the threshold, it is first split into 
$\mathbf{x}_h=[0,0,0,0]^t$ and $\mathbf{x}_t= [ 1,1,1, 1]^t$.
Since $\mathbf{x}_h$ is an all-zero vector, $Y_h$ is an empty set according to line $14$ to $15$. 
To compute $Y_t=$\textproc{DRS-Split}($2, [ 1,1,1, 1]^t$), the function splits the input into half again, thereby obtaining $\mathbf{x}_h'=[1,1]^t$ and $\mathbf{x}_t'= [ 1,1]^t$.
The corresponding $Y_h'$ and $Y_t'$ are then both $\mathset{[1,1]^t}$ and, hence, we have $Y_t= \mathset{ [0,0, 1,1]^t}\cup \mathset{ [1,1,0,0]^t}= \mathset{ [0,0, 1,1]^t, [1,1,0,0]^t}$.  
Since $\mathbf{x}_h= \zero_{4\times 1}$, the function proceeds to lines $7$ and $8$, and returns $\mathset{ [0,0,0,0,0,0, 1,1]^t, [0,0,0,0,1,1,0,0]^t}$.

We note that splitting a column using the DRS algorithm may have more output columns than using the splitting algorithm introduced in Section \ref{sec:ConstructionNewApproach}. For example, let a column vector $u = [1,0,1, 1 , 1,0 ,1,1]$ and the threshold $w_{u.b.} =2$ be given.  Applying the DRS algorithm on $u$ gives 4 new vectors, while the splitting algorithm in Section \ref{sec:ConstructionNewApproach} gives 3 new vectors.

In order to analyze the effect of the DRS algorithm on the matrix $G_2^{\otimes n}$, 
% we give a more general version of lemma \ref{lem:LR=RL}, 
we show that the size of the algorithm output does not depend on the order of a sequence of Kronecker product operations.
Suppose that the Kronecker product operations with the vector $[1 ,\; 1]^t $ for $n_1$ times and  with the vector $[0 ,\; 1]^t $ for $n_2$ times are applied on a vector $v$, where $n =n_1+n_2$ and the order of the operations is specified by a sequence $(s_1, s_2, \ldots, s_{n} ) \in \mathset{-,+}^{n}$ with $\abs{ \mathset{ i:s_i = - } } = n_1$ and $\abs{ \mathset{ i:s_i = + } } = n_2$. 
Also, let $v^{(i)}$ denote the output of applying the first $i$ Kronecker product operations on $v$. It is defined by the following recursive relation: 
\begin{equation} \label{eq:vi_recursion}
    v^{(i)} = \begin{cases}
    v^{(i-1)} \otimes [1, 1]^t,\; \mbox{if } s_i= -, \\
    v^{(i-1)} \otimes [0, 1]^t,\; \mbox{if } s_i= + ,
    \end{cases}
\end{equation} for $i \geq 1 $ and the initial condition $v^{(0)}= v$. We use  $v^{(s_1, s_2, \ldots s_i)}$ instead of $v^{(i)}$ when the sequence is needed for clarity.

The following lemma shows that any two vectors of the form $v^{(s_1, s_2, \ldots s_{n})}$ will be split into the same number of columns under the DRS algorithm as long as  the sequences associated with them contain the same number of $-$ and $+$ signs.

\begin{lemma}\label{lem:rateLossIndepOfOrder}
    Let $n= n_1+n_2$ and  $(s_1, s_2, \ldots, s_{n} ) \in \mathset{-,+}^{n}$ be a sequence with $n_1$ minus signs and $n_2 $ plus signs.
    Let $v^{(n)}$ be the vector defined by a vector $v$ and the sequence $(s_1, s_2, \ldots, s_{n} )$  through equation \eqref{eq:vi_recursion}. 
    Then the size of the DRS algorithm output for $v^{(n)}$ depends only on the values $n_1$ and $n_2$.  
\end{lemma}

    Let a $K\times N$ matrix $M= [\mathbf{u}_1,\mathbf{u}_2,\ldots, \mathbf{u}_N] $ and a threshold $w_{u.b.}$ be given. 
    Suppose that the DRS algorithm is applied to each column in $M$ and the sum of the sizes of the outputs is $N(1+\gamma)$. Then 
    \textproc{DRS}$(M)$ is defined as the $K \times N(1+\gamma)$ matrix consisting of all the vectors in the outputs (with repetition). 
    
    Same as in Section \ref{subsection:mainL2}, we study the effect of the DRS algorithm in terms of the multiplicative rate loss, i.e., $1+\gamma$. In particular, the following proposition shows an appropriate choice of $w_{u.b.}$ guarantees the existence of a sparse polar-based GM with vanishing $\gamma$.
    
\begin{proposition}\label{prop:DRSrateloss}
    Let the columns of  $G_2^{\otimes n}$ be the inputs for the DRS algorithm and \textproc{DRS}$(G_2^{\otimes n})$ be the $N \times N(1+\gamma)$ matrix consisting of the outputs. The term $\gamma$ vanishes as $n$ goes to infinity for any $w_{u.b.} = 2^{n\lambda}$ with $\lambda > \lambda^* \triangleq h_b(\frac{2}{3}) -\frac{1}{3} \approx 0.585$.
\end{proposition}

\begin{remark}\label{rem:wubThresholdCompare}
We may compare the column weight thresholds in Proposition \ref{prop:DRSrateloss} and Theorem \ref{Thm:ratelossAndEpsi'} in Section \ref{subsection:mainL2}.   
The column weight threshold $w_{u.b.}$ is given by ${N}^{  \frac{1}{2}+\epsilon' }$ in Section \ref{subsection:mainL2}. Proposition \ref{Thm:ratelossAndEpsi'} states that the term $\gamma$ is vanishing as long as $\epsilon' > \epsilon^*$, where $ \epsilon^* =  \hspace{1mm} \log3 -\frac32
\approx 0.085$.
In Proposition \ref{prop:DRSrateloss} the threshold $w_{u.b.}$ is given by $2^{n\lambda}= N^{\lambda}$, and the term $\gamma$ goes to $0$ as long as $\lambda \geq \lambda^* \triangleq h_b(\frac{2}{3}) -\frac{1}{3}= \log3 -1 \approx 0.585$. 
Note that the conditions for the column weight thresholds for vanishing $\gamma$ in both cases have the same exponent of $N$, i.e., $\frac{1}{2}+\epsilon^* = \lambda^* $.
Hence the DRS algorithm does not incur extra rate loss compared to the splitting algorithm introduced in Section \ref{sec:ConstructionNewApproach} asymptotically.
\end{remark} 

\subsection{Low-complexity Decoder for $G_2$-based LDGM: BEC}
\label{subsection:LCdec_BEC}

In this section, we show two results for the code corresponding to  \textproc{DRS}($G_2^{\otimes n}$) over the BEC. 
First, we propose a low-complexity suboptimal decoder for the code corresponding to \textproc{DRS}($G_2^{\otimes n}$).
Second, with the low-complexity suboptimal decoder, the code corresponding to  \textproc{DRS}($G_2^{\otimes n}$) is capacity-achieving for suitable column weight threshold. 

It is known that when the channel transformation with kernel $G_2$ is applied to two BECs, the two new bit-channels are also BECs. 
Specifically, for two binary erasure channels $W_1 $ and $W_2$ with erasure probabilities $\epsilon_1 $ and  $\epsilon_2$, respectively, the polarized bit-channels $W^-(W_1, W_2) $ and $W^+(W_1, W_2) $ are binary erasure channels with erasure probabilities $\epsilon_1 +\epsilon_2 - \epsilon_1  \epsilon_2  $ and $\epsilon_1  \epsilon_2  $, respectively.

The mutual information $I(\cdot)$ and Bhattacharyya parameter $Z(\cdot)$ of a BEC $W$ with erasure probability $\epsilon $ are given by: 
    % \begin{align*}
        $I(W)= 1-\epsilon, 
        Z(W)= \epsilon.$
    % \end{align*}
For a sequence $(s_1, s_2, \ldots, s_n) \in \Ftwo^n$, the function  $\mathsf{Bi2De}(s_1, s_2, \ldots, s_n)$ returns the decimal value of the binary string in which a minus sign for $s_i$ is regarded as a $0$ and a plus sign as a $1$, for each $i$. For example, $\mathsf{Bi2De}(-, +, +)= (011)_2 =3$.
Let $G$ denote $G_2^{\otimes n}$ and $G'$ denote  \textproc{DRS}($G_2^{\otimes n}$), and let $Z_{G}^{(s_1s_2\ldots s_n)}$ denote the Bhattacharyya parameter of the bit-channel $W^{s_1s_2\ldots s_n}$, which is defined to be $W_N^{(\mathsf{Bi2De}(s_1, s_2, \ldots, s_n)+1)}$ as in \cite[page 3]{arikan2009channel}.
The term $Z_{G'}^{(s_1s_2\ldots s_n)}$ denotes the Bhattacharyya parameter of the bit-channel observed by the source bit of the same index corresponding to $G'$. 

The following lemma shows that the bit-channel observed by each source bit is better in terms of the Bhattacharyya parameter when $G'$ is the generator matrix instead of $G$.

\begin{lemma}\label{lem:DRS_Z_BEC}
Let $w_{u.b.}$ and $n$ be given, and let $G$ denote $G_2^{\otimes n}$ and $G'$ denote  \textproc{DRS}($G_2^{\otimes n}$).
The following is true for any $(s_1, s_2, \ldots, s_n)\in \mathset{-,+}^n$:
    \[
    Z_{G'}^{(s_1s_2\ldots s_n)} \leq Z_{G}^{(s_1s_2\ldots s_n)}.
    \]
\end{lemma}

We are ready to show the existence of a sequence of capacity-achieving codes over the BEC with GMs where the column weights are bounded by a polynomial in the blocklength, and that the block error probability under a low  complexity decoder vanishes as the $n$ grows large.
\begin{proposition}
     \label{Thm:BEClowcomplDEC}
    Let $\beta < E(G_2)=0.5 $, $\lambda >\lambda^*=h_b(\frac{2}{3}) -\frac{1}{3}$ and a BEC $W$ with capacity $C$ be given. There exists a sequence of codes corresponding to  \textproc{DRS}($G_2^{\otimes n}$)
    with the following properties for all sufficiently large $n$:
    \begin{enumerate}
        \item 
        The error probability is upper bounded by $2^{-N^{\beta}}$, where $N= 2^n$.
        \item
        The Hamming weight of each column of the GM is upper bounded by $N^\lambda$.
        \item
        The rate approaches $C$ as $n$ grows large.
        \item 
        The codes can be decoded by a successive-cancellation decoding scheme with complexity $O(N\log N)$.
    \end{enumerate}
    \end{proposition}
    \proof
Let the threshold for DRS algorithm be $w_{u.b.} = 2^{n\lambda}$,  $G$ denote $G_2^{\otimes n}$, and $G'$ denote \textproc{DRS}($G_2^{\otimes n}$) in this proof. 
We prove the four claims in order. 
First, Lemma \ref{lem:DRS_Z_BEC} shows that for a given $n$ and any  $t>0$, the following is true:
\be{eq:bhataPara_G'}
	\mathset{\mathbf{s}\in \mathset{-,+}^n :Z_{G}^{\mathbf{s}} \leq t  }
	\subseteq
	\mathset{
		\mathbf{s}\in \mathset{-,+}^n :Z_{G'}^{\mathbf{s}} \leq t  
	}.
\ee
Using \cite[Theorem 2]{arikan2009channel}, for any $\beta< \frac{1}{2}$, we have
\be{eq:rate_to_C}
\liminf_{n \to \infty} \frac{1}{N}
\abs{\mathset{\mathbf{s}\in \mathset{-,+}^n :Z_{G}^{\mathbf{s}} \leq 2^{-N^\beta}  }} = I(W) =C
\ee

Let $\mathbfsl{S_G} $ and $\mathbfsl{S_{G'}} $ denote the sets of the sequences $\mathbf{s} \in \mathset{-,+}^n$ that satisfy $Z_{G}^{\mathbf{s}} \leq 2^{-N^\beta}$ and $Z_{G'}^{\mathbf{s}} \leq 2^{-N^\beta}$, respectively.  Equation \eqref{eq:bhataPara_G'} guarantees that $\mathbfsl{S_G} $ is a subset of $\mathbfsl{S_{G'}}$.
Assume the code corresponding to $G$ freezes the input bits observing bit-channels $W^{s_1s_2\ldots s_n}$ for all $(s_1, s_2, \ldots, s_n) \notin \mathbfsl{S_G}$.
For the code corresponding to  $G'$, we use the bit-channels with the same index as the code corresponding to  $G$, for transmission of information bits, and leave the rest as frozen. 
The probability of block error for the code corresponding to  $G'$, $P_{e, G'}$, can be bounded above, as in \cite{arikan2009channel}, by the sum of the Bhattacharyya parameters of the bit-channels for the source bits (that are not frozen), that is,
$$
P_{e, G'} \leq 
\sum_{\mathbf{s}\in \mathbfsl{S_{G}}}Z_{G'}^{\mathbf{s}} 
\leq 
\sum_{\mathbf{s}\in \mathbfsl{S_{G}}} 2^{-N^\beta}
=
\abs{\mathbfsl{S_G} } 2^{-N^\beta},
$$ where the second inequality follows because, for  $\mathbf{s}\in  \mathbfsl{S_G} $, we must have  $\mathbf{s}\in  \mathbfsl{S_{G'}}$ and thus $Z_{G'}^{\mathbf{s}}\leq 2^{-N^\beta}$.
From \eqref{eq:rate_to_C}, for all sufficiently large $n$, we have 
\begin{equation}
     P_{e, G'} \leq NC 2^{-N^\beta}.
\end{equation}
By an argument similar to that in the proof of Lemma \ref{Lemma:newG_PeBound},  for any $\beta'< \frac{1}{2}$, 
% \begin{equation}
     $P_{e, G'} \leq 2^{-N^{\beta'}} $ for all sufficiently large $n$.
% \end{equation} 

The second claim follows from the fact that the GM for the code corresponding to  $G'$ is a submatrix of $G'$, and the Hamming weight of each column of $G'$ is upper bounded by $w_{u.b.} = 2^{n\lambda}= N^\lambda$. 

The third claim is a {consequence} of Proposition \ref{prop:DRSrateloss} and Lemma \ref{lem:DRS_Z_BEC}. 
The number of information bits of the code corresponding to  $G'$ is given by 
$\abs{\mathbfsl{S_G} }$, and the length of the code is $N(1+\gamma)$. 
Hence the rate is 
\begin{equation}
    \frac{\abs{ \mathbfsl{S_G} }}{N(1+\gamma)}
 \end{equation}
Since the term $\gamma$ vanishes as $n$ grows large, we have 
\begin{equation}
\liminf_{n \to \infty} \frac{\abs{\mathbfsl{S_G}} }{N(1+\gamma)}
=
\liminf_{n \to \infty} \frac{\abs{\mathbfsl{S_G}} }{N}
= I(W) =C.    
\end{equation}

Finally, we prove the claim for the existence of a low-complexity decoder. 
Let $U_1, \ldots, U_N$ be the inputs, and $Y_1, \ldots, Y_{N_1}, Y_{N_1 +1}, \ldots,  Y_{N_1+ N_2} $ the outputs, where $N_1+N_2 = N(1+\gamma)$, as shown in Figure\,\ref{fig:G2m'_BEC_EncoderWm}. While the polar code based on $G_2^{\otimes n}$, as shown in Figure\,\ref{fig:G2m_BEC_EncoderWm}, is recursive in the encoder structure (with bit reversal permutation), the code based on $(G_2^{\otimes n})'$ is not, as the blocks $W_n^u$ and  $W_n^l$ are not necessarily equal. In fact, when there is a split at the last iteration of polarization, i.e., when one or more of the XOR operations shown in Figure\,\ref{fig:G2m'_BEC_EncoderWm} is replaced by a solid black circle, the number of inputs of the block $W_n^l$ will be larger than that of $W_n^u$. 

\begin{figure}[t]
     \centering
     \begin{subfigure}[b]{0.44\textwidth}
         \centering
         \includegraphics[trim=1.5cm 1.5cm 12.5cm 3.5cm, clip, width=\textwidth]{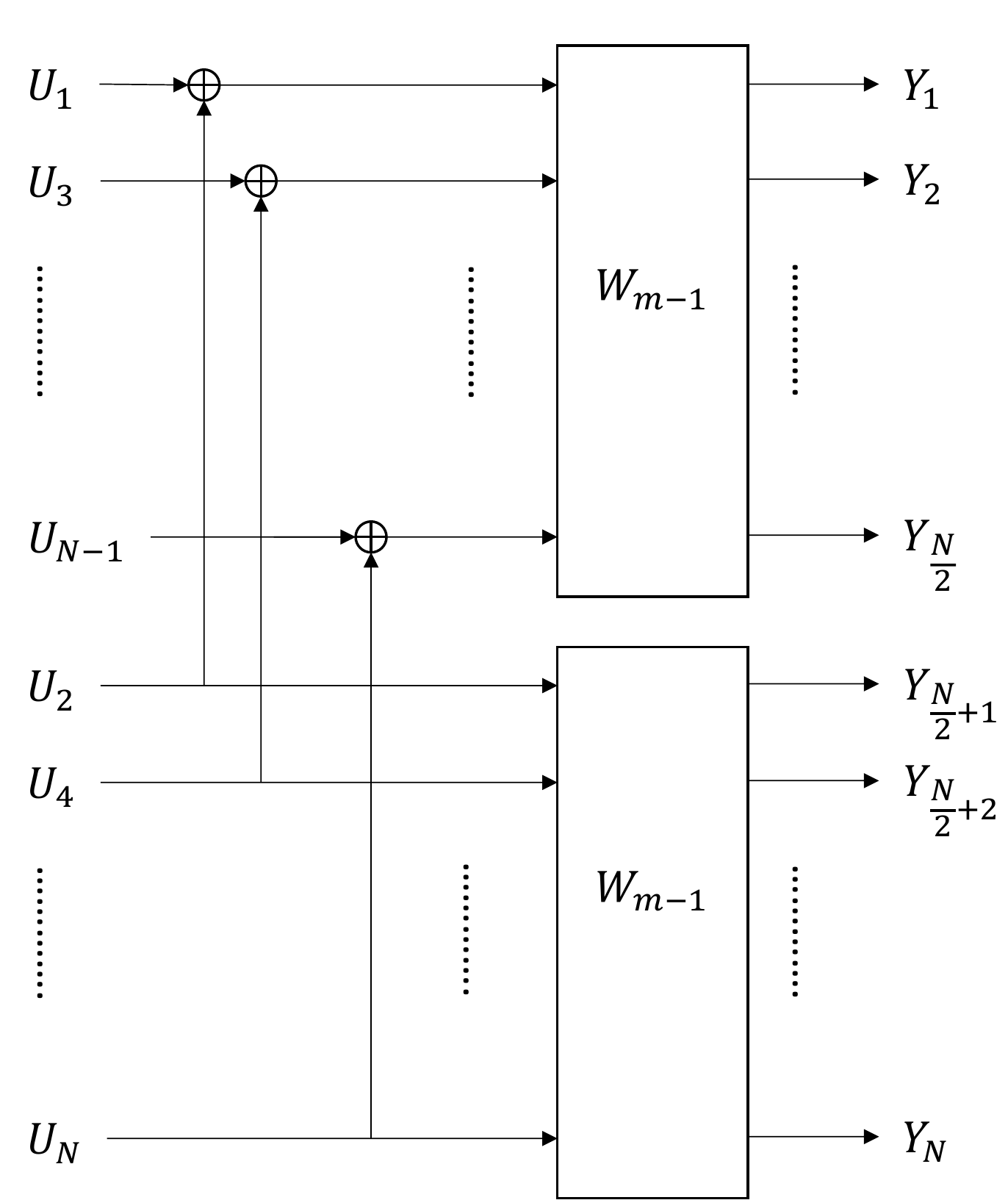}
         \caption{Encoding block for generator matrix $G_2^{\otimes m}$}
        \label{fig:G2m_BEC_EncoderWm}
     \end{subfigure}
     \hfill
    % \vspace{5mm}
     \begin{subfigure}[b]{0.45\textwidth}
         \centering
         \includegraphics[trim=2cm 1cm 10.5cm 2.5cm,  clip, width=\textwidth]{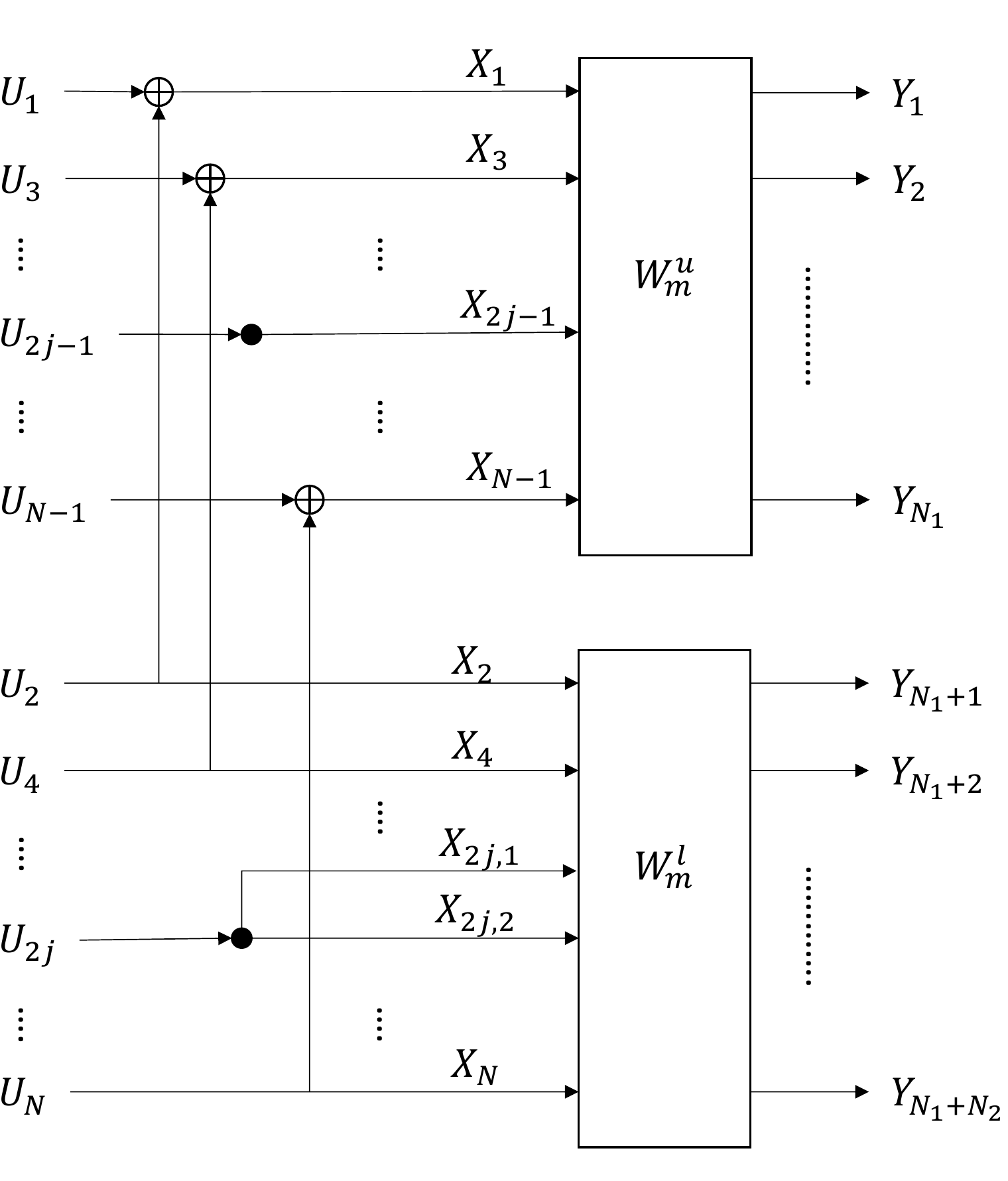}
        \caption{Encoding block for generator matrix $(G_2^{\otimes m})'$}
    \label{fig:G2m'_BEC_EncoderWm}
     \end{subfigure}
\caption{}
        \label{fig:Wm_BEC}
\vspace{-7mm}
\end{figure}

Let $\cF \subseteq \mathset{1,\ldots, N}$ be the set of the indices of the frozen bits. 
The decoder declares estimates $\hat{U}_i$ of the inputs, for $1\leq i \leq N$, sequentially by:
\begin{equation}\label{eq:UiestimateBEC}
    \hat{U_i} = \begin{cases}
u_i, &\mbox{ if } i\in \cF, \\
\psi_i(Y_1^{ N_1+N_2}, \hat{U}_1^{i-1} ,W_n) &\mbox{ if } i \notin \cF,
\end{cases}
\end{equation}
where     $\psi_i(Y_1^{ N_1+N_2}, \hat{U}_1^{i-1} ,W_n)$ can be found in following four cases, and $W_n$ denotes the encoding block shown in Figure \ref{fig:G2m'_BEC_EncoderWm}. 
Let the symbol $e$ denotes an erasure, and assume $e \oplus b = e$ for $b \in \mathset{0,1,e}$.

\begin{itemize}
    \item If $i$ is odd and $X_i = U_i \oplus U_{i+1}$, which corresponds to an unsplit XOR operation observed by $U_i$, \begin{equation}\
        \psi_i(Y_1^{ N_1+N_2}, \hat{U}_1^{i-1} ,W_n) \triangleq
    \begin{cases}
        \hat{X_i}\oplus \hat{X}_{i+1} &\mbox{ if } \hat{X_i}\neq e, \hat{X}_{i+1}\neq e, \\
        e, &\mbox{ otherwise.} 
    \end{cases}
    \end{equation}   

    \item If $i$ is odd and $X_i = U_i$, which corresponds to a split XOR operation, 
    $\psi_i(Y_1^{ N_1+N_2}, \hat{U}_1^{i-1} , W_n) \triangleq \hat{X}_i$.
    
    \item If $i$ is even and $X_{i-1} = U_i \oplus U_{i-1}$, which corresponds to an unsplit XOR operation,
    \begin{equation}
        \psi_i(Y_1^{ N_1+N_2}, \hat{U}_1^{i-1} , W_n ) \triangleq
    \begin{cases}
        \hat{X_i}, &\mbox{ if } \hat{X_i}\neq e, \hat{X}_{i-1} = e, \\&\mbox{ or if } \hat{X_i}\neq e, \hat{X}_{i-1} \neq e, \hat{X_i} = \hat{X}_{i-1} \oplus \hat{U}_{i-1},\\
         \hat{X}_{i-1} \oplus \hat{U}_{i-1}, &\mbox{ if } \hat{X_i}= e, \hat{X}_{i-1} \neq e, \hat{U}_{i-1} \neq e\\
        e, &\mbox{ otherwise.} 
    \end{cases}
    \end{equation}  
    
    \item If $i$ is even and $X_{i,1}=X_{i,2} = U_i$, which corresponds to a split XOR operation,
    \begin{equation}
        \psi_i(Y_1^{ N_1+N_2}, \hat{U}_1^{i-1} , W_n ) \triangleq
    \begin{cases}
        \hat{X}_{i,1}, &\mbox{ if } \hat{X}_{i,1}\neq e, \hat{X}_{i,2} = e, \mbox{ or if } \hat{X}_{i,1} = \hat{X}_{i,2} \neq e,\\
        \hat{X}_{i,2}, &\mbox{ if } \hat{X}_{i,1}= e, \hat{X}_{i,2} \neq e,\\
        e, &\mbox{ otherwise.} 
    \end{cases}
    \end{equation} 
\end{itemize}
The estimates $\hat{X}_1,\hat{X}_3,\ldots, \hat{X}_{N-1}$ and $\hat{X}_{2}, \hat{X}_{4}, \ldots, \hat{X}_{2j,1} , \hat{X}_{2j,2}, \ldots, \hat{X}_{N}$ are then separately found in a similar approach using the blocks $W_n^u$ and $ W_n^l$ along with the outputs $Y_1, \ldots, Y_{N_1}$ and  $Y_{N_1 +1}, \ldots,  Y_{N_1+ N_2} $, respectively. That is, we can express $\hat{X}_i$ in terms of the estimates of the input variables of the four encoding blocks with $(m-2)$ iterations of XOR operations. 

For the right-most variables, the blocks they observe are identical copies of
%single-input, single-output  blocks that are equivalent to 
the BEC $W$.  Hence the estimates of the variables, denoted as $\hat{X}_1^{(n)}, \hat{X}_2^{(n)}, \ldots, \hat{X}_{N_1+ N_2}^{(n)}$ are naturally defined by the outputs of the channels, i.e., $\hat{X}_i^{(n)} = Y_i$ for $i = 1, 2, \ldots , N_1+N_2.$

At each stage there are at most $ N_1+N_2= N(1+\gamma) = O(N)$ estimates to make, and the recursion ends in $\log(N)$ steps.
Since each estimate is obtained with constant complexity, the total decoding complexity for the code based on \textproc{DRS}($G_2^{\otimes n}$) is bounded by $O(N\log(N))$.

\endproof

We leverage the construction given in Section \ref{sec:ConstructionNewApproach} and Proposition \ref{Thm:BEClowcomplDEC} to show, in  Theorem \ref{coro:DRS_lowComp_sparse_BEC} below, the existence of a sequence of capacity-achieving codes over BECs with sparse GMs and low-complexity decoders. In particular, the upper bounds of the column weights of the GMs are equal to those of Theorem \ref{Coro:LDGMwithUniformBound}, which scales polynomially in the logarithm of the blocklength with a degree slightly larger than 1.

\begin{theorem}\label{coro:DRS_lowComp_sparse_BEC}
    Let $\beta < E(G_2)=0.5 $, $\lambda >\lambda^*$ and a BEC $W$ with capacity $C$ be given.
    Let $n'$ be chosen as $n' =\exp_2{(N^{0.5(1-\delta)})}$ with $\delta \in (1 - \beta/0.5, 1)$, where $G$ denotes  $G_2^{\otimes n}\otimes I_{n'}$ and $N' = n' N=n'2^n$.
    Then, for sufficiently large $n$, there exists a sequence of codes corresponding to \textproc{DRS}($G$) satisfying the following properties:
    \begin{enumerate}
        \item 
        The error probability is upper bounded by $\exp_2{( -(\log N')^{2\beta} )}$.
        \item
        The Hamming weight of each column of the GM is upper bounded by $(\log N')^{2\lambda}$.
        \item
        The rate approaches $C$ as $n$ grows large.
        \item 
        There is a SC-based decoder with time complexity $O(N' \log\log N')$.
    \end{enumerate}    
\end{theorem}

\proof Let the column weight threshold for the DRS algorithm be set as $w_{u.b.}= 2^{n\lambda}= N^\lambda$ and $n'$ be set as in \eqref{eq:n'def}. 
%Let $(G_2^{\otimes n})'$ denote the output of the DRS algorithm for matrix $G_2^{\otimes n}$. 
Note that $G' \triangleq \textproc{DRS} (G)$ can be written as  $G' = \textproc{DRS}(G_2^{\otimes n})\otimes I_{n'}$. Then the first claim holds for sufficiently large $n$ by noting the asymptotic error probability bound shown in Proposition \ref{Thm:BEClowcomplDEC} together with an argument similar to that of Lemma \ref{Lemma:newG_PeBound}. 

The second claim follows by combining the column weight threshold $w_{u.b.}= N^\lambda$ and equation \eqref{eq:logN'formula}, which states:  $N^{\frac{1}{2}}  \geq \log(N') \geq  N^{\frac{1-\delta}{2}},$ for any $\delta>0$, for sufficiently large $n$. 
The third claim holds by noting that the rates of the code corresponding to $G'$ and the code corresponding to \textproc{DRS}($G_2^{\otimes n}$) are equal, and the latter approaches the channel capacity by Proposition \ref{Thm:BEClowcomplDEC}. 

Finally, note that each codeword in the code can be regarded as a collection of $n'$ separate codewords in the code corresponding to \textproc{DRS}($G_2^{\otimes n}$). 
Hence, they can be decoded separately in parallel using $n'$ SC decoders. 
 Each copy of the code corresponding to \textproc{DRS}($G_2^{\otimes n}$) can be decoded with complexity $O(N\log N)$, as established by Proposition \ref{Thm:BEClowcomplDEC}. 
 Hence, it follows that the total complexity is $O(n' N\log N)= O(N' \log\log N').$

\endproof

\begin{remark}\label{rem:BECtwoThmCompare}
As pointed out in Remark \ref{rem:wubThresholdCompare}, we have $\lambda^* = \frac12 +\epsilon^*$. Hence, the constraints in the exponents of the column weight upper bounds for the sparsity benchmark (see Section \ref{sec:PeForNewApproach}) in Theorem \ref{coro:DRS_lowComp_sparse_BEC} and Theorem \ref{Coro:LDGMwithUniformBound}, given by $2\lambda$ with $\lambda > \lambda^*$ and $1+\epsilon$ with $\epsilon > 2\epsilon^*$, respectively, are identical. 
\end{remark}

\begin{remark} \label{rem:DRSnotForBMS} 
% \color{blue}
We note that for general BMS channels, Lemma \ref{lem:DRS_Z_BEC} may fail. One key part in the proof (see Appendix \ref{appendix:proof_LCdec_BEC}) 
is the fact that the 
Bhattacharyya parameter for the bit-channel observed by $U_i$ is a non-decreasing function of those of $W(X_1), \ldots, W(X_{f{(m)}})$ for  $i \leq 2^m$, and of $W(X_{f{(m)}+1}), \ldots,$  $W(X_{2f{(m)}})$ for  $i >2^m$, when all the channels are BECs. 
We now provide an example where we see the argument for Lemma \ref{lem:DRS_Z_BEC} fail for BMS channels. Let two BMS channels $W_1, W_2: \mathset{0, 1} \rightarrow \mathcal{Y} = \mathset{\eta, \theta, \kappa, \lambda}$ be given, and that  $W_1( y | 0) = W_1(\phi(y) |1)$ and $W_2(y | 0) = W_2( \phi(y) |1)$ for all $y \in \mathcal{Y}$ where the bijection $\phi$ is the mapping $\eta \mapsto \lambda, \lambda \mapsto \eta, \theta \mapsto \kappa, \kappa \mapsto \theta$. Assume $W_1(\eta | 0)=    6/9, W_1(\theta | 0)=   1/9 $, $W_1(\kappa | 0)=     1/9, W_1(\lambda | 0)=     1/9$ and $ W_2(\eta | 0)=    5/11$, $W_2(\theta | 0) = 4/11$, $W_2(\kappa | 0)=     1/11$, $W_2(\lambda | 0)=1/11$. The Bhattacharyya parameters for $W_1, W_2$ are    respectively $0.7666 $     and $0.7702$. If $m=1$ and $B_m$ is simply the kernel $G_2$, the symbols $X_1, X_2$ are functions of $U_1, U_2$ given by 
$X_1 = U_1 +U_2$ and $X_2 = U_2$. 

We now consider two possible cases for the pair $(W(X_1), W(X_2)) $. 
If $(W(X_1), W(X_2)) = (W_1, W_2 )$, the Bhattacharyya parameters for the bit-channels observed by $U_1, U_2$ are respectively  $    0.9147$ and   $0.5904$. 
If $(W(X_1), W(X_2)) = (W_2, W_2 )$, the Bhattacharyya parameters for the bit-channels observed by $U_1, U_2$ are respectively  $    0.9137$ and   $0.5932$. 
We note that while the Bhattacharyya parameters for $W(X_1), W(X_2)$ in the second case are no less than in the first case, the Bhattacharyya parameter $Z(U_1)$ in the second case is smaller than in the first case. 

With the above observation, one can not claim the validity of Proposition \ref{Thm:BEClowcomplDEC} and Theorem \ref{coro:DRS_lowComp_sparse_BEC} for general BMS channels.  
This motivates a new code construction for general BMS channels.
\end{remark}
% \color{black}

\subsection{Low-complexity Decoder for $G_2$-based LDGM: BMS}
\label{subsection:LCdec_BMS}
This section introduces a capacity-achieving LDGM coding scheme with low-complexity decoder for general BMS channels. 
For general BMS channels, the Bhattacharyya parameter of the bit-channel $W^-$ cannot be expressed only in terms of parameters of the channel $W$. This implies that Lemma \ref{lem:DRS_Z_BEC}, Proposition \ref{Thm:BEClowcomplDEC}, and Theorem \ref{coro:DRS_lowComp_sparse_BEC} are not applicable for channels other than BEC, as pointed out in Remark \ref{rem:DRSnotForBMS}. 
A procedure that augments the generator matrix corresponding to  $G'$, the output of the DRS algorithm for the matrix $G_2^{\otimes n}$, may be used to construct a capacity-achieving linear code over any BMS channel $W$.

The encoding scheme, termed augmented-DRS (A-DRS) scheme, avoids heavy columns in the GM and, at the same time, guarantees that the bit-channels observed by the source bits $U_i$ have the same statistical characteristics as when they are encoded with the generator matrix $G_2^{\otimes n}$.
Specifically, the A-DRS scheme modifies the encoder for $G_2^{\otimes n}$ starting from the split XOR operations associated with the first polarization recursion, then the second recursion, and proceed all the way to the $n$-th recursion, where a XOR operation is split if and only if it is split in an encoder with generator matrix \textproc{DRS}($G_2^{\otimes n}$).

Assume an XOR operation with operands $U_{i_1}^{(n-j)}$ and $U_{i_2}^{(n-j)}$ and the output $U_{i_1}^{(n-j+1)}$, where $i_1 = \mathsf{Bi2De}(s_1, s_2, \ldots,s_{j-1}, s_j =-, s_{j+1}, \ldots, s_{n}  )+1$ and $i_2 = \mathsf{Bi2De}(s_1, s_2, \ldots, s_{j-1}, s_j =+, s_{j+1}, \ldots, s_{n} )+1= i_1 +2^{n-j}$, is to be split (see Section \ref{subsection:LCdec_BEC} for the function $\mathsf{Bi2De}(\cdot)$). 
If $j =1$, before modification, the variables $U_{i_1}^{(n)}$ and $U_{i_2}^{(n)}$ are transmitted through two copies of $W$, and the bit-channels observed by $U_{i_1}^{(n-1)}$ and $U_{i_2}^{(n-1)}$ are $W^{-}$ and $W^{+}$, respectively, as shown in Figure\,\ref{fig:A-DRS-generalUnsplit_j1}. 
If the XOR operation is split according to \textproc{DRS}($G_2^{\otimes n}$), it is replaced by the structure given in Figure\,\ref{fig:A-DRS-generalSplit_j1}, where $n_{i_1, 1}$ is a Bernoulli($0.5$) random variable independent of all the other variables. 
\begin{figure}
      \centering
      \begin{subfigure}[b]{0.43\textwidth}
          \centering
		\includegraphics[trim=2cm 2.5cm 15cm 15cm, clip,  width= \textwidth]{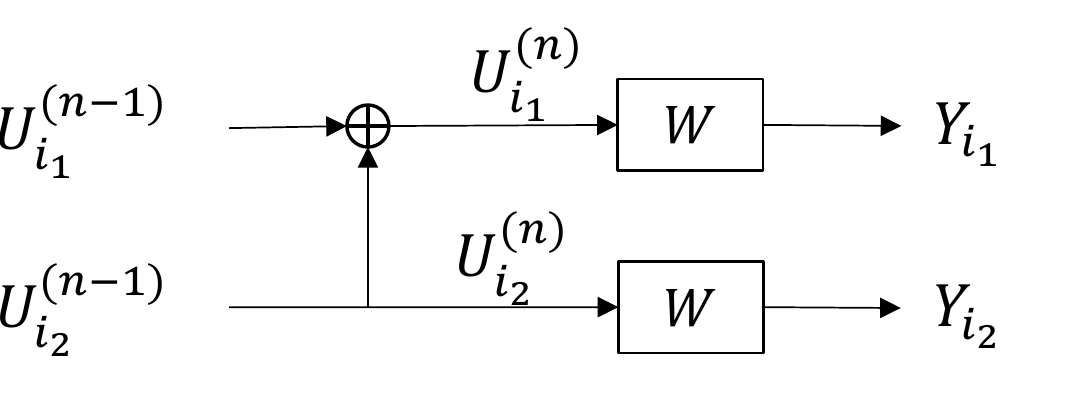}
          \caption{Before Modification}
          \label{fig:A-DRS-generalUnsplit_j1}
      \end{subfigure}
       \hfill
      \begin{subfigure}[b]{0.45\textwidth}
          \centering
		\includegraphics[trim=2cm 1cm 14cm 14cm, clip,  width= \textwidth]{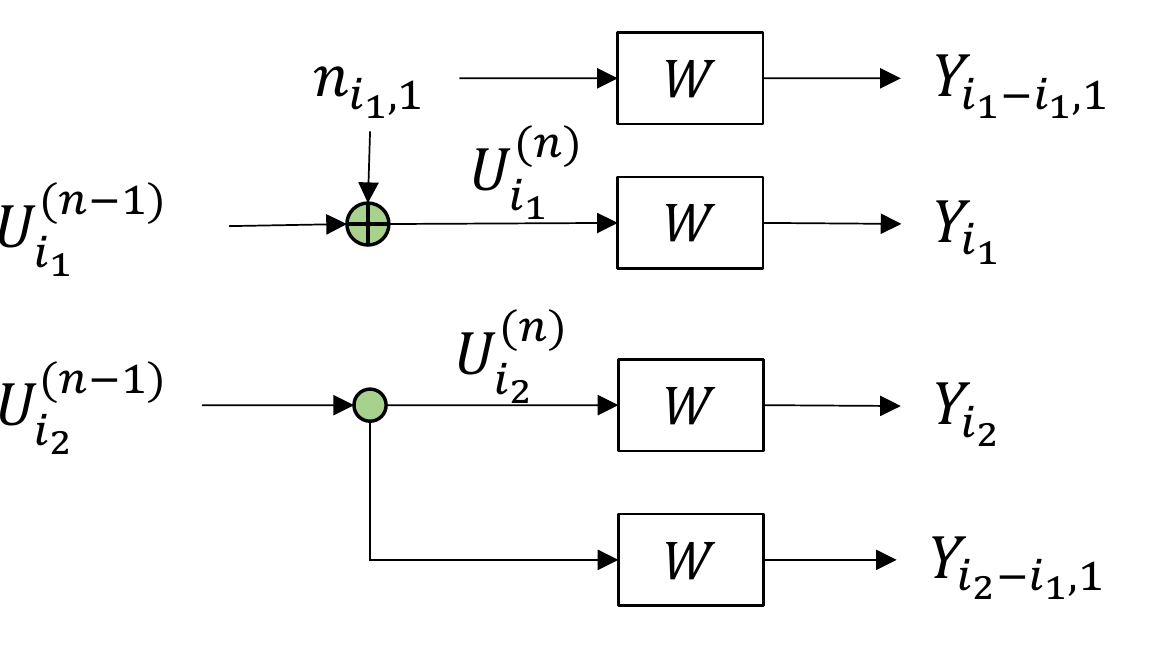}
          \caption{A-DRS modification}
          \label{fig:A-DRS-generalSplit_j1}
      \end{subfigure}
     \caption{A-DRS scheme for a split XOR of first iteration of polarization}
     \label{fig:A-DRS-j=1}
\vspace{-7mm}
\end{figure}

If $j \geq 2$, assume that the A-DRS modification for the split operations for the first $(j-1)$ recursions are completed. 
Let $n_{i_1, j}$ be a Bernoulli($0.5$) random variable independent of all the other given variables. 
The part of encoding diagram to the right of $U_{i_1}^{(n-j+1)}$ is replicated, where  $n_{i_1, j}$ takes the place of $U_{i_1}^{(n-j+1)}$ in the replica. And then we let  $U_{i_1}^{(n-j+1)} = U_{i_1}^{(n-j)}\oplus n_{i_1, j}$. 
In addition, the part of encoding diagram to the right of $U_{i_2}^{(n-j+1)}$ is replicated, 
and a copy of $U_{i_2}^{(n-j)}$ is transmitted through the replica. The variable  $U_{i_2}^{(n-j+1)} $ remains  $U_{i_2}^{(n-j+1)}= U_{i_2}^{(n-j)}$. 
 
We demonstrate the procedure described above through the following example. 
Assume $n=3$, $N=8$, and $w_{u.b.} =2$. 
The encoding diagram for $G_2^{\otimes 3}$ is shown in Figure\,\ref{fig:ADRS_unsplit}, and the XOR operations that are split in \textproc{DRS}($G_2^{\otimes 3}$) are marked in green and blue, which indicate the operations are due to the first and the second polarization recursions, respectively. 
The notations $U_i', U_i'', U_i'''$ are used to represent $U_i^{(1)}, U_i^{(2)}, U_i^{(3)}$. 
\begin{figure}[h]
	\centering
	\includegraphics[trim=1.5cm 1.5cm 9.5cm 8.5cm, clip,  width= 0.5\textwidth]{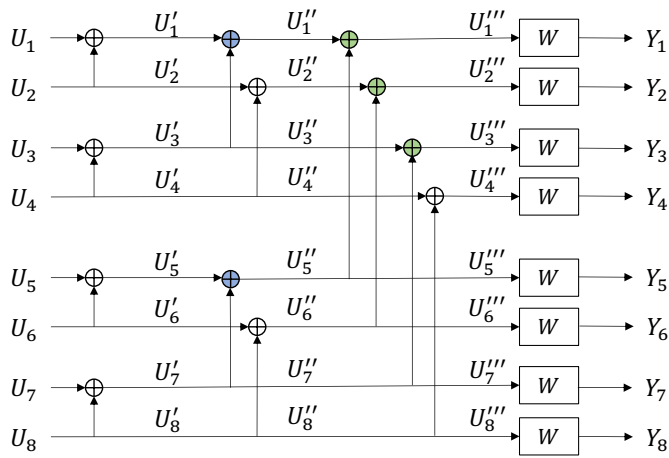}
	\caption{Encoding diagram for $G_2^{\otimes 3}$}\label{fig:ADRS_unsplit}
\vspace{-7mm}
\end{figure}

Replacing the XOR operations marked in green as described for the case of $j=1$, the encoding diagram is now shown in Figure\,\ref{fig:ADRS_1st_iteration}. 
For the XOR operations marked in blue, we proceed by using the step for $j\geq 2$ and obtain the diagram shown in Figure\,\ref{fig:ADRS_n3_entire}. 

\begin{figure}[h!]
	\centering
	\includegraphics[trim=1cm 1.5cm 5cm 2.5cm, clip, width= 0.6\textwidth]{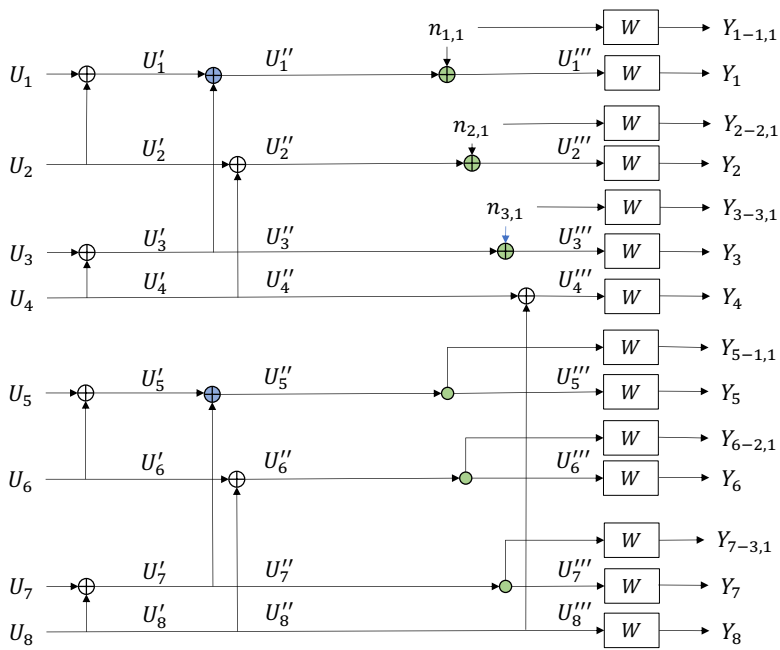}
	\caption{A-DRS for splits corresponding to $s_1$ in $G_2^{\otimes 3}$ }\label{fig:ADRS_1st_iteration}
\vspace{-7mm}
\end{figure}
\begin{figure}[h!]
	\centering
	\includegraphics[trim=4cm 4cm 4cm 6cm, clip, width= 0.6\textwidth]{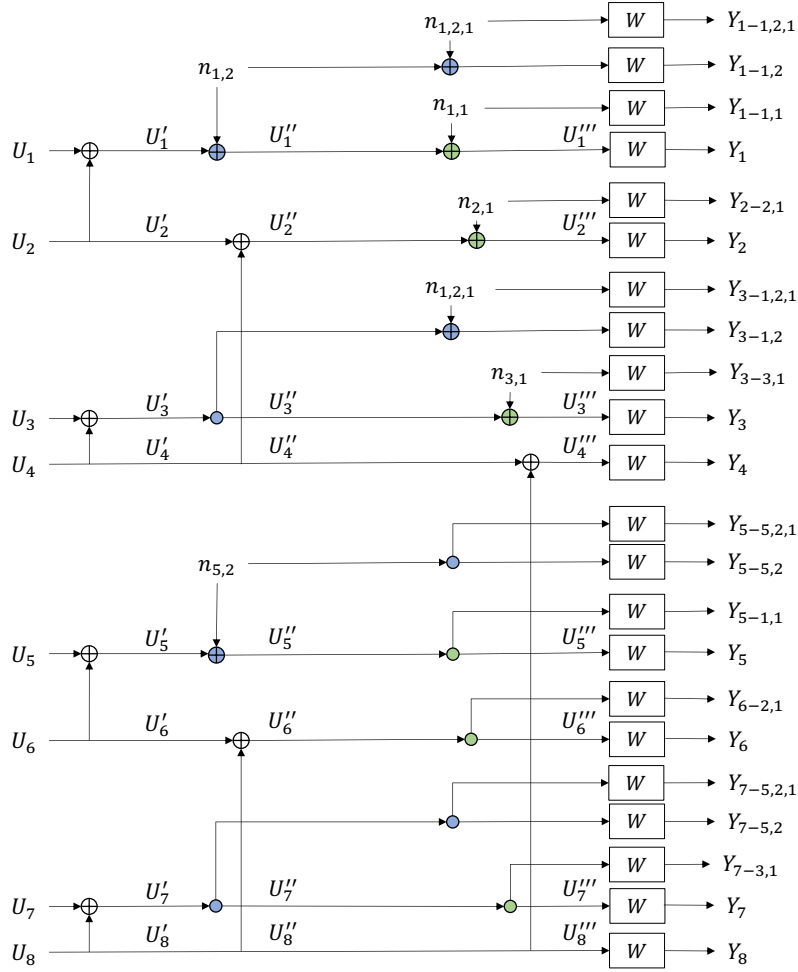}
	\caption{A-DRS Encoding Diagram for $G_2^{\otimes 3}$ with $w_{u.b.}=2$}\label{fig:ADRS_n3_entire}
	\vspace{-7mm}
\end{figure}

It can be noted that the bit-channels observed by each of   $U_i^{(j)}$, for $i= 1, 2, \ldots, N$ and $j = 0,1,2, \ldots, n$, in the A-DRS encoder are the same as those in the standard encoder for the generator matrix $G_2^{\otimes n}$ (The variable $U_i^{(0)}$ are given by $U_i$ for $1\leq i \leq N$). 
When an XOR operation associated with the $j$-th recursion, with operands $U_{i_1}^{(n-j)}$ and $U_{i_2}^{(n-j)}$ and the output $U_{i_1}^{(n-j+1)}$, is split and modified under the A-DRS scheme, the  complexity of computing the likelihood or log-likelihood for $U_{i_1}^{(n-j)}$ and $U_{i_2}^{(n-j)}$ can be upper bounded by $2(2^1 + 2^2+ \ldots + 2^j)c = 2(2^{j+1} -2)c$, for some constant $c>0$.

\begin{proposition}\label{prop:ADRScompl}
    Let a constant $\lambda > \lambda^{\dagger} \triangleq (\log_2 3)^{-1} \approx 0.631$ be given.
    The decoding complexity for a SC decoder for the A-DRS scheme is bounded by $O(N\log N)$ for all sufficiently large $n$ if the threshold for the DRS algorithm is $w_{u.b.} = 2^{n\lambda}$. 
\end{proposition}

It can be observed that the number of additional copies of channels due to the modification for an XOR operation at the $j$-th polarization recursion is $2^j$. 
We find the total number of extra channel uses and the ratio $\gamma$ of that to the number $N= 2^n$ of channel uses for the code corresponding to $G_2^{\otimes n}$ in the following.  Assume that the column weight threshold of the DRS algorithm is given by $w_{u.b.} = 2^{n\lambda}$.

\begin{proposition}\label{Prop:BMSrateloss}
    Let $N(1+\gamma)$ be the number of channel uses of the encoder for the A-DRS scheme based on \textproc{DRS}($G_2^{\otimes n}$) with $w_{u.b.} = 2^{n\lambda}$.
    Then the term  $\gamma$ goes to $0$ as $n$ grows large, if we have $\lambda > \lambda^{\dagger}$.
\end{proposition}

We are ready to show the existence of a sequence of capacity-achieving codes over general BMS channels with GMs where the column weights are bounded by a polynomial in the blocklength, and that the block error probability under a low  complexity decoder vanishes as the $n$ grows large.
\begin{proposition}
    \label{Thm:BMSlowcomplDEC}
    Let $\beta < E(G_2)=0.5 $, $\lambda > \lambda^{\dagger}$ and a BMS channel $W$ with capacity $C$ be given. There exists a sequence of codes 
    with the following properties for all sufficiently large $n$:
    \begin{enumerate}
        \item 
        The error probability is upper bounded by $2^{-N^{\beta}}$, where $N= 2^n$.
        \item
        The Hamming weight of each column of the GM is upper bounded by $N^\lambda$.
        \item
        The rate approaches $C$ as $n$ grows large.
        \item 
        The codes can be decoded by a successive-cancellation decoding scheme with complexity $O(N\log N)$.
    \end{enumerate}
    
\end{proposition} 
\proof 
We prove the four properties in order as follows.
First, similar to the proof of Theorem\,\ref{Thm:BEClowcomplDEC}, for $i =1, 2,\ldots, N$, the bit $U_i$ is frozen in the A-DRS code with rate $R< C$ if and only if it is frozen in the polar code with kernel $G_2$, blocklength $N= 2^n$, and the rate $R$.
Hence, the probability of error of the A-DRS code can be bounded in the same way as its polar-code counterpart, since the bit-channels observed by the source bits $U_i$, and the corresponding Bhattacharyya parameters, are identical to those when they are encoded with the standard polar code.

Second, when the A-DRS scheme is based on \textproc{DRS}($G_2^{\otimes n}$) with $w_{u.b.} = 2^{n\lambda}$, 
the generator matrix for the A-DRS code is a submatrix of \textproc{DRS}($G_2^{\otimes n}$). 
The column weights of the GM for the A-DRS code are thus upper bounded by $w_{u.b.} = 2^{n\lambda} = N^{\lambda}$.
The third claim holds by using an argument similar to the one used in the 
proof of Proposition\,\ref{Thm:BEClowcomplDEC}. This is because the term $\gamma$ vanishes as $n$ grows large according to Proposition\,\ref{Prop:BMSrateloss}.
Finally, note that the fourth claim is equivalent to Proposition\,\ref{prop:ADRScompl}.

\endproof

\begin{theorem}
    \label{prop:Sparse_ADRScompl}
   Let $\beta < E(G_2)=0.5 $ and $\lambda >\lambda^{\dagger}$ be given. Then there exists a sequence of codes corresponding to  $G'$, constructed by applying the A-DRS algorithm to $G= G_2^{\otimes n}\otimes I_{n'}$,
    with the following properties:
    \begin{enumerate}
        \item 
        The error probability is upper bounded by $\exp_2{(-(\log N')^{2\beta} )}$.
        \item
        The Hamming weight of each column of the GM is upper bounded by $ (\log N')^{2\lambda}$.
        \item
        The rate approaches the capacity $C$ as $n$ grows large.
        \item
        The decoding time complexity is upper bounded by $O(N' \log\log N').$
    \end{enumerate}
\end{theorem}
\proof
The theorem can be shown using the same steps as in the proof of Theorem\,\ref{coro:DRS_lowComp_sparse_BEC}.
\endproof

\section{Sparsity with General Kernels} \label{sec:L>2}

In this section we consider $l \times l$ kernels $G_l$ with $l > 2$, and show the existence of $G_l$ with $w_{GM}(n, G_l)= O( ({\log{N'}})^{\lambda})$ for some $\lambda <1$, where $N' = n' \,l^n$ is the number of columns of $G_l^{\otimes n} \otimes I_{n'}$.  Using a similar argument as in Section \ref{subsection:mainL2},  most of the column weights of $G_l^{\otimes n}$ can be made to scale with a smaller power of $\log N'$ asymptotically than when $G_2$ is used as the kernel.
To characterize the geometric mean column weight and the maximum column weight, the \textit{sparsity order} is defined as follows:
\begin{definition}\label{def:lambdaMC}
	The sparsity order of the geometric mean column weight is
	\begin{equation}
	\lambda_{GM}(n,G_l)  \, \triangleq \, \log_{\log(N')}{w_{GM}(n,G_l)} =\frac{\log {w_{GM}(n,G_l)}}{ \log {\log(N')}},
	\end{equation}
    where $n'$ and $N'$ are defined in \eqref{eq:n'def} and \eqref{eq:N'def}, respectively. 
\end{definition}
\begin{definition}\label{def:lambdaMax}
    The {sparsity order} of the maximum column weight is 
    \begin{equation}
    \lambda_{max}(n, G_l)  \, \triangleq \, 
    \log_{\log(N')}{w_{max}(n,G_l)} =\frac{\log {w_{max}(n,G_l)}}{ \log {\log(N')}}.
    \end{equation}
\end{definition}

For example, if ${w_{GM}(n,G_l)}$ (or ${w_{max}(n,G_l)} $) scales as $\Theta([\log N']^r) $, then $ \lambda_{GM}(n,G_l)$ (or $\lambda_{max}(n, G_l) $ ) goes to $r$ as $n$ grows large. Table \ref{table:orderGeneralL}\footnote{The limits of the sparsity orders when $n\to \infty$ are shown, hence $o(1)$ terms are neglected.} shows the values of $\lambda_{GM}(n,G_l)$ and $\lambda_{max}(n,G_l)$ when $n\to \infty$  for some of the kernels considered in \cite{korada2010polar}: 
\[
G_3^* = \begin{bmatrix}
0   &   1   &   0\\
1   &   1   &   0\\
1   &   0   &   1
\end{bmatrix}, 
G_4^* = \begin{bmatrix}
1   &   0   &   0   &   0\\
0   &   1   &   0   &   1\\
0   &   0   &   1   &   1\\
1   &   1   &   1   &   1
\end{bmatrix},
\] and $G_{16}^*$ (the smallest $l$ with $E_l>0.5$; see \cite{korada2010polar} for explicit construction), which are the matrices achieving $E_3, E_4$ and $E_{16}$, the maximal rates of polarization for $l=3,4, \mbox{ and } 16$, respectively. 
Recall equation \eqref{eq:n'def} for the definition of $\delta$, which determines the length of the code, and $\delta' = \frac{\delta}{1- \delta}$ as in the proof of Proposition \ref{prop:wMCloginN}.
\begin{table}[t] 
\centering
\begin{tabular}{l | l | l | l}
 		& $E(G_l)$     & $\lambda_{GM}(n,G_l)$    & $\lambda_{max}(n, G_l)$   \\ [2mm]  \hline  
$G_2$ 	& 0.5              
		& $1+\delta'$  
		& $2(1+\delta')$   \\  [2mm]
		
$G_3^*$ 	&  $\frac{2}{3}\log_3{2}\approx 0.42$    
		&  $1+\delta'$    
		& $1.5 (1+\delta')$  \\  [2mm]

$G_4^*$ 	&   0.5                     
		&  
% 		$\frac{3+\log 3}{4}(1+\delta')
		$\approx 1.15(1+\delta')$                    
		& $\log{3}(1+\delta')$  \\  [2mm]

$G_{16}^*$ 
		&  $\approx 0.5183 $                   
		& $\approx  1.443(1+\delta')  $        
		& omitted \\
\end{tabular} 
\caption{$\lambda_{GM}$ and $\lambda_{max}$ for $G_2, G_3^*,G_4^* $ and $G_{16}^*$ as $n\to \infty$ }

\label{table:orderGeneralL}
\vspace{-8mm}

\end{table}

However, the rate of polarization is not the only factor that determines the sparsity orders. For example, for $l=3$ and $l=4$, the matrices
\[
G_3' = \begin{bmatrix}
1   &   0   &   0\\
1   &   1   &   0\\
1   &   0   &   1
\end{bmatrix}, 
G_4' = \begin{bmatrix}
1   &   0   &   0   &   0\\
1   &   1   &   0   &   0\\
1   &   0   &   1   &   0\\
1   &   0   &   0   &   1
\end{bmatrix},
\] instead of $G_3^*$ and $G_4^*$, have the smallest sparsity orders of the geometric mean column weight (found through exhaustive search), as shown in table \ref{table:BetterorderGeneralL}.
By central limit theorem, most column weights scale exponentially in the logarithm of the block length in the same way as the geometric mean column weight. 
Therefore, if sparsity constraint is only required for almost all of the columns of the GM, $G_3'$ and $G_4'$ are the more preferable polarization kernels over   $G_3^*$ and $G_4^*$, respectively.

\begin{table}[t] 
\centering 
\begin{tabular}{l | l | l | l}
 		& $E(G_l)$     & $\lambda_{GM}(n,G_l)$    & $\lambda_{max}(n, G_l)$   \\ [2mm]  \hline  
$G_3'$ 	&  $\frac{2}{3}\log_3{2} \approx 0.42$    
		&  $
% 		\frac{1}{2}\log{3}(1+\delta') 
		\approx
		0.79(1+\delta') $    
		& $
% 		\frac{3}{2}\log{3}(1+\delta')
		\approx
		2.38(1+\delta')$  \\  [2mm]

$G_4'$ 	&  $ \frac{3}{8}= 0.375  $        
		&  $\frac{2}{3}(1+\delta')$       
		& $\frac{8}{3}(1+\delta')$  \\
\end{tabular} 
\caption{$\lambda_{GM}$ and $\lambda_{max}$ for   $G_3'$ and $G_4'$ as $n\to \infty$}
\label{table:BetterorderGeneralL}
\vspace{-8mm}
\end{table}

For a given $G_l$, we may relate the two terms $E(G_l)$ and $w_{GM}(n)$, or, more specifically, the partial distances $ D_1, \ldots, D_l$ and the column weights $w_1, \ldots, w_l $ as follows.

\begin{lemma}\label{lemma:SparsityOrderGeneralL}
Let $G_l$ be a $l\times l$ polarization kernel, and let $\beta < E(G_l)$ be given. 
Then for any $1 > \delta > 1- \frac{\beta}{E(G_l)}$, the term $\lambda_{GM}(n,G_l)$ can be bounded as 
 \[
  \frac{ \sum_{i=1}^l \log_l w_i } {\sum_{i=1}^l \log_l D_i } \leq
\lambda_{GM}(n,G_l) \leq 
\frac{1}{1-\delta}
 \frac{ \sum_{i=1}^l \log_l w_i } {\sum_{i=1}^l \log_l D_i },
\]
for all sufficiently large $n$.  
\end{lemma}

The following theorem shows that an arbitrarily small order can be achieved with a large $l$ and some $G_l$. 

\begin{theorem}\label{Thm:lambdaMC_lgeneral}
	For any fixed constant $0 < r\leq 1$, there exist  an $l\times l$ polarizing kernel $G_l$, where $l= l(r,\delta)$, such that $\lambda_{GM}(n,G_l) <r  $ for all sufficiently large $n$.
\end{theorem}

Let $r<1$ and $\eta >0 $ be fixed. For an appropriate choice of  $G_l$ with $\lim_{n\rightarrow \infty}\lambda_{GM}(n,G_l) <r$, concentration of the column weights implies that only a vanishing fraction of columns in $G$ has weight larger than $ [{\log{N'}}]^{(1+\eta)r}$ for all sufficiently large $n$. 
% \blue{
The reader may follow the steps used in the proof of Lemma \ref{lem:wMClogN_most}, and use the fact that the logarithm of the column weight has the same distribution as a sum of i.i.d. random variables, which goes to $\lambda_{GM}(n,G_l)$ w.p.1. when normalized by $n$, as guaranteed by the law of large numbers.
% }

\section{Conclusion}\label{sec:Conclusion}
We proposed three constructions for capacity-achieving polar-based LDGM codes where all the generator matrix column weights are upper bounded as $O( (\log N)^\lambda )$, where $\lambda$ is slightly larger than $1$. Our schemes are based on a concatenation of $G_2^{\otimes n}$ and a rate-$1$ code, and column-splitting algorithms which guarantee the heavy columns are replaced by lighter ones. Two of the constructions also allow the codes to be decodable with low-complexity decoders for the BECs and general BMS channels. 
% \blue{
An RLE-based construction for a capacity-achieving code sequence is given with $O(\log N)$ sparsity over general BMS channels under polynomial-time-complexity decoding. 
% }
Broadly stated, this paper studies the existence of LDGM codes with constraint on the column weights of the GMs. It remains an open question whether LDGM codes with even sparser GMs (column weights scaling sub-logarithmically in $N$) exist. A future direction in this regard is to determine explicitly the scaling behaviour of smallest upper bound on the column weights of GMs for a sequence of rate-$R$ achieving (with arbitrarily low probability of error as $N$ grows) codes.

\appendix
\label{sec:Appendix}

\subsection{Proofs for Subsection \ref{sec:sparsityConstraint}}

\textit{{Proof of Theorem \ref{prop:polarWithPolyColumnWeight}}: }
Consider an $l\times l$ polarizing matrix
	\begin{equation*}
		G=
		\begin{bmatrix}
		I_{\frac{l}{2}} & 0_{\frac{l}{2}} \\
		I_{\frac{l}{2}} & I_{\frac{l}{2}} \\
		\end{bmatrix}, 
		\label{eq:polyGform}  
	\end{equation*} where $l$ is an even integer such that $l \geq 2^{\frac{1}{s}}$.
Note that by \eqref{eq:DiDef} and \eqref{eq:DlDef}, we have 
$ D_i= 1$ for $1\leq i\leq \frac{l}{2}$ and $D_i= 2$ for $\frac{l}{2}+1 \leq i\leq l$. Hence, the rate of polarization 
	$E(G) = \frac{1}{2}\log_l 2  >0 $, and there is a sequence of capacity-achieving polar codes constructed using $G$ as the polarizing kernel. Note that in $G$, each column has weight at most $2$ and, hence, the column weights of $G^{\otimes n}$ is upper bounded  by $2^n$. By the specific choice of $l$, we have
    \[
    2^n \leq {(l^s)}^n = {(l^n)}^s = N^s,
    \] where $N= l^s$ is the block length of the code. This completes the proof. 
\endproof

\textit{{Proof of Proposition \ref{prop:noLogColumns}}: }
%    Reason: 
    Since $G$ is a polarization kernel, there is at least one column in $G$ with weight at least $2$. 
    To see this, note that $G$ being invertible implies that all rows and columns are nonzero vectors. Now, if all the columns of $G$ have weight equal to $1$, then all the rows must also have weight equal to $1$, i.e., $G$ is a permutation matrix. Then $D_i =1, \forall i$, and $ E(G) =0$, which implies that $G$ can not be polarization kernel. The contradiction shows that at least one column in $G$ must have a weight at least $2$. 
    
    Let $k \geq 1$ denote the number of columns in $G$ with a weight at least $2$. 
    Let \textbf{v} be a randomly uniformly chosen column of $G^{\otimes n}$, and $w(\textbf{v}) $ be the Hamming weight of \textbf{v}. 
    For $r>0$,
    \begin{align*}
        &\Pr\left( w(\textbf{v}) =  O({(\log N)}^r)  \right)
        \leq \Pr \left(2^{\sum_{i=1}^n F_i}   = O({(\log N)}^r)= O(  ({\log l^n} )^r  ) = O(n^r) \right), 
    \end{align*}
    where $F_i $ is the indicator variable that one of $k$ non-unit-weight columns is used in the $i$-th Kronecker product of $G$ to form \textbf{v}. 
    The variables $F_1, F_2, \ldots, F_n$ are i.i.d. as $\text{Ber}(k/l) $. 
    Law of large numbers implies that
    $\sum_{i=1}^n F_i = \Theta(n)$ with high probability.
    Thus, $$ \Pr \Big(2^{\sum_{i=1}^n F_i} = O(n^r) \Big) \to 0, $$ for any $r>0$ as $n \to \infty .$
\endproof

\subsection{Proofs for Subsection \ref{sec:PeForNewApproach}}

\textit{Proof of Lemma \ref{Lemma:newG_PeBound}:}
	Let $P_e$ denote the block error probability. 
	For $\beta' < \beta < E(G_l) $, $P_e \leq 2^{-N^\beta}  $ implies that $P_e \leq 2^{-N^{\beta'}}  $. Therefore, it suffices to show that bound on the probability of error holds for  $\beta =( 1-\eta)E(G_l)$ for any $ \eta\in (0, 1/2)$. 
	
	Let $1- \frac{\beta}{ E(G_l)} =  \eta \in (0, 1/2)$ be fixed and $\delta \in (\eta, 1)$ be chosen. 
	Note that polar codes with rate $R< C$ constructed using kernel $G_l$ have the error probability upper bounded by $2^{-N^{(1- \eta /2)E(G_l)}}$ as $n$ grows large \cite{korada2010polar}. 
	For the code corresponding to  $G$, $P_e$ is then bounded by 
	\begin{align*}P_e 
	&\leq n’ \times 2^{-N^{(1- \eta /2) E(G_l)}} = 2^{N^{(1-\delta)E(G_l)}} \times  2^{-N^{(1- \eta /2) E(G_l)}} \\
	&\leq 
	2^{N^{(1-\eta)E(G_l)}} \times  2^{-N^{(1- \eta /2) E(G_l)}} \\
	&=  2^{-N^{(1- \eta /2) E(G_l)} (1-N^{-\frac{\eta}{2}  E(G_l) } )  } \leq 2^{-N^{(1- \eta ) E(G_l) } } =
	2^{-N^{\beta } } 
	\end{align*}
	for all sufficiently large $n$.
	
	The expression in terms of $N'$ follows from the bound in \eqref{eq:logN'formula}.
\endproof

\textit{\\Proof of Proposition \ref{prop:codeWithG'}:}

\begin{itemize}
\item{(Step 1) SC decoding: }\label{pfProp4:SCdecoding}
Successive cancellation decoding has been used in \cite{ArikanRatePolarization09} as a low-complexity decoding scheme for capacity-achieving polar codes. 
In the analysis of block error probability, however, considering a genie-aided successive cancellation decoding scheme \cite{moser2019information}, where the information of correct $U_1^{i-1}$ is available when the decoder is deciding on $\hat{U_i}$, based on the maximum likelihood estimator, often simplifies the analysis. In this case, note that $\hat{U}_i$ is a function of $Y_1^N$ and 
$U_{1}^{i-1}$. 
It is stated in \cite[Lemma 14.12]{moser2019information} that the probability of error of the original SC decoder and that of a genie-aided successive cancellation decoder are in fact equal. 

In terms of the generator matrix $G$ of a code, for the estimation of $\hat{U_i}$, the channel output can be thought of as the noisy version of the codeword obtained when $U_i^{n}$ is encoded with the matrix consisting of the bottom $(n-i+1)$ rows  of $G$. 

\item{(Step 2) Error probability bound for SC decoding: }\label{pfProp4:PeforSC}
As discussed in Section \ref{Prelim:polarExp}, for any $\beta < E(G_l)$, there is  a sequence of  capacity-achieving polar codes and with kernel $G_l$ such that $P_{e,SC} \leq 2^{-N^\beta}$ for all sufficiently large $n$.. 

\item{(Step 3) Splitting on polar code improves the code:}
Let the column weight threshold $w_{u.b.}$ of $G$ be given. Let $ (G_l^{\otimes n})^{sp}$ denote the $N\times N(1+\gamma)$ matrix generated by the splitting algorithm acting on $G_l^{\otimes n}$. 
We have the following the lemma whose proof will be provided later. 

\begin{lemma}\label{lemma:SCDforBlock}
	Under SC decoding, the probability of error of the polar code with kernel $G_l$ is no less than that of the code, with the same row indices, corresponding to  $ (G_l^{\otimes n})^{sp}$. 
% 	Therefore, same $2^{-N^\beta}$ bound applies for any $\beta < E(G_l)$.
\end{lemma}

\item{(Step 4) Decoder for the code corresponding to  $G'$:}
	When the splitting algorithm is applied on $G$, we may assume that all the new columns resulting from a column in the $j$-th chunk of $G$ are placed in the $j$-th chunk, where a chunk is the set of $N$ columns using the same $G_l^{\otimes n}$. In addition, we may require that the splitting algorithm adopts the same {division principle}. For example, the first new column includes the $w_{u.b.}$ ones with smallest row indices in the split column, the second new column includes  another $w_{u.b.}$ ones with smallest row indices, excluding those used by the first new column, and so on.
	By the structure of $G$ and the above requirement on the splitting algorithm, the matrix $G'$ has the following form: 
	\begin{equation}
	G'=
	\begin{bmatrix}
	{(G_l^{\otimes n})}^{sp}       & \zero  & \dots & \zero \\
	\zero       & {(G_l^{\otimes n})}^{sp}   & \dots & \zero \\
	\vdots & \vdots & \ddots & \vdots \\
	\zero       & \zero  & \dots & {(G_l^{\otimes n})}^{sp} 
	\end{bmatrix}, \label{eq:G'matrixForm}
	\end{equation}
	where each $\zero$ represents an $N\times N(1+\gamma)$ zero matrix. 
	
    For the code corresponding to  $G'$, we can divide the information bits $u_1, \ldots, u_{K'}$ into $n'$ chunks, $(u_1, \ldots, u_K)$, $(u_{K+1}, \ldots, u_{2K}),$ $\ldots$,$ (u_{(n'-1)K+1}, \ldots ,u_{n'K}=u_{K'}   )$, written as 
    $\mathbf{u}_1, \mathbf{u}_2, \ldots, \mathbf{u}_{n'}$.
	Similarly, the coded bits $c_1, \ldots, c_{N'(1+\gamma)}$ can be divided into $n'$ chunks, $(c_1, \ldots, c_N(1+\gamma))$, $(c_{N(1+\gamma)+1}, \ldots, c_{2N(1+\gamma)})$, $\ldots, (c_{(n'-1)N(1+\gamma)+1}, \ldots, c_{n'N(1+\gamma)}= c_{N'(1+\gamma)})$, denoted by 
	$\mathbf{c}_{1},\mathbf{c}_{2}, \ldots, \mathbf{c}_{n'}$. 
    For the structure of $G'$ and memorylessness of the channels, the chunk $\mathbf{c}_{j}$ depends only on the information chunk $\mathbf{u}_{j}$, through a $K\times N(1+\gamma)$ submatrix of $ {(G_l^{\otimes n})}^{sp}$, and independent of other information chunks.

We now describe the decoder for the code corresponding to  $G'$. The decoder consists of $n'$ identical copies of SC decoders, where the $j$th decoder decides on $\hat{\mathbf{u}}_{j}$ using the channel output when $\mathbf{c}_{j}$ is transmitted.

\item{(Step 5) Error rate for the code corresponding to  $G'$:}
 Lemma \ref{lemma:SCDforBlock} shows that the block-wise error probability of the code corresponding to  $ (G_l^{\otimes n})^{sp}$ is smaller or equal to that of $ G_l^{\otimes n}$. For any $ \beta < E(G_l)$, the rate of the code, whose block error probability is bounded by $2^{-N^\beta} $, approaches $C$ as $n$ grows. 
    Using the union bound as in Lemma \ref{Lemma:newG_PeBound}, for any $ \beta < E(G_l)$, when $n'$ is chosen as in \eqref{eq:n'def}, the probability of error of the code corresponding to  $G'$ with the proposed decoder, denoted by $P_e^{SC}(G')$,  can be bounded by $P_e^{SC}(G') \leq 2^{-N^\beta} $ for all sufficiently large $n$.
 	\end{itemize}
\endproof

\textit{Proof of Lemma \ref{lemma:SCDforBlock}: }
We will show that, for an  $N_1\times N_2$ matrix $M$, when a column is split into two nonzero columns to form an  $N_1\times (N_2+1)$ matrix $M'$, the probability of error of the code with generator matrix $M$ is lower bounded by that with $M'$.
%		 Here the generator matrices of the two  codes are assumed to use the same row indices of $M$ and $M'$, respectively. 
Without loss of generality, we may assume that the first column of $M$, denoted by $v_1$,  is split into two columns and become, say, the first two columns of $M'$, denoted by $v_1', v_1'' $. We may also assume that the first two elements of $v_1$ are both equal to $1$. 
The columns are associated by the following equation:\
\begin{equation}\label{eq:v1v1'}
v_1= 
\begin{pmatrix} 1 \\ 1 \\ \ast \\ \vdots \\ \ast \\ \star  \\ \vdots \\ \star \end{pmatrix} = 	
\begin{pmatrix} 1 \\ 0 \\  \ast \\ \vdots \\  \ast \\ 0  \\ \vdots \\ 0 \end{pmatrix} +
\begin{pmatrix} 0 \\ 1 \\ 0 \\ \vdots \\ 0 \\ \star  \\ \vdots \\ \star \end{pmatrix} = 	
v_1'+ v_1''.	 		
\end{equation}			
For $v_1'$, row operations can be applied to cancel out the nonzero entries in the positions occupied by $\ast$ using the $1$ at the first row. Similarly, the second row of $v_1''$ can be used in row operations to cancel out any nonzero entries in the positions occupied by $\star$. Let $E$ be the matrix of the concatenation of the row operations on $v_1'$ and $v_1''$. We have 
\begin{equation}\label{eq:v1v1'}
E v_1= 
\begin{pmatrix} 1 \\ 1 \\ 0 \\ \vdots \\ 0 \\ 0  \\ \vdots \\ 0 \end{pmatrix} = 	
\begin{pmatrix} 1 \\ 0 \\ 0 \\ \vdots \\ 0 \\ 0  \\ \vdots \\ 0 \end{pmatrix} +
\begin{pmatrix} 0 \\ 1 \\ 0 \\ \vdots \\ 0 \\ 0  \\ \vdots \\ 0 \end{pmatrix} = 	
E(v_1'+ v_1'') 		
\end{equation}
Since $E$ is composed of a sequence of invertible $N_1 \times N_1$ matrices, the row space of $M$ is the same as that of $E M$. 
Similarly, the row spaces of $M'$ and $E M'$ are equal. Hence the codes with GMs  $M$ and $E M$ are the same, and so are the codes with GMs $M'$ and $E M'$. 
% \blue{
First, we show the probabilities of error for codes with GMs $M$ and $EM$ are equal. 
For each sequence $\mathbf{u} =(u_1, \ldots, u_{N_1})$ there is a unique sequence  $\mathbf{ \tilde{u} } =(\tilde{u}_1, \ldots, \tilde{ u}_{N_1} )$ such that the codeword $\mathbf{x} = \mathbf{u}M$ is equal to $\tilde{\mathbf{x}} = \tilde{\mathbf{u}} (EM)$. Denote the bijective mapping by $f: \mathset{0, 1}^{N_1} \rightarrow \mathset{0, 1}^{N_1}$, $ f( \mathbf{u} ) = \mathbf{ \tilde{u} }$. 
Assume the channel outputs are denoted by $\mathbf{y} $ and $\mathbf{ \tilde{y} }$ when the inputs are $\mathbf{x} $ and $\mathbf{ \tilde{x} }$, respectively. The sequences $\mathbf{y} $ and $\mathbf{ \tilde{y} }$ are identically distributed. 
Assume now there is an SC decoder $\mathcal{D}_M $ that returns an estimate of $\mathbf{u} $ based on $\mathbf{y} $, with error probability $P_e(\mathbf{u} )$. For the code defined by GM being $EM$, one may construct a decoding algorithm as follows.  Given channel output  $\mathbf{ \tilde{y} }$, invoke the SC decoder $\mathcal{D}_M  $, which returns an estimate of $ \mathbf{u}  $ with error probability $P_e(\mathbf{u} )$. The algorithm then map  $\mathbf{u}  $ to $\mathbf{ \tilde{u} }$ using $f$. The above algorithm is an SC-based decoder with the same block error probability as that of $\mathcal{D}_M $. 
Conversely, one may show that, if there is an SC decoder for the code defined by GM being $EM$, an SC-based decoding algorithm for the code defined by $M$ can be found such that both have the same error probability.
% }

We show in the following that the code with GM $E M'$, denoted by $C_{EM'} $, is at least as good as the code with GM $EM$, denoted by $C_{EM} $, in terms of  error probability  under SC decoding. 

Let  $U_1, \ldots, U_{N_1}$ be the information bits, $Y_1,Y_2$,$ \ldots$,$ Y_{N_2}$ be the channel output of $C_{EM} $ and $Y_1',Y_1'', Z_2 \ldots, Z_{N_2}$ be the channel output of $C_{EM'}$. 
Note that $(Y_2, \ldots, Y_{N_2})$ and  $(Z_2 \ldots, Z_{N_2})$ are identically distributed given the information bits. We assume SC decoding of the  $\hat{u_i}$'s based on the output in increasing order of the index $i$.

As mentioned in Step 1, to decide on $\hat{U}_i$, we may assume that the $(N_1-i+1)\times N_2$ submatrices of $EM$ and $EM'$ are used as GMs to encode information bits $U_i, \ldots, U_{N_1}$ and the codewords are transmitted through the channel. 
For $i=2$, the submatrices of $EM$ and $EM'$ are the same for the $(N_2-1)$ columns from the right. The first column of the submatrix of $EM'$ is a zero vector, and the first column of the submatrix of $EM$ is equal to the second column of that of $EM'$. 
% Therefore, the  bit-channel observed by $U_i$ is the same for both codes $C_{EM} $ and $C_{EM'}$. 
For $i>2$, the submatrices of $EM$ and $EM'$ are the same for the $(N_2-1)$ columns from the right, and the columns corresponding to $Ev_1, Ev_1'$, and $Ev_1''$ are zero vectors.  
Therefore, the  bit-channel observed by $U_i$ is the same for both codes $C_{EM} $ and $C_{EM'}$ for $i\geq 2$.
Hence it suffices to show that probability of error for the estimate of $U_1$  with $C_{EM'} $  is at least as good as that with $C_{EM}.$

% \color{blue}
Assume the BMS channel $W_{BMS}$ has the output alphabet $\mathcal{Y} = \mathset{a_1, a_1', \ldots, a_p, a_p', s_1, \ldots, s_q}$, such that $W_{BMS}( a_i | 0) =  W_{BMS}( a_i' | 1) $, $W_{BMS}( a_i | 1) =  W_{BMS}( a_i' | 0) $ for all $1 \leq i \leq p$, and that $W_{BMS}( s_i | 0) =  W_{BMS}( s_i | 1) $ for all $1 \leq i \leq q$. 
Without loss of generality, assume $W_{BMS}( a_i | 0)  \geq W_{BMS}( a'_i | 0) $ for all $i$.
Denote by $\mathcal{Y}_a$, $\mathcal{Y}'_a$ and $\mathcal{Y}_s$ the sets 
$\mathcal{Y}_a = \mathset{a_1, \ldots, a_p}$, $\mathcal{Y}'_a = \mathset{a_1', \ldots, a_p'}$, and $\mathcal{Y}_s =\mathset{s_1, \ldots, s_q}$, and a mapping $g: \mathcal{Y} \rightarrow \mathcal{Y}$ defined by $g(a_i) = a'_i, g(a'_i) = a_i, g(s_j) = s_j$ for $1\leq i \leq p,  1 \leq j  \leq q$. 

Consider two bit-channels: 
$W$ with input $U_1$ and output $Y_1, \ldots, Y_{N_2}$, and 
$W'$ with input $U_1$ and output $Y_1',Y_1''$,$Z_1$,$\ldots$, $Z_{N_2}$. Let $\mathbf{y}$ and $\mathbf{y}_{-1}$ denote $\mathbf{y}_1^{N_2}$ and $\mathbf{y}_2^{N_2}$, respectively. 
The probability of error of the first channel 
\begin{align}
P_e(W)= 
\frac{1}{2}
\sum_{\mathbf{y}}
\min( &W(\mathbf{y}|U_1= 1), W(\mathbf{y}|U_1= 0)  )  \nonumber \\
=\frac{1}{2}\sum_{\mathbf{y}_{-1}} \sum_{y_1 \in \mathcal{Y}}
\min(&W(Y_1= y_1,Y_2^{N_2}=\mathbf{y}_{-1}|0),\nonumber
W(Y_1 = y_1, y_2^{N_2}=\mathbf{y}_{-1}|1)  ) \nonumber \\
=\frac{1}{2}\sum_{\mathbf{y}_{-1}} \Big( 
\sum_{y_1 \in \mathcal{Y}_a}
\min(&W(Y_1= y_1,Y_2^{N_2}=\mathbf{y}_{-1}|0),\nonumber
W(Y_1 = y_1, y_2^{N_2}=\mathbf{y}_{-1}|1)  ) \nonumber \\
+
\sum_{y_1 \in \mathcal{Y}'_a}
\min(&W(Y_1= y_1,Y_2^{N_2}=\mathbf{y}_{-1}|0),\nonumber
W(Y_1 = y_1, y_2^{N_2}=\mathbf{y}_{-1}|1)  ) \nonumber \\
+
 \sum_{y_1 \in \mathcal{Y}_s}
\min(&W(Y_1= y_1,Y_2^{N_2}=\mathbf{y}_{-1}|0),\nonumber
W(Y_1 = y_1, y_2^{N_2}=\mathbf{y}_{-1}|1)  )   \Big) \\
= \frac{1}{2}\sum_{\mathbf{y}_{-1}}  
\sum_{a \in \mathcal{Y}_a \cup \mathcal{Y}_s} \big[ \min(&A_{a\vert 0}, A_{a\vert 1}) + \min(A_{g(a)\vert 0}, A_{g(a)\vert 1}) \big] 
% + \sum_{s \in \mathcal{Y}_s} \min(A_{s \vert 0}, A_{s \vert 1}) \Big),
\label{eq:PeW_eq1}
\end{align} 
where $A_{a\vert j} \triangleq W(Y_1= a,Y_2^{N_2}=\mathbf{y}_{-1}| U_1 =j)$ for $a \in \mathcal{Y}_a \cup \mathcal{Y}'_a$, and 
$A_{s \vert j} \triangleq \frac12 W(Y_1= s,Y_2^{N_2}=\mathbf{y}_{-1}| U_1 =j)$ for $s \in \mathcal{Y}_s$. 
Defining     $B_{i,j}(\mathbf{y}_{-1})=\Pr{(Y_2^{N_2}=\mathbf{y}_{-1} | U_1= i, U_2= j)}$, for $i,j =0,1$,
the terms in \eqref{eq:PeW_eq1} can be written as ($\mathbf{y}_{-1}$ is omitted for simplicity)
\begin{align*}
    &A_{a|0}=\frac12 P_a B_{0,0} +\frac12 P_{g(a)}B_{0,1} , 
    &A_{a|1}=\frac12 P_{g(a)} B_{1,0} +\frac12 P_{a}B_{1,1} ,  \\
    &A_{g(a)  |0}=\frac12 P_{g(a)} B_{0,0} +\frac12 P_{a}B_{0,1} , 
    &A_{g(a) |1}=\frac12 P_a B_{1,0} +\frac12 P_{g(a)}B_{1,1}, \\
    &A_{s |0}=\frac12 P_s ( B_{0,0} +B_{0,1} ), 
    &A_{s |1}=\frac12 P_s (B_{1,0} +B_{1,1} ) , 
\end{align*}
where $P_a \triangleq W_{BMS}( a | 0) $ for $a \in  \mathcal{Y}_a \cup \mathcal{Y}'_a$, and $P_s \triangleq \frac12 W_{BMS}(s | 0) $ for $s \in \mathcal{Y}_s$.

Similarly, we can  write  $P_e(W')$ as 
\begin{align}
    P_e(W')
    =\frac{1}{2}\sum_{\mathbf{y}_{-1}} \Big( 
    \sum_{a \in \mathcal{Y}_a} \big[ \sum_{b \in \mathcal{Y}_a} 
      \min(&A_{a, b\vert 0}, A_{a, b\vert 1}) + \min(A_{g(a), b\vert 0}, A_{g(a), b\vert 1}) ] \nonumber\\+
      \min(&A_{a, g(b)\vert 0}, A_{a, g(b)\vert 1}) + \min(A_{g(a), g(b)\vert 0}, A_{g(a), g(b)\vert 1}) \nonumber \\+
    \sum_{s \in \mathcal{Y}_s} \min(&A_{a, s \vert 0}, A_{a, s \vert 1})+\min(A_{g(a), s \vert 0}, A_{g(a), s \vert 1}) \big] \nonumber \\
    + \sum_{t \in \mathcal{Y}_s} \big[ \sum_{b \in \mathcal{Y}_a} 
      \min(&A_{t, b\vert 0}, A_{t, b\vert 1}) + \min(A_{t, g(b)\vert 0}, A_{t, g(b)\vert 1}) + 
    \sum_{s \in \mathcal{Y}_s} \min(A_{t, s \vert 0}, A_{t, s \vert 1}) \big] \Big) \nonumber \\
    \begin{split}
    = \frac{1}{2}\sum_{\mathbf{y}_{-1}} \Big( 
    \sum_{a \in \mathcal{Y}_{as}} \big[ \sum_{b \in \mathcal{Y}_{as}} 
      \min(&A'_{a, b\vert 0}, A'_{a, b\vert 1}) + \min(A'_{g(a), b\vert 0}, A'_{g(a), b\vert 1}) ] \\+
      \min(&A'_{a, g(b)\vert 0}, A'_{a, g(b)\vert 1}) + \min(A'_{g(a), g(b)\vert 0}, A'_{g(a), g(b)\vert 1}) \big] ,
    \end{split}
    \label{eq:PeW'terms}
\end{align}
    where the set $\mathcal{Y}_{as} = \mathcal{Y}_a \cup \mathcal{Y}_s$, and the terms $A_{a, b\vert j} \triangleq W(Y'_1= a,Y''_1= b,Y_2^{N_2}=\mathbf{y}_{-1}| U_1 =j)$ for $a, b \in \mathcal{Y}$, $A'_{a, b\vert j} \triangleq A_{a, b\vert j} $ , $A'_{s, b\vert j} \triangleq \frac12 A_{s, b\vert j}$, $A'_{a, s'\vert j} \triangleq \frac12 A_{a, s' \vert j}$ and $A'_{s, s'\vert j} \triangleq \frac14 A_{s, s' \vert j}$, for $a, b \in \mathcal{Y}_a\cup \mathcal{Y}'_a$, $s, s' \in \mathcal{Y}_s$. 

Again, it is helpful to rewrite the $A'_{a, b\vert j}$ in terms of $P_a, P_{g(a)}, P_b, P_{g(b)}, $ and $B_{i,j}$, such as  $A'_{a, b\vert 0} = \frac12 P_a P_b B_{0,0} + \frac12 P_a P_{g(b)} B_{0,1}$ and   $A'_{a, s \vert 0} = \frac12 P_a P_s (B_{0,0}+ B_{0,1})$.

We show $P_e(W') \leq P_e(W)$ by showing that for each $\mathbf{y}_{-1}$ and $a \in \mathcal{Y}_{as}$, the summand in \eqref{eq:PeW'terms} are smaller or equal to that in \eqref{eq:PeW_eq1}.
We may simplify the problem by noting that, given $\mathbf{y}_{-1}$ and $a \in \mathcal{Y}_{as}$, if for each $b$ the sum of the four minimum terms in \eqref{eq:PeW'terms} can be upper bounded by 
\[
 (P_b + P_{g(b)} ) \left[ \min(A_{a\vert 0}, A_{a\vert 1}) + \min(A_{g(a)\vert 0}, A_{g(a)\vert 1}) \right],
\]
the sum over $b \in \mathcal{Y}_{as}$ may be considered as a weighted sum, which is then bounded by $\min(A_{a\vert 0}, A_{a\vert 1}) + \min(A_{g(a)\vert 0}, A_{g(a)\vert 1})$. This would yield the inequality we want to show. 
Hence we may assume $P_b + P_{g(b}) =1$ and show the following inequality
\begin{equation}
    \begin{split}
       \min(A'_{a, b\vert 0}, A'_{a, b\vert 1}) + \min(A'_{g(a), b\vert 0}, A'_{g(a), b\vert 1})  +
      \min(A'_{a, g(b)\vert 0}, A'_{a, g(b)\vert 1}) + \min(A'_{g(a), g(b)\vert 0}, A'_{g(a), g(b)\vert 1}) \\ 
      \leq 
      \min(A_{a\vert 0}, A_{a\vert 1}) + \min(A_{g(a)\vert 0}, A_{g(a)\vert 1}). 
\end{split}\label{eq:PeW'bound_Aexpression}    
\end{equation}
We can normalize both sides of \eqref{eq:PeW'bound_Aexpression} by $1/ (P_a + P_{g(a)})$, and the goal is to show the following inequality
\begin{equation}
    \begin{split}
    &\min\left( P_a( P_b B_{0,0} +(1- P_b) B_{0,1}), (1- P_a) ( P_b B_{1,0} +(1- P_b) B_{1,1}) \right) + \\
    &\min\left( P_a( P_b B_{0,1} +(1- P_b) B_{0,0}), (1- P_a) ( P_b B_{1,1} +(1- P_b) B_{1,0}) \right) + \\ 
    &\min\left( (1- P_a)( P_b B_{0,0} +(1- P_b) B_{0,1}),  P_a ( P_b B_{1,0} +(1- P_b) B_{1,1}) \right) + \\
    &\min\left( (1- P_a) ( P_b B_{0,1} +(1- P_b) B_{0,0}),  P_a ( P_b B_{1,1} +(1- P_b) B_{1,0}) \right)  \\
    \leq 
    &\min\left( P_a B_{0,0} + (1- P_a) B_{0,1}, (1- P_a) B_{1,0} + P_a B_{1,1} \right) + \\
    &\min\left( (1- P_a) B_{0,0} +  P_a B_{0,1}, P_a B_{1,0} + (1- P_a) B_{1,1} \right), 
\end{split}\label{eq:PeW'bound_Bij}
\end{equation}
where $P_a, P_b \in [0.5, 1]$, and $B_{i,j} \in [0, 1]$. The inequality \eqref{eq:PeW'bound_Bij} can be numerically verified.

\subsection{Proofs for Subsection \ref{subsection:mainL2}}

\textit{Proof of Proposition \ref{prop:wMCloginN}: }
The geometric mean column weight of $G$ equals to that of  $G_2^{\otimes n}$, which is $2^{n/2}= \sqrt{N}$. The sparsity benchmark, as specified in equation \eqref{eq:logN'formula}, is $\log(N') =  N^{(1-\delta)/2} +\log N= (N^\frac{1}{2})^{(1-\delta) +o(1)} $ for sufficiently large $n$.  
Thus, for large $n$, \vspace{-2mm}
\begin{align*}
w_{GM}(n, G_2) &=  \sqrt{N}  = [\log(N')]^{\frac{1}{1-\delta+o(1)}}  = [\log(N')]^{1+\delta'+o(1)},
\end{align*}  where $\delta' \triangleq \frac{\delta}{1-\delta}$.
Note that $\delta'\rightarrow 0$ as $\delta \rightarrow 0$, and $\delta>1- \frac{\beta}{E(G_2)}$ can be made arbitrarily small by choosing $\beta$ close enough to $E(G_2) =\frac{1}{2}$. Hence, the proposition holds for any $\delta' >0.$ 
\endproof
\vspace{2mm}

\textit{Proof of Lemma \ref{lem:wMClogN_most}:}
Note that the  weights of the columns of $G$ have the same distribution as the variable $2^{X_1+\ldots +X_n}$, where $X_i \sim Ber(1/2), i =1, 2, \ldots, n$, are i.i.d.
By the strong law of large numbers, the logarithm of the column weights concentrate around that of the geometric mean column weight. 
To be more precise, the strong law of large numbers implies 
$\lim_{n\rightarrow \infty}\frac1n \sum_{i=1}^n X_i = \frac12 $ almost surely, or equivalently, $\frac{\sum_{i=1}^n X_i - \frac{n}{2}}{n}  = o(1)$ for all sufficiently large $n$ with probability $1$. 
Hence 
\begin{align*}
    \frac{2^{\sum_{i=1}^n X_i}}{2^{n/2}} = 2^{n \, o(1)} = (2^{n/2})^{o(1)} = \left([\log(N')]^{1+\delta'+o(1)} \right)^{o(1)} = [\log(N')]^{o(1)} \; \mbox{ w.p.} 1,
\end{align*} where the second-last equality is due to Proposition\,\ref{prop:wMCloginN}. Now, we have
\[
2^{\sum_{i=1}^n X_i} = [\log(N')]^{1+\delta'+o(1)}\times [\log(N')]^{o(1)} \leq 
[\log(N')]^{1+\delta''} \; \mbox{ w.p.} 1,
\] for all sufficiently large $n$.
Thus  the ratio of columns with weights exceeding $[\log(N')]^{1+\delta''}$ is vanishing as $n$ grows large.
\endproof

\vspace{2mm}
\textit{Proof of Lemma \ref{lemma:RinXn}:}
% \red{Fix proof in appendix (No need to introduce Wi)}
Let $W$ denote the column weight of a randomly selected column of $G_2^{\otimes n}$. Then we have 
\begin{align*}    
\gamma    =& \left(\Pr ( 2 w_{u.b.} \geq W >w_{u.b.} )\times 1 \right) 
    + \left(\Pr ( 3 w_{u.b.} \geq W >2 \cdot w_{u.b.} )\times 2\right) \\
    &+\ldots
     +\left(\Pr( 2^n \geq W > k_{max} \cdot w_{u.b.})\times k_{max}\right)  \nonumber\\
    =& \Pr (W >w_{u.b.}) +\Pr (W >2 w_{u.b.}) 
    + \ldots +\Pr (W > k_{max}\cdot w_{u.b.}). \nonumber
    % \label{eq:weightCount_2}.
\end{align*}
 Note that $\log W$ has the same distribution as $X(n) = X_1+X_2+ \ldots +X_n$, which is a Binomial$(n, \frac{1}{2})$ random variable. 
Hence, \begin{align*}
\gamma =& \Pr(2^{X(n)}  > w_{u.b.} )+ \ldots 
+\Pr (2^{X(n)} > k_{max}\cdot w_{u.b.}) \nonumber \\
= &\Pr(X(n)  > \log{(w_{u.b.}  )} )+ \ldots +\Pr (X(n) > \log{\left( k_{max}\cdot w_{u.b.} \right)}) \\
=&\sum_{k=1}^{k_{max}} \Pr(X(n)  > \log{(k\cdot w_{u.b.}  )} ).
\end{align*}
\endproof

\textit{Proof of Lemma \ref{lemma:RasSum_ai}:}
Recall that $X(n) \triangleq \sum_{i=1}^n X_i$ is an integer-valued random variable. Then the terms in \eqref{eq:RinXn} are grouped as 
\begin{align}
\gamma	
     =&\Pr{( X(n) >n_{lub})} +	\Pr{( X(n) >\log{2}+n_{lub})} +\Pr{( X(n) >\log{3}+n_{lub})} \nonumber \\
     &+\Pr{( X(n) >\log{4}+n_{lub})} +\ldots +\Pr{( X(n) >n-  n_{lub}+ n_{lub}) } \label{eq:R_lemmaPf_1}\\
= &(\Pr{( X(n) >n_{lub})} \times 2^0) +(\Pr{( X(n) >1+n_{lub})}\times 2^1) + (\Pr{( X(n) >2+n_{lub})}\times 2^2) +\ldots \nonumber \\
&+(\Pr{( X(n) >n-1)}\times 2^{(n-n_{lub}-1)}) \label{eq:R_lemmaPf_2}\\
 = & (\Pr{( X(n) \geq 1+n_{lub})} \times 2^0) +(\Pr{( X(n) \geq 2+n_{lub})}\times 2^1) + (\Pr{( X(n) \geq 3+n_{lub})}\times 2^2) +\ldots \nonumber \\
 &+ (\Pr{( X(n) \geq (n-n_{lub})+n_{lub})}\times 2^{(n-n_{lub}-1)})\nonumber\\
= & a_0+a_1+a_2 +\ldots +a_{n-n_{lub}-1},   \label{eq:weightCount_3}
\end{align} where \eqref{eq:R_lemmaPf_2} holds by noting that $X(n) > \log(j) +n_{lub}$ for some $2^m \leq j < 2^{m+1}$ if and only if $X(n) > m +n_{lub}$ for integer $m=  1,2, \ldots$.
\endproof

\vspace{2mm}
\textit{Proof of Lemma \ref{lemma:aiAsymp}:}
Using Sanov's Theorem (\cite[Thm 11.4.1]{cover2012elements}), the $a_i $ term  in Lemma \ref{lemma:RasSum_ai} is bounded as follows\footnote{Remark: In fact, even the polynomial term in the upper bound can be dropped since the set of distribution $E$, as defined in \cite{cover2012elements}, is convex.}: 
	 \begin{equation*}
	\frac{1}{(n+1)^2} 2^{-n D(P_i^*  || Q  )}\cdot 2^i \leq a_i \leq (n+1)^2  2^{-n D(  P_i^*  || Q   )}\cdot 2^i, 
	\end{equation*} where $P_i^*$ and $Q$ are the Ber$( \frac{i+1+n_{lub}}{n})$ and Ber$(\frac{1}{2})$ distributions, respectively.
	For $0\leq i \leq n-n_{lub}-1 $, we have the equation
	\[a_i \doteq 2^{-n D(P_i^*  || Q  )+i+1} = 2^{n (- D(P_i^*  || Q  ) +\alpha_i)  }.\]
Also, by \eqref{eq:wubDef}, 
	\[
	n_{lub}=
	\log{(w_{u.b.})}= \log{{N}^{ \frac{1}{2}+\epsilon' }} 
	= (\frac{1}{2}+\epsilon')n. \]
	
Hence, $P_i^*$ can be written as the Ber$(\frac{1}{2}+\epsilon'+\alpha_i)$ distribution. 
\endproof
\vspace{3mm}

\textit{Proof of Proposition \ref{Thm:ratelossAndEpsi'}:}
%	Now, the ratio between the number of extra columns after splitting and the number of columns in $G$ 
The ratio $\gamma$ is bounded by \begin{align}
\max_{i}a_i \leq \gamma \leq n\cdot \max_{i}a_i.
\label{eq:Rbounds}
\end{align}
Therefore, the  rate  loss $\gamma$ of the code corresponding to $G'$ instead of $G$ either approaches 0, when $\max_i \lambda(\epsilon', \alpha_i)<0, $ or infinity, when there's one $i$ such that $ \lambda(\epsilon', \alpha_i)>0$. 
The exponent can be analyzed as follows:
 	\begin{align*}
 	&\lambda(\epsilon', \alpha_i)= \alpha_i-(\frac{1}{2}+\epsilon'+\alpha_i )\cdot \log{(2 (\frac{1}{2}+\epsilon'+\alpha_i ) )} - (\frac{1}{2}-\epsilon'-\alpha_i )\cdot 
 	\log{(2 (\frac{1}{2}-\epsilon'-\alpha_i ) ) }\\
    &= \alpha_i- \frac{1}{2}\log{\left[  (1+2\epsilon'+2\alpha_i)(1-2\epsilon'-2\alpha_i) \right]} -(\epsilon'+\alpha_i)\log{\frac{1+2\epsilon'+2\alpha_i}{1-2\epsilon'-2\alpha_i} }.
 	\end{align*}
 	Consider $\lambda(\epsilon', \alpha)$ as a function of $\alpha$ over the interval $[\frac1n, \frac12 -\epsilon']$. Then its first and second derivatives with respect to $\alpha$ are as follows:
 	\begin{align}
 	&\frac{\partial \lambda(\epsilon', \alpha)}{\partial \alpha} = 1-\log{\frac{1+2\epsilon'+2\alpha}{1-2\epsilon'-2\alpha} }, \label{eq:lambdaDeri}
 	\end{align}
\vspace{-5mm} 	
    \begin{align}
 	\frac{\partial^2 \lambda(\epsilon', \alpha)}{\partial^2 \alpha}\nonumber = -\frac{\partial \log ( 1+2\epsilon'+2\alpha)}{\partial \alpha} +\frac{\partial \log( 1-2\epsilon'-2\alpha)}{\partial \alpha}
 	= -\frac{2}{\ln{2}}\big( \frac{1}{ 1+2\epsilon'+2\alpha}+ \frac{1}{ 1-2\epsilon'-2\alpha} \big) <0,
 	\end{align} for any $\frac{1}{n} \leq \alpha \leq \frac{1}{2}-\epsilon'$. Thus, for any fixed $\epsilon'$, $\lambda(\epsilon', \alpha)$ is a concave function of $\alpha$ and has maximum when \[\frac{\partial \lambda(\epsilon', \alpha)}{\partial \alpha} =0 . \]
 	From \eqref{eq:lambdaDeri}, the above equality holds if and only if $\alpha= \frac{1}{6}-\epsilon'$ and attains the maximum value when
 	\begin{equation}
 	\max_{\alpha\in [\frac1n, \frac12 -\epsilon']}\lambda(\epsilon', \alpha) =
 	\lambda(\epsilon', \frac{1}{6}-\epsilon')= \frac{1}{6}-\frac{2}{3}\log{\frac{4}{3}}-\frac{1}{3}\log{\frac{2}{3}} -\epsilon' = \epsilon^*- \epsilon'.
 	\label{maxLambda}
 	\end{equation}
 	If $\epsilon'> \epsilon^* $, $ \lambda(\epsilon', \alpha_i) \leq \max_{\alpha}\lambda(\epsilon', \alpha) = \epsilon^*- \epsilon' <0$ for any $0\leq i \leq n-n_{lub}-1 $. Hence, by \eqref{eq:aiApprox} and \eqref{eq:Rbounds}, $\gamma$, which is upper bounded by $ n (n+1)^2 \hspace{1mm}  2^{n(\epsilon^*- \epsilon')} \doteq 2^{n(\epsilon^*- \epsilon')} $, approaches $0$ exponentially fast.

 	Alternatively, if $\epsilon' < \epsilon^* $, by the continuity of $\lambda(\epsilon', \alpha)$ in $\alpha$, for sufficiently large $n$, there is $\alpha_i$ for some $0\leq i \leq n-n_{lub}-1$ such that $\lambda(\epsilon', \alpha_i) \geq \frac{\max_{\alpha}\lambda(\epsilon', \alpha)}{2} >0$. Then
 	$\gamma$, which is bounded from below by $\frac{1}{(n+1)^2} 2^{n \lambda(\epsilon', \alpha_i)} \doteq 2^{n \lambda(\epsilon', \alpha_i)}$ in this case, approaches infinity exponentially fast.
 \endproof
\vspace{2mm}

\textit{Proof of Corollary \ref{Coro:ratelosAndEpsi}:}
From \eqref{eq:epsilon'def}, we have
\begin{equation*}
    \epsilon'= (1+\epsilon ) \left( \frac{1-\delta}{2}  +o(1) \right)-\frac{1}{2} \,=\, \frac{\epsilon}{2}-\frac{\delta}{2}(1+\epsilon)+o(1).
\end{equation*} 
Since $\delta>0$ can be chosen arbitrarily small, the conditions in Proposition \ref{Thm:ratelossAndEpsi'} are expressed in terms of $\epsilon$ as follows:
\begin{align*}
    \gamma\to 0 \mbox{ exponentially fast  } \iff  \epsilon' >\epsilon^* \iff \epsilon > 2\epsilon^* \\
    \gamma\to \infty \mbox{ exponentially fast } \iff  \epsilon' < \epsilon^* \iff \epsilon < 2\epsilon^* 
\end{align*} \endproof

\textit{Proof of Theorem \ref{Coro:LDGMwithUniformBound}:}
Since the code corresponding to  $G'$ uses a submatrix of $G'$ as its GM, the column weights of this submatrix are upper bounded by $w_{u.b.}$ as well. By Proposition \ref{prop:codeWithG'} the probability of error of this code is upper bounded by that of the code corresponding to  $G$.
From Lemma \ref{Lemma:newG_PeBound}, the code corresponding to  $G$ is capacity achieving, and, from Corollary \ref{Coro:ratelosAndEpsi}, the rate loss $\gamma$ goes to $0$ as $n$ grows large. 
\endproof

\subsection{Proofs for Subsection \ref{sub:DRSalg}}

In order to understand the DRS algorithm's effect on $G= G_2^{\otimes n} $, we first study how the order of two special Kronecker product operations affects the number of output vectors.
We present the following Lemmas \ref{lem:LR=RL} and \ref{lem:DRSforTwoVectors} toward the proof of Lemma \ref{lem:rateLossIndepOfOrder}.

\begin{lemma} \label{lem:LR=RL}
Let a column vector $v$ and a column weight threshold $w_{u.b.}$ be given. Then the outputs of the DRS algorithm for
\begin{align*}
(v \otimes     [1,1]^t )
    % \begin{bmatrix}
    %       1 \\1
    % \end{bmatrix})
    \otimes [0,1]^t
    % \begin{bmatrix}
    %       0 \\1
    % \end{bmatrix}
% \end{align*}    
%     and 
% \begin{align*}
\; \mbox{and} \;
(v \otimes     [0,1]^t )
    % \begin{bmatrix}
    %       0 \\1
    % \end{bmatrix})
    \otimes  [1,1]^t
    % \begin{bmatrix}
    %       1 \\1
    % \end{bmatrix}
\end{align*}    
contain the same number of vectors.
\end{lemma}
\proof{} 
% \textit{Proof for Lemma \ref{lem:LR=RL} }
The input vectors can be denoted by:
\begin{align*}
(v \otimes     [1,1]^t)
    % \begin{bmatrix}
    %       1 \\1
    % \end{bmatrix})
    \otimes [0,1]^t
    % \begin{bmatrix}
    %       0 \\1
    % \end{bmatrix}
    % = \begin{bmatrix}
    %       \zero \\ \zero \\ v \\v 
    % \end{bmatrix} 
    \equiv v_{LR}
% \end{align*}    
%     and 
% \begin{align*}
\; \mbox{ and } \;
(v \otimes     [0,1]^t )
    % \begin{bmatrix}
    %       0 \\1
    % \end{bmatrix})
    \otimes [1,1]^t
    % \begin{bmatrix}
    %       1 \\1
    % \end{bmatrix}
    % = \begin{bmatrix}
    %       \zero \\ v \\ \zero \\v
    % \end{bmatrix}
    \equiv v_{RL}.
\end{align*}    

We note that $ 2 w_H(v) =w_H(v_{LR}) = w_H(v_{RL}) $, and prove the lemma in two  cases:
\begin{enumerate}
    \item $2 w_H(v) \leq w_{u.b.}$:
        In this case, the algorithm will not split either $ v_{LR}$ or $v_{RL}$. Both outputs contain exactly one vector.
    
    \item $2 w_H(v) > w_{u.b.}$:
    Let $n_{DRS}(v)$ denote the number of column vectors the DRS algorithm returns when it is applied to $v$.
    
    For $ v_{LR}$, the DRS algorithm observes $\mathbf{x}_h=\zero$, hence the number of output vectors is the same as the size of \textproc{DRS-Split}($w_{u.b.}, (v^t, v^t)^t $) (see Section \ref{sub:DRSalg}). 
    With $2 w_H(v) > w_{u.b.}$, the size of \textproc{DRS-Split}($(v^t, v^t)^t $) is the sum of the sizes of  $Y_h =$ \textproc{DRS-Split}($w_{u.b.}, \mathbf{x}_h=v$) and $Y_t =$ \textproc{DRS-Split}($w_{u.b.}, \mathbf{x}_t=v$). By assumption,  $\abs{Y_h}=\abs{Y_t}= n_{DRS}(v)$, giving     $n_{DRS}(v_{LR})=2\, n_{DRS}(v) $.

    For $ v_{RL}$, the number of vectors in the DRS algorithm output is the sum of the sizes of two sets $Y_h =$ \textproc{DRS-Split}($w_{u.b.}, \mathbf{x}_h=(\zero^t, v^t)^t $) and $Y_t =$ \textproc{DRS-Split}($w_{u.b.}, \mathbf{x}_t=(\zero^t, v^t)^t$).
    It is easy to see that $\abs{Y_h}=\abs{Y_t}= \abs{
     \textproc{DRS-Split}(w_{u.b.}, v)
    }$, which equals $n_{DRS}(v)$.
    Thus, $n_{DRS}(v_{RL})=2\, n_{DRS}(v) $.

\end{enumerate}
\endproof

The next lemma shows the effect of the DRS algorithm from a different perspective. If there are two vectors with the same column weights, and numbers of vectors of the DRS  algorithm outputs are identical when they are the inputs, the properties will be preserved when they undergo some basic Kronecker product operations. 

\begin{lemma}\label{lem:DRSforTwoVectors}
    Let $u_1$ and $u_2$ be two vectors with equal Hamming weights. Assume, for a given $w_{u.b.}$, the DRS algorithm splits $u_1$ and $u_2$ into the same number of vectors. Then the DRS algorithm also returns the same number of vectors for $u_1 \otimes [1,1]^t$ and $u_2\otimes [1,1]^t$, as well as for $u_1 \otimes [0,1]^t$ and $u_2\otimes [0,1]^t$. 
\end{lemma}
\proof
We first discuss the case when $u_1 \otimes [1,1]^t$ and $u_2\otimes [1,1]^t$ are processed by the DRS algorithm. 
If $2w_H(u_1)= 2w_H(u_2) \leq w_{u.b.}$, no splitting is done. 
If  $2w_H(u_1)= 2w_H(u_2) > w_{u.b.}$,  the size of the DRS algorithm output for the input $u_1 \otimes [1,1]^t$ is the sum of the sizes of $Y_h =$ \textproc{DRS-Split}($w_{u.b.}, \mathbf{x}_h=u_1$) and $Y_t =$ \textproc{DRS-Split}($w_{u.b.}, \mathbf{x}_t=u_1$), both of which are $n_{DRS}(u_1)$. 
The size of the output for the input $u_2 \otimes [1,1]^t$ can be found in a similar way to be $2 n_{DRS}(u_2)$.
Note that $n_{DRS}(u_1)=n_{DRS}(u_2)$ by assumption. Therefore, the sizes of the outputs of the DRS algorithm, when $u_1 \otimes [1,1]^t$ and $u_2\otimes [1,1]^t$ are the inputs, are equal.

Similarly, one can easily show that when $u_1 \otimes [0,1]^t$ and $u_2\otimes [0,1]^t$ are processed by the DRS algorithm, the number of output columns are equal.
\endproof
\vspace{2mm}

\textit{Proof of Lemma \ref{lem:rateLossIndepOfOrder}:}
Suppose that there is an index $i$ such that $(s_i, s_{i+1} )= (+,-)$.  
Let $v^{(i+1)}$ and $(v^{(i+1)})'$ be defined by (\ref{eq:vi_recursion}) with sequences $(s_1, \ldots, s_{i-1}, s_i= +, s_{i+1} = -)$ and $(s_1, \ldots, s_{i-1}, s'_i= -, s'_{i+1} = +)$, respectively.
We note that 
\begin{align*}
    v^{(i+1)} = \left( v^{(i-1)} \otimes     
    [0,1]^t
    \right) \otimes 
    [1, 1]^t
    \; \mbox{and } \;
    (v^{(i+1)})' = \left( v^{(i-1)} \otimes 
    [1,1]^t
    \right) \otimes 
    [0,1]^t.
\end{align*}
Lemma \ref{lem:LR=RL} shows that the DRS algorithm splits $v^{(i+1)}$ and $(v^{(i+1)})'$ into the same number of columns.
Furthermore, Lemma \ref{lem:DRSforTwoVectors} shows that the number of output vectors of the DRS algorithm for 
  $  v^{(n)} = [v^{(i+1)}]^{(s_{i+2}, \ldots, s_{n})} $
and 
     $ (v^{(n)})' = [ (v^{(i+1)})' ]^{(s_{i+2}, \ldots, s_{n})}  $
are equal. 
    
Therefore, an occurrence of $(s_i, s_{i+1} )= (+,-)$ in a sequence can be replaced by $(s_i, s_{i+1} )= (-,+)$ without changing the number of output vectors of the DRS algorithm. 
Since any sequence $(s_1, s_2, \ldots, s_{n} )$  with $n_1$ minus signs and $n_2 $ plus signs can be permuted into ${(s'_1, s'_2, \ldots s'_{n})}, $ where $s'_i= -$ for $i\leq n_1$ and $s'_i= +$ for $i> n_1$, by repeatedly replacing any occurrence of $(+,-)$ by $(-,+)$, the above arguments show 
$n_{DRS}(v^{(n)}) =n_{DRS}(v^{(s'_1, s'_2, \ldots s'_{n})})$ always holds.
Hence, the size of DRS algorithm output for $v^{(n)}$ depends only on the values $n_1$ and $n_2$.  
\endproof

\vspace{2mm}

\textit{Proof for Proposition \ref{prop:DRSrateloss}:}
First note that there is a bijection between $\mathset{-,+}^n$ and the columns of $G_2^{\otimes n}$ as follows.  
For each $\mathbf{s}= (s_1, \ldots, s_n)\in \mathset{-,+}^n$, there is exactly one column of $G_2^{\otimes n}$ in the form $1^{(\mathbf{s})}$ (see equation \eqref{eq:vi_recursion} and the paragraph following it, where we use $v= [1] \in \Ftwo$).
The term $\gamma$ can be characterized as follows:
\[
\gamma = \Big[ \frac{1}{N}\sum_{\mathbf{s}\in \mathset{-,+}^n} n_{DRS}(1^{(\mathbf{s})}) \Big] -1.
\]
By Lemma \ref{lem:rateLossIndepOfOrder}, the terms in the summation can be grouped according to the number of minus and plus signs in the sequence. Hence, 
\[
\gamma = \Big[\frac{1}{N}\sum_{i=0}^n \binom{n}{i} n_{DRS}(1^{(s_1= -, \ldots, s_i= -, s_{i+1}= +, \ldots, s_n= + ) } ) \Big] -1.
\]
Let $u_i$ denote the vector $1^{(s_1, \ldots, s_n)}$ with $s_l =-$ for $l\leq i$ and $s_l= +$ for $l> i$. 
Without loss of generality, let $n\lambda \in \N $.
For $i\leq n\lambda$, the Hamming weight of $u_i$ is $2^i \leq 2^{n\lambda} = w_{u.b.}$. Hence, $n_{DRS}(u_i) =1$. 
For $i > n\lambda$,  $u_i$ is split into $2^{i-n\lambda}$ vectors, each of which having weight equal to $2^{n\lambda}$. Therefore, $n_{DRS}(u_i)= 2^{i-n\lambda}$. 

The term $\gamma$ can be written as follows:
\begin{align}
    \gamma &= \sum_{i=0}^{n\lambda} \frac{1}{N}\binom{n}{i} +  \sum_{i=n\lambda+1}^n \frac{1}{N}\binom{n}{i}2^{i-n\lambda} -1  
    = \sum_{i=n\lambda+1}^n a_i , \label{eq:DRS_R_ai}
    \; 
\end{align} where $ a_i \triangleq \frac{1}{N}\binom{n}{i}(2^{i-n\lambda}-1)$.
Now, let $\alpha= i/n$. Since $i > n\lambda$ for each summand $a_i$, we consider $\alpha >\lambda >\frac{1}{2}$ in the following calculations. 
The term $a_i$ can be written as 
\begin{align}
    a_i = a_{n\alpha} = 2^{-n}\binom{n}{n\alpha} 2^{n\alpha-n\lambda +o(1)} 
    =  2^{-n}2^{n h_b(\alpha) +o(1)} 2^{n\alpha-n\lambda +o(1)} 
    = 2^{n\cdot f(\alpha, \lambda) +o(1)}, \label{eq:DRS_ai_f} \; 
\end{align}
where the third equality is due to an asymptotic approximation of the binomial coefficient, and $f(\alpha,\lambda) \triangleq h_b(\alpha)+\alpha-\lambda -1$. 

 Consider $f(\alpha,\lambda)$ as a function of $\alpha$ over the interval $[0,1]$. We find its first and second derivatives with respect to $\alpha$ as follows:
\begin{align}\label{eq:fDeri}
     &\frac{\partial f(\alpha,\lambda)}{\partial \alpha} = 1-\log{\frac{\alpha}{1-\alpha} }, \\
     &\frac{\partial^2 f(\alpha,\lambda)}{\partial^2 \alpha}\nonumber = 
     -\frac{1}{\ln{2}}\big( \frac{1}{ \alpha}+ \frac{1}{ 1-\alpha} \big) <0, 
     \mbox{ for any $0< \alpha <1$.}
\end{align}  
 Thus, for any fixed $\lambda$, $f(\alpha,\lambda)$ is a concave function of $\alpha$ and has local maximum when 
 $\frac{\partial f(\alpha,\lambda)}{\partial \alpha} =0 . $
From \eqref{eq:fDeri}, the equality holds if and only if  $\alpha= \frac{2}{3}$, and the maximum is 
 \begin{equation}
 \sup_{\alpha\in (0,1)}f(\alpha,\lambda) =
 	h_b\left(\frac{2}{3}\right)+\frac{2}{3}-\lambda -1
 = \lambda^*- \lambda .
 \label{eq:max_f_overalpha}
 \end{equation}
Also, when $\alpha =0$, $f(0,\lambda) = -\lambda -1$, and when  $\alpha =1$, $f(1,\lambda) = -\lambda$.  
Hence, \[
\sup_{\alpha\in [0,1]}f(\alpha,\lambda) = \lambda^*- \lambda. \]
 When $\lambda> \lambda^* $, we know $f(\frac{i}{n} ,\lambda) \leq \sup_{\alpha}f(\alpha,\lambda)  <0 $ for all integers $0 \leq i \leq n$, and equation \eqref{eq:DRS_ai_f} implies that  $a_i \to 0$ exponentially fast for each $i$. 
 Equation \eqref{eq:DRS_R_ai} then shows that $\gamma$ also vanishes exponentially fast in $n$.
 \endproof
 \vspace{3mm}

\subsection{Proofs for Subsection \ref{subsection:LCdec_BEC}}\label{appendix:proof_LCdec_BEC}

\textit{Proof for Lemma \ref{lem:DRS_Z_BEC}:} 
We show the claim by proving the following:
when we encode the source bits according to $G'$, the bit-channels observed by the source bits are BECs and that the erasure probabilities are less than or equal to those when $G$ is used.  
We use proof by induction on  $n$.
For ease of notation, we use $M'$ to denote \textproc{DRS}$(M)$ for a given matrix $M$ in this proof.

For $n =1$, if $G_2= G_2'$, we naturally have $Z_{G_2'}^{(s_1)} = Z_{G_2}^{(s_1)}$ for $s_1\in \mathset{-,+}$.
If $G_2 \neq G_2'$, the latter must be
$\begin{bmatrix}
		1 & 0 &0 \\
		0 & 1 & 1\\
\end{bmatrix}$, corresponding to the encoding block diagram in Figure\,\ref{fig:SplitG2}, where the solid black circles indicate a \textit{split} of the XOR operation, i.e., the two operands of the original XOR operation are transmitted through two copies of channel $W$. 
\begin{figure}[h]
    \centering
\includegraphics[width=0.3\textwidth]{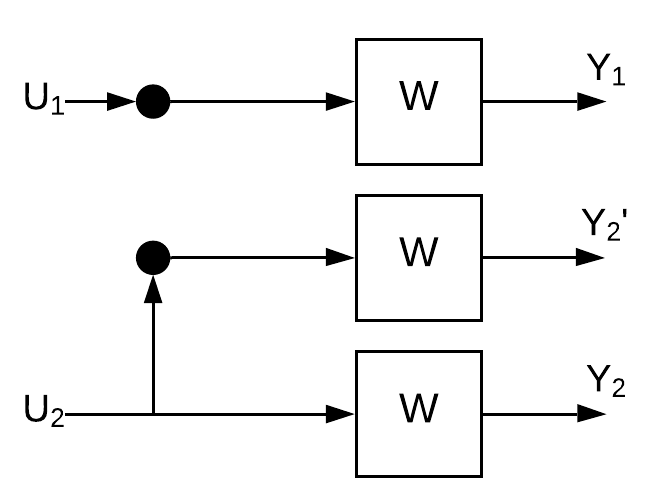}
\caption{Encoding Block for $G_2'$}\label{fig:SplitG2}

 \vspace{-10mm}
\end{figure}
The bit-channels observed by $U_1$ and $U_2$, denoted by $W^\boxminus$ and $W^\boxplus$,  are BECs with erasure probability $\epsilon$ and $\epsilon^2$, respectively. Note that the Bhattacharyya parameters satisfy the following: 
\begin{align*}
    Z_{G_2'}^{(-)}&= Z(W^\boxminus) = 1- \epsilon  \leq Z_{G_2}^{(-)} = Z(W^-) = 1-2\epsilon+\epsilon^2,  \\
    Z_{G_2'}^{(+)}&= Z(W^\boxplus) = 1- \epsilon^2 = Z_{G_2}^{(+)} = Z(W^+).  
\end{align*}

Suppose that, for a fixed $w_{u.b.}$, the claim holds for all $n\leq m$ for some integer $m\geq 1$.
Let $B_m$ denote the encoding block corresponding to the generator matrix $(G_2^{\otimes m})'$, with inputs $U_1, \ldots, U_{2^m}$ and encoded bits  $X_1, \ldots, X_{f(m)}$, where $f(m)$ is the number of columns in $(G_2^{\otimes m})'$. 
Using the matrix notation, the relation between the 
input bits and encoded bits is:
\begin{equation*}\label{eq:WmMatrixNotation}
    (U_1, \ldots, U_{2^m}) (G_2^{\otimes m})' = (X_1, \ldots, X_{f(m)}). 
\end{equation*}
    
For $n= m+1$, the matrix $(G_2^{\otimes m+1})'$ is associated with $(G_2^{\otimes m})'$ as follows:
\begin{equation}\label{eq:DRS_G_recursion}
    (G_2^{\otimes m+1})' = \textproc{DRS}\Big(
\begin{bmatrix}
		(G_2^{\otimes m})' & \zero \\
		(G_2^{\otimes m})' & (G_2^{\otimes m})'
\end{bmatrix}\Big).
\end{equation}
Since $(G_2^{\otimes m})'$ consists of the outputs of the DRS algorithm, the columns in the right half of the input matrix in equation \eqref{eq:DRS_G_recursion} remain unaltered in the output. 
For the columns in the left half, they are of the form $[v^t, v^t]^t$ for some column $v$ of $(G_2^{\otimes m})'$.
If $2 w_H(v) > w_{u.b.}$, the outputs of the DRS algorithm are $[\zero^t, v^t]^t$ and $[v^t, \zero^t]^t$ because the vector $v$ must have weight no larger than the threshold. 
If  $2 w_H(v) \leq w_{u.b.}$, the algorithm leaves the vector unchanged.
We may represent the encoding block $B_{m+1}$ as in Figure\,\ref{fig:BmBlock}, where it is assumed that the $j$-th column of the input matrix in \eqref{eq:DRS_G_recursion} is halved by the DRS algorithm.

\begin{figure}
     \centering
     \begin{subfigure}[b]{0.7\textwidth}
         \centering
         \includegraphics[trim=2cm 2cm 5.5cm 2cm, clip,  width= \textwidth]{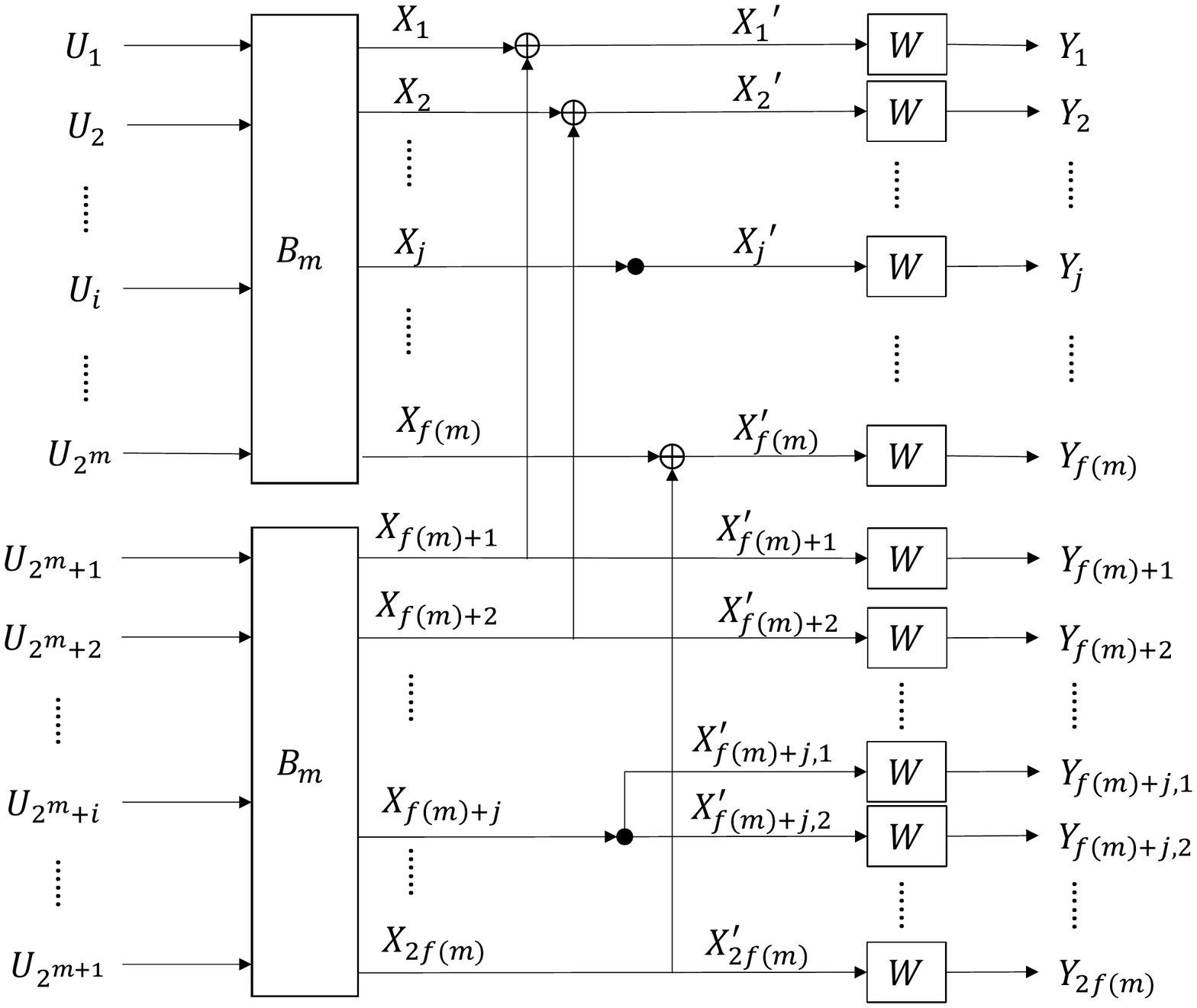}
         \caption{Encoding Block $B_{m+1}$}
        \label{fig:BmBlock}
     \end{subfigure}
    %  \hfill
    \vspace{5mm}
    
     \begin{subfigure}[b]{0.8\textwidth}
         \centering
         \includegraphics[trim=2cm 2cm 2cm 2cm, clip, width= \textwidth]{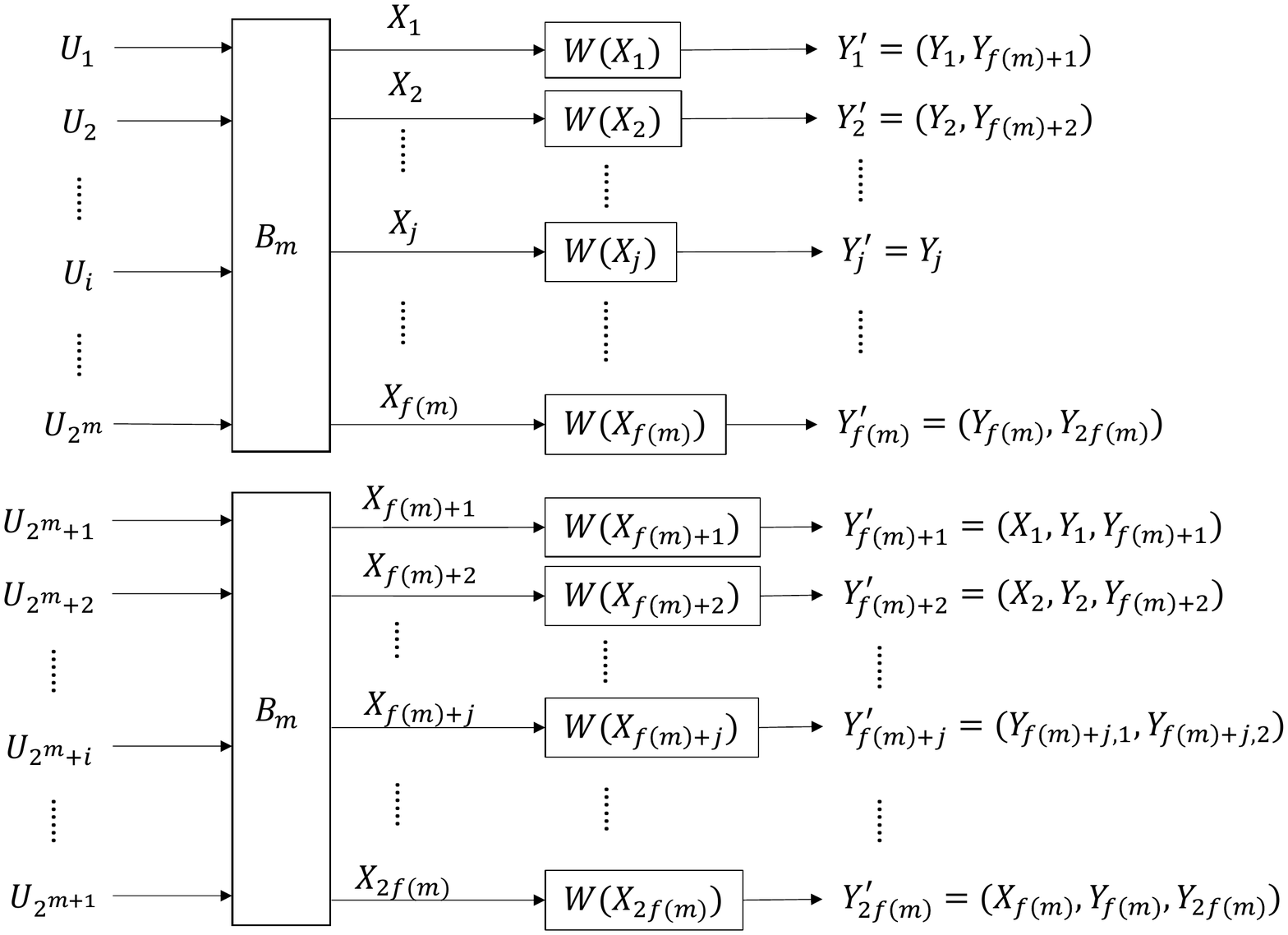}
        \caption{Equivalent Encoding Block $B_{m+1}$}
    \label{fig:BmBlock_Equvi}
     \end{subfigure}
\caption{}
        \label{fig:BmDiagrams}
\end{figure}

The erasure probabilities for the bit-channels observed by $X_i$, denoted here as $W(X_i)$, are less than or equal to $2\epsilon-\epsilon^2$ for $1\leq i \leq f(m)$, and are equal to $\epsilon^2$ for $f(m)+1 \leq i \leq 2 f(m)$, respectively.
Hence we may replace the XOR operations to the right of the $X_i$'s as well as the transmission over $W$'s by BECs $W(X_1), \ldots,W(X_{f{(m)}}),$ $W(X_{f{(m)}+1}), \ldots, W(X_{2f{(m)}})$, as in Figure\,\ref{fig:BmBlock_Equvi}.

One may observe that the erasure probability for the bit-channel observed by $U_i$ is a non-decreasing function of those of $W(X_1), \ldots, W(X_{f{(m)}})$ for  $i \leq 2^m$, and of $W(X_{f{(m)}+1}), \ldots,$  $W(X_{2f{(m)}})$ for  $i >2^m$. 
So for $i \leq 2^m$, we have  
\begin{align*}
    Z_{(G_2^{\otimes m+1})'}\left(U_i \,| \, Z(W)= \epsilon \right) 
    % &=Z\big(W(U_i) \big| Z(W(X_j))\leq  2\epsilon- \epsilon^2  \mbox{ for } 1\leq j \leq f(m) \big)   \\
    &\leq Z\left[ W(U_i)\, |\,  Z(W(X_j))= 2\epsilon- \epsilon^2, \mbox{ for } 1\leq j \leq f(m) \right]  \\
    &= Z_{(G_2^{\otimes m})'} \left(U_{i} \,|\, Z(W)= 2\epsilon- \epsilon^2 \right)\\
    &\leq Z_{G_2^{\otimes m}} \left(U_{i}\, | \, Z(W)= 2\epsilon- \epsilon^2 \right) \\
    &= Z_{G_2^{\otimes m+1}} \left(U_i \,| \,  Z(W)= \epsilon \right), 
\end{align*}
    where the first inequality is due to $Z(W(X_j))\leq  2\epsilon- \epsilon^2  \mbox{ for } 1\leq j \leq f(m)$ and the second inequality follows from the hypothesis of the induction. 
    
    Similarly, for  $i >2^m$, we have 
\begin{align*}
    Z_{(G_2^{\otimes m+1})'}\left(U_i  \,| \, Z(W)= \epsilon  \right) 
    &= Z\left[ W(U_i) | Z(W(X_j))= 2\epsilon- \epsilon^2 ,\mbox{ for } f(m)+1\leq j \leq 2f(m)  \right]    \\
    &= Z_{(G_2^{\otimes m})'} \left(U_{i-f(m)}  \,| \, Z(W)= 2\epsilon- \epsilon^2  \right)\\
    &\leq Z_{G_2^{\otimes m}} \left(U_{i-f(m)}  \,| \, Z(W)= 2\epsilon- \epsilon^2  \right) \\
    &= Z_{G_2^{\otimes m+1}} \left(U_i  \,| \, Z(W)= \epsilon  \right). 
\end{align*}    
Hence the inequality holds when $n= m+1$ as well. 
\endproof

\subsection{Proofs for Subsection \ref{subsection:LCdec_BMS}}\label{appendix:proof_LCdec_BMS}

\textit{Proof of Proposition \ref{prop:ADRScompl}:}
% \proof
First note that each XOR operation at the $j$-th recursion can be associated with exactly one vector $\mathbf{s} =(s_1, s_2, \ldots,  s_n) \in \mathset{-,+}^n$ and $s_j = -$. 
For example, assume that the number of minus signs in $\mathbf{s}$, denoted as $m(\mathbf{s})$, is larger than $n_{lub} =\log w_{u.b.} = n\lambda$, and let $\tau =\tau(\mathbf{s})$ be the index such that $m(s_{\tau}, s_{\tau+1},\ldots, s_n)= n_{lub} $ and $s_{\tau} =-$. 
Then for each index $i$ in the set $\mathset{k: 1 \leq k < \tau, s_k =-}$, there is a bijection between the pair $(\mathbf{s}, i)$ and an XOR operation at the $i$-recursion which is split and modified in the A-DRS scheme. 
Hence, the extra complexity of the SC decoder for the A-DRS scheme, compared to that of the SC decoder for the code based on $G_2^{\otimes n}$, is given by
\begin{align}\label{eq:ADRScomp_bound}
    \sum_{l=1}^{n-n_{lub}+1} & \abs{\mathset{\mathbf{s} \in \mathset{-,+}^n: \tau(\mathbf{s} ) > l , s_l =- } }  2(2^{l+1} -2)c \nonumber \\
    = &\sum_{l=1}^{n-n_{lub}+1} \sum_{k=l+1}^{n-n_{lub}+1} \abs{\mathset{\mathbf{s} \in \mathset{-,+}^n: \tau(\mathbf{s} ) =k , s_l =-  } }  2(2^{l+1} -2)c \nonumber \\
    =  &\sum_{k=1}^{n-n_{lub}+1} \binom{n-k+1}{n_{lub} } 2^{k-2} \sum_{l=1}^{k-1}    2(2^{l+1} -2)c \nonumber \\
    \leq & 4c\sum_{k=1}^{n-n_{lub}+1} \binom{n-k+1}{ n_{lub}} 2^{k-2}  \sum_{l=1}^{k-1}   2^{l}\nonumber  \\
    = & 4c\sum_{k=1}^{n-n_{lub}+1} \binom{n-k+1}{ n_{lub}}  2^{k-2}  (2^{k}-2) \nonumber \\
    \leq &4c\sum_{k=0}^{n-n_{lub}} \binom{n-k}{ n_{lub}}  2^{2k}.
\end{align}
Now, let  $\alpha = \frac{k}{n} \in [0, 1-\lambda]$. Using Stirling's approximation we have
\begin{align*}
    &\binom{n-k}{ n_{lub}}  2^{2k} =  \binom{n(1-\alpha)}{ n\lambda }  2^{2\alpha n} \approx 2^{n(1-\alpha) h_b(\frac{\lambda }{1-\alpha}) +2\alpha n},
    \end{align*} for all sufficiently large $n$.
It suffices to assume that $\lambda < \frac{3}{4}$. For $\lambda \geq \frac34$, note that fewer XOR operations are split and modified, and that the resulting additional decoding complexity is not larger than when $\lambda < \frac34$ is used. 
Now let $f(\alpha, \lambda) = (1-\alpha) h_b(\frac{\lambda}{1-\alpha}) +2\alpha$.
We find its maximum, for a given $\lambda$, by solving 
\begin{align*}
0 &=\frac{\partial}{\partial \alpha}f(\alpha, \lambda)\\
&=\frac{\partial}{\partial \alpha}\left[ -(1-\alpha) (\frac{\lambda}{1-\alpha} \log \frac{\lambda}{1-\alpha})- 
(1-\alpha)(1-\frac{\lambda}{1-\alpha}) \log (1-\frac{\lambda}{1-\alpha}) +2\alpha
\right] \\
&=\frac{1}{\ln 2} \left[
-\ln(1-\alpha)+ \ln(1-\alpha -\lambda)
\right] +2,
\end{align*}
which is true if and only if $\alpha = 1- \frac{4}{3}\lambda$.
And note that \[
\frac{\partial^2}{\partial \alpha^2}f(\alpha, \lambda) = \frac{1}{\ln 2}\left[ \frac{1}{1-\alpha} - \frac{1}{1-\alpha -\lambda} \right] <0 
\] for all $\alpha  \in [0, 1-\lambda]$.
The maximum of the function is then given by 
\[
f(1- \frac{4}{3}\lambda, \lambda) = 2 -  \lambda \log3.
\]
Using the union bound, the sum in \eqref{eq:ADRScomp_bound} can be bounded by $
4c n 2^{n( 2 - \lambda \log 3 )}
$, and the ratio of the sum to  $N= 2^n$, denoted by $\gamma_C$, is bounded from above as
$
\gamma_C \leq 4c n 2^{n( 1 -\lambda \log 3 )}.
$
Since the exponent $n( 1 - \lambda \log 3)$ goes to negative infinity as $n$ grows when $\lambda > \lambda^{\dagger} = 1/\log3 \approx 0.631$, the ratio $\gamma_C$ vanishes exponentially in $n$ when  $\lambda >\lambda^{\dagger}$. 
The proposition follows by noting that the SC decoding complexity for the code based on $G_2^{\otimes n}$ is $N\log N$.
\endproof

\vspace{2mm}

% \proof
\textit{Proof of Proposition\,\ref{Prop:BMSrateloss}:}
Similar to the proof of Proposition\,\ref{prop:ADRScompl}, the number of additional channels due to the A-DRS scheme modification is given by 
\begin{align}\label{eq:ADRSrate_bound}
    \sum_{l=1}^{n-n_{lub}+1} & \abs{\mathset{\mathbf{s} \in \mathset{-,+}^n: \tau(\mathbf{s} ) > l , s_l =- } }  2^l \nonumber \\
    = &\sum_{l=1}^{n-n_{lub}+1} \sum_{k=l+1}^{n-n_{lub}+1} \abs{\mathset{\mathbf{s} \in \mathset{-,+}^n: \tau(\mathbf{s} ) =k , s_l =-  } }  2^l \nonumber \\
    \leq & \sum_{k=1}^{n-n_{lub}+1} \binom{n-k+1}{ n_{lub}} 2^{k-2}  \sum_{l=1}^{k-1}   2^{l}\nonumber  \\
    % = &\sum_{k=1}^{n-n_{lub}+1} \binom{n-k+1}{ n_{lub}}  2^{k-2}  (2^{k}-2) \nonumber \\
    \leq &\sum_{k=0}^{n-n_{lub}} \binom{n-k}{ n_{lub}}  2^{2k}.
\end{align}
By the argument in the proof of Proposition\,\ref{prop:ADRScompl}, the sum in \eqref{eq:ADRSrate_bound} can be upper bounded by 
$ n 2^{n( 2 - \lambda \log 3 )}$, and the ratio of the sum to  $N= 2^n$, denoted by $\gamma$, is bounded from above as
$ \gamma  \leq  n 2^{n( 1 -\lambda \log 3 )}.$
Since the exponent $n( 1 - \lambda \log 3)$ goes to negative infinity as $n$ grows when $\lambda > \lambda^{\dagger} = 1/\log_2 3 \approx 0.631$, the ratio $\gamma $ vanishes exponentially in $n$ when  $\lambda >\lambda^{\dagger}$. 
\endproof

\vspace{3mm}

\subsection{Proofs for Section \ref{sec:L>2}}

\textit{Proof of Lemma \ref{lemma:SparsityOrderGeneralL}:}
Note that 
\begin{align*}
    w_{GM}(n,G_l) &=  [(w_1\times w_2 \times \ldots \times w_l)^\frac{1}{l}]^n = GM^n = (l^n)^{\log_l{GM} } = N^{\log_l{GM} },
\end{align*}
where $GM= GM(w_1, w_2 , \ldots , w_l)$. Using \eqref{eq:logN'formula},  we have $[\log{N'}]^{\frac{1}{E(G_l)}} \leq   N \leq [\log{N'}]^{\frac{1}{{(1-\delta)E(G_l)}}}$ for all sufficiently large $n$. 
Therefore, 
\begin{align*}
&[\log{N'}]^{  \frac{\log_l{GM}}{E(G_l)} }   \leq
w_{GM}(n,G_l) \leq 
[\log{N'}]^{\frac{\log_l{GM}}{{(1-\delta)E(G_l)}}  }.
\end{align*}
Using the definition of the sparsity order of the geometric mean column weight, $\lambda_{GM}$, can be bounded as:
\begin{equation}\label{eq:lambdaMCbound}
 \frac{\log_l{GM}} {E(G_l)}  \leq 
\lambda_{GM}(n,G_l) \leq 
\frac{\log_l{GM}}{{(1-\delta)E(G_l)}}.
\end{equation}
Writing  $GM$ and  $E(G_l)$ in terms of   $w_i$'s and $D_i$'s, we have
\[
 \frac{\log_l{GM}} {E(G_l)} 
 = \frac{ \sum_{i=1}^l \log_l w_i } {\sum_{i=1}^l \log_l D_i }. \label{eq:orderMC_bd}
\]
So equation \eqref{eq:lambdaMCbound} can be written as:
 \[
  \frac{ \sum_{i=1}^l \log_l w_i } {\sum_{i=1}^l \log_l D_i } \leq
\lambda_{GM}(n,G_l) \leq 
\frac{1}{1-\delta}
 \frac{ \sum_{i=1}^l \log_l w_i } {\sum_{i=1}^l \log_l D_i }.
\]
\endproof

\textit{Proof of Theorem  \ref{Thm:lambdaMC_lgeneral}:}
Let 
\[
G_l=
\left[
\begin{array}{c|c}
1 & \textbf{0}_{1,l-1} \\
\hline
\textbf{1}_{l-1,1} & I_{l-1}
\end{array}
\right]
\] be an $l\times l$ matrix. 

The geometric mean of column weights, $GM(w_1, \ldots, w_l )$, is $ GM(l, 1, \ldots, 1)= l^\frac{1}{l}$. The partial distances $\{D_i\}_{i=1}^l$ of $G_l$ are
$$
D_i=
\begin{cases}
2, \mbox{ for } i\geq 2\\
1, \mbox{ for } i =1.
\end{cases}
$$

It can be observed that $ \lim_{ l \to \infty }\frac{\log_l{GM}} {E(G_l)} = 0 $. Hence, for any fixed $\delta $, there is some $l^*$ such that $\frac{1}{(1-\delta)}\frac{\log_{l^*}{GM}}{{E(G_{l^*})}} <r $. 
By Lemma \ref{lemma:SparsityOrderGeneralL}, we have $\lambda_{GM}(n,G_l) < r$, for sufficiently large $n$.
\endproof

%% References:
%% We recommend the usage of BibTeX:
%%
\bibliographystyle{IEEEtran}
\bibliography{IEEEabrv}
%%
%% where we here have assume the existence of the files
%% definitions.bib and bibliofile.bib.
%% BibTeX documentation can be obtained at:
%% http://www.ctan.org/tex-archive/biblio/bibtex/contrib/doc/
%%%%%%

%% Or you use manual references (pay attention to consistency and the
%% formatting style!):
\end{document}